\documentclass[reprint,
superscriptaddress,
 amsmath,amssymb,
 aps,
pra,
floatfix,
]{revtex4-1}
\usepackage{float}
\usepackage{physics}
\usepackage{graphicx}
\graphicspath{ {images/}, {figs/} }
\usepackage{dcolumn}
\usepackage{bm}
\usepackage{hyperref}
\usepackage[mathlines]{lineno}
\usepackage{amssymb}
\usepackage[toc,page]{appendix}
\usepackage{amsthm}
\newtheorem{definition}{Definition}

\newtheorem{theorem}{Theorem}[section]
\newtheorem{lemma}[theorem]{Lemma}
\newtheorem{proposition}{Proposition}
\newtheorem{corollary}{Corollary}[section]

\usepackage{xcolor}
\hypersetup{
	colorlinks,
	linkcolor={red!50!black},
	citecolor={green!50!black},
	urlcolor={blue!80!black}
}
\usepackage[caption=false]{subfig}
\makeatletter
\newcommand{\colorcaption}[2][]{%
  \begingroup%
  \renewcommand{\@caption@fignum@sep}{ (color online). }%
  \caption[#1]{#2}%
  \endgroup%
}
\makeatother

\makeatletter
\renewcommand{\maketag@@@}[1]{\hbox{\m@th\normalsize\normalfont#1}}%
\makeatother

\usepackage{graphicx}

\begin{document}

\title{Analytic exposition of the graviton modes in fractional quantum Hall effects\\ and its physical implications}

\author{Wang Yuzhu}
\email{yuzhu001@e.ntu.edu.sg}
\affiliation{School of Physical and Mathematical Sciences, Nanyang Technological University, 639798 Singapore}

\author{Yang Bo}
\email{yang.bo@ntu.edu.sg}
\affiliation{School of Physical and Mathematical Sciences, Nanyang Technological University, 639798 Singapore}
\affiliation{Institute of High Performance Computing, A*STAR, 138632 Singapore}

\date{\today}

\begin{abstract}
Neutral excitations in a fractional quantum Hall droplet define the incompressibility gap of the topological phase. In this work, we derived a set of analytical results for the energy gap of the graviton modes with two-body and three-body Hamiltonians in both the long-wavelength and thermodynamic limit. These allow us to construct model Hamiltonians for the graviton modes in different FQH phases, and to elucidate a hierarchical structure of conformal Hilbert spaces (null spaces of model Hamiltonians) with respect to the graviton modes and their corresponding ground states. Using the analytical tools developed, we perform numerical analysis with a particular focus on the Laughlin $\nu=1/5$ and the Gaffnian $\nu=2/5$ phases. Our calculation shows that for gapped phases, low-lying neutral excitations can undergo a ``phase transition" even when the ground state is invariant. We discuss about the compressibility of the Gaffnian phase, the possibility of multiple graviton modes, and the transition from the graviton modes to the ``hollow-core" modes, as well as their experimental consequences.
\end{abstract}

\maketitle

\section{\label{sec:intro}Introduction}
The cryogenic two-dimensional electrons in a strong magnetic field can form an incompressible quantum fluid from strong interactions, leading to a large number of fractional quantum Hall (FQH) states describing a zoology of strongly correlated topological phases\cite{klitzing1980new,tsui1982two}. The characteristic fractional plateau of the Hall resistivity and the quantized thermal Hall conductance 
in these phases have been observed in experiments, revealing the topological nature of the quantum fluids\cite{PhysRevLett.101.246806, banerjee2018observation, banerjee2017observed}. The charged excitations in the FQH systems are predicted to carry fractional charge, with anyonic and even non-abelian statistics\cite{laughlin1983anomalous, clark1988experimental, de1998direct, saminadayar1997observation, reznikov1999observation, camino20073, nakamura2020direct, mcclure2012fabry, goldman1995resonant, radu2008quasi, lin2012measurements, venkatachalam2011local, PhysRevB.77.205310, moore1991nonabelions}. Under the right conditions, these topological properties are expected to be invariant against local disturbance, making them desirable for the robust manipulation of quantum information\cite{freedman2003topological, PhysRevLett.94.166802, PhysRevA.102.022607}. 

Strictly speaking, topological systems are ideal systems where all energy scales in the system are sent to either zero or infinity. In realistic experiments or materials, all energy scales are finite, thus the dynamical aspects involving low-lying excitations cannot be ignored. In FQH systems, the low-lying excitations not only determine the quantization of certain transport properties in the thermodynamic limit, they can even form topological quantum fluids of their own\cite{arovas1984fractional, blok1990effective, read1996quasiholes, baraban2009numerical}. The robustness of different topological indices (e.g. the Hall conductivity, the topological shift, and the central charge from the quasihole counting, etc.) can depend on different energy scales, and the interplay between the corresponding low-lying excitations leads to rich physics even within the same topological phase. 

Understanding the low-lying excitations in strongly correlated systems is, however, a difficult task due to the lack of theoretical tools available. One can numerically study the energy spectrum of only small systems given the exponential increase of the Hilbert space with the system size, and the extrapolation to the thermodynamic limit is oftentimes unreliable. Analytically, perturbative calculations are not viable for FQH systems, because the kinetic energy is quenched and we are left with a purely interacting Hamiltonian. A popular approach is to construct model wavefunctions of these excitations either from the Jack polynomial formalism and its generalisation, or using the composite fermion theory\cite{laughlin1981quantized, laughlin1983quantized, jain1989composite, PhysRevLett.100.246802, PhysRevLett.101.246806, bernevig2008generalized, Coimbatore_Balram_2021, PhysRevB.40.8079, PhysRevB.98.035127, PhysRevResearch.2.032035}. These model states are indispensable for finite size numerical analysis, as well as offering insights into the universal nature of these excitations. The construction of these states also hints at the tantalising possibility of analytical treatment for very large system sizes, but so far the ability to rigorously compute physically relevant quantities in the thermodynamic limit is mostly lacking.

In this paper, we focus on the neutral excitations of the FQH fluids, and derive a number of analytical results that are valid in the long-wavelength and thermodynamic limit. Such excitations define the incompressibility gap of the topological phase, and they can also be responsible for quantum phase transitions within the same topological order\cite{girvin1986magneto, girvin1987off, yang2020microscopic, bid2010observation, inoue2014proliferation}. In the long-wavelength limit, neutral excitations have a quadrupolar structure and can be understood as a ``spin-2 graviton", which are closely related to the geometric deformation of the incompressible ground state\cite{Haldane2011, Yang2013, Golkar2016, Luo2016}. We will thus denote excitations in this limit as quadrupole excitations or graviton modes. For short-range interactions, numerical analysis shows the energy of the graviton is large as compared to the charge gap given by the roton minimum\cite{PhysRevLett.108.256807}. However, it also seems possible to lower the graviton excitation energy by tuning the interaction between electrons, while maintaining the charge gap\cite{kang2017neutral, jolicoeur2014absence, yang2019effective, yang2021gaffnian}. The analytical results we derive in this work help us construct model Hamiltonians not only for the graviton modes, but also for capturing such transitions related to the graviton gap.

It is also important to note that the neutral excitations can be experimentally measured with inelastic photon scattering. This is especially true for graviton modes, since the momentum transfer of the photons is small\cite{pinczuk1993observation, pinczuk1994inelastic, pinczuk1998light, wurstbauer2013resonant, Liou2019}. Acoustic crystalline wave is also predicted to act like gravitational wave, which can interact with the graviton modes as a probe\cite{Yang2016}. Furthermore, an optimal-control-based variational quantum algorithm has been designed for realizing the graviton mode in quantum computers\cite{kirmani2021realizing}. There has also been much interest in the graviton mode recently due to its spin structure, allowing it to couple selectively to circularly polarized light, making them useful for experimentally distinguishing different topological phases\cite{Golkar2016}. In addition, it has been recently suggested that the coupling of the incompressible ground state to the graviton mode from geometric deformation can be responsible for the quench dynamics in FQH\cite{PhysRevLett.126.076604}. In the context of these experimental proposals and numerical results, Dirac composite fermion theory conjectures that certain FQH phases may have more than one graviton mode\cite{Haldane2021, Nguyen2021}. Thus both from the theoretical and experimental perspectives, analytic results of long-wavelength neutral excitations can help us understand the fundamental nature of the geometric aspects of the FQH topological phases.

The organization of this paper will be as follows. The characteristic matrix formalism for calculating the energy of the graviton modes with two-body pseudopotentials is reviewed and a novel hierarchy structure of the density modes of the Laughlin states is rigorously proved with this method in Sec.~\ref{sec_2bdy}. This allows us to introduce the model Hamiltonians for the graviton modes in several FQH phases. Sec.~\ref{sec_3bdy} provides a more general formalism for the graviton modes with the three-body interactions, which can be used to study the graviton modes of non-abelian FQH states such as the Moore-Read state and the Gaffnian state. In Sec.~\ref{sec_numerics} we perform numerical analysis of the low-lying excitations of the Laughlin $\nu = 1/5$ state (Laughlin-$1/5$ state for short) and the Gaffnian $\nu = 2/5$ state (Gaffnian state for short), with the theoretical tools we have developed. The experimental implications of the graviton modes in the Laughlin-$1/5$ phase are discussed when the model interaction is tuned, and we also discuss the gaplessness of the Gaffnian phase. In Sec.~\ref{sec_con} we summarise our results and discuss about the future outlooks.

\begin{table}[]
\renewcommand\arraystretch{1.5}
\centering
\begin{tabular}{|c|c|}%
\hline
$\hat R_i^a$              & Guiding center operator  \\[3pt] \hline
$\hat\rho_{\bm q}$        & Guiding center density operator  \\ \hline
$\delta \hat\rho_{\bm q}$ & Regularized guiding center density operator  \\ \hline
$V_{\bm q}$               & Effective potential  \\[3pt] \hline
$\tilde S_{\bm q}$        & Static structure factor  \\[3pt] \hline
$s_{\bm q}$               & Reduced structure factor  \\ \hline
$L_{k}\left(q^2\right)$    &  Laguerre polynomials  \\[3pt] \hline
$c^m$               &  Expansion coefficients of effective potential\\ \hline
$d^n$               &  Expansion coefficients of reduced structure factor\\ \hline
$\left| \psi_{\boldsymbol{q}}\right. \rangle$               &  SMA Model wavefunction\\ \hline
$\Gamma^{\text{2bdy}}_{m n}$              &  Two-body characteristic matrix\\ \hline
\end{tabular}
\caption{Definition of various symbols used in the text.}
\label{table2}
\end{table}

\section{Two-body interactions}\label{sec_2bdy}
Let us start with the dynamics of the long-wavelength neutral excitations for FQH states with two-body interactions . This formalism was first developed in Ref.\cite{yang2020microscopic}, and we will summarise the results here and set the notations. This will be followed by a number of new results that are relevant to model two-body interactions, with rigorous statements on the nature of the neutral excitations that are valid in the thermodynamic limit.

The most general Hamiltonian for two-body interaction is:
\begin{eqnarray}
\hat H_{\text{2bdy}}&=&\int \frac{d^{2} \boldsymbol{q}}{(2 \pi)^2} V_{\bm q}\sum_{i\neq j}e^{iq_a\left(\hat R_i^a-\hat R_j^a\right)}\nonumber\\
&=&\int \frac{d^{2} \boldsymbol{q}}{(2 \pi)^2} V_{\bm q}\hat\rho_{\bm q}\hat\rho_{-\bm q}-N_e\int \frac{d^{2} \boldsymbol{q}}{(2 \pi)^2} V_{\bm q}
\label{h0}
\end{eqnarray}
Here $N_e$ is the number of electrons, $\hat H_{\text{2bdy}}$ is a gapped two-body interaction with the guiding center density operator $\hat\rho_{\bm q}=\sum_ie^{i q_a\hat R_i^a}$ and $\bm q=(q_x,q_y)$, where $i$ is the electron index, and $\hat R_i^a$ is the guiding center coordinates satisfying the commutation relation $[\hat R_i^a,\hat R_j^b]=-i\epsilon^{ab}\delta_{ij} l_B^2$. The length scale in the problem is the magnetic length $l_B=\sqrt{1/eB}$, with the perpendicular magnetic field $B$. Throughout this work we assume translational invariance, so the ground state $| \psi_0\rangle $ of $\hat H_{\text{2bdy}}$ is both rotationally and translational invariant, with $\hat H_{\text{2bdy}}|\psi_0\rangle = E_0|\psi_0\rangle$.

The static structure factor for the unperturbed ground state is given by:
\begin{eqnarray}
\tilde S_{\bm q}=\frac{1}{N_e}\langle\psi_{0}|\hat\rho_{\bm q}\hat\rho_{-\bm q}|\psi_0\rangle=1+s_{\bm q}
\end{eqnarray}
where the reduced structure factor is its own Fourier transform, given by:
\begin{eqnarray}
s_{\bm q}&=&\frac{1}{N_e}\langle\psi_{0}|\sum_{i\neq j}e^{iq_a\left(\hat R_i^a-\hat R_j^a\right)}|\psi_0\rangle=-\int\frac{d^2q'}{2\pi}e^{i\bm q\times\bm q'}s_{\bm q'}\qquad
\label{ft}
\end{eqnarray}
Since the Laguerre polynomials $L_k\left(x\right)$ are eigenfunctions of the two-dimensional Fourier transform, we can thus expand $s_{\bm q}$ in the basis of the Laguerre-Gaussian series as follows:
\begin{eqnarray}
s_{\bm q}=d^k L_{k}\left(q^2\right) e^{-\frac{1}{2}q^2}\label{expand}
\end{eqnarray}
where Einstein summation rule is assumed throughout this work for repeated indices. Note that $d^k$ is only non-zero for odd $k$ due to the fermion statistics, and Eq.(\ref{ft}) and Eq.(\ref{expand}) is true for generic rotationally invariant many-body states in a single Landau level. The regularised structure factor can be defined using the regularised guiding center density operator:
\begin{equation}
\delta \hat{\rho}_{\boldsymbol{q}}=\hat{\rho}_{\boldsymbol{q}}-\left\langle\psi_{0}\left|\hat{\rho}_{\boldsymbol{q}}\right| \psi_{0}\right\rangle=\hat{\rho}_{\boldsymbol{q}}-\frac{N_{e}}{2 \pi q}\delta_{q,0}
\end{equation}
where $q = |\bm q|$. so we have the following:
\begin{eqnarray}
S_{\bm q}&=&\frac{1}{N_e}\langle\delta\hat\rho_{\bm q}\delta\hat\rho_{\bm q}\rangle_0=\tilde S_{\bm q}-\frac{1}{N_e}\langle\hat\rho_{\bm q}\rangle_0\langle\hat\rho_{-\bm q}\rangle_0 \label{sf0}\\
&=&1+\sum_{k=0}^\infty d^{\left(0\right)}_{2k+1} L_{2k+1}\left(q^2\right)e^{-\frac{1}{2}q^2}-\frac{1}{N_e}\langle\hat\rho_{\bm q}\rangle_0\langle\hat\rho_{-\bm q}\rangle_0\quad\quad
\label{sf}
\end{eqnarray}
and we define $\langle\hat O\rangle_0$ to be the expectation value of $\hat O$ with respect to the ground state.

It is natural to expand the effective potential $V_{\boldsymbol{q}}$ using the same orthogonal basis, which gives:
\begin{equation}
V_{\bm{q}}=c^{m} L_{m}(q^2) e^{-\frac{q^2}{2}}
\end{equation}
Due to the orthogonality between the Laguerre polynomials, the ground state energy has the following simple expression:
\begin{equation}
E_0 = c^m d^n\delta_{m,n}
\label{gsenergy}
\end{equation}
where $\delta_{m,n}$ is the Kronecker delta. Once the expansion coefficients are given, one can easily calculate the ground state energy. In general $d^n$ still needs to be computed numerically as the variational energy of the ground state with respect to the two-body Haldane pseudopotential $\hat V_n^{\text{2bdy}}=\int d^2q L_{n}(q^2) e^{-\frac{q^2}{2}}\sum_{i\neq j}e^{i q_a\left(\hat R_i^a-\hat R_j^a\right)}$. However for model wavefunctions, certain coefficients of $d^n$ are known exactly, since they are exact zero energy states of their respective model Hamiltonians.

\subsection{Energetics of the graviton modes}
It is well-known that the model wavefunctions for the graviton modes in the long-wavelength limit can be exactly constructed using the single mode approximation (SMA), which is defined as follows\cite{girvin1986magneto}:
\begin{equation}
\left| \psi_{\boldsymbol{q}}\right. \rangle = \delta \hat{\rho}_{\boldsymbol{q}} \left| \psi_{\boldsymbol{0}}\right. \rangle
\label{SMAstate}
\end{equation}
This family of wavefunctions is orthogonal to the ground state and at the same time retains some of the intrinsic correlation properties of the ground state. Our task is to calculate the variational energy of the graviton mode given by $\lim_{\bm q\rightarrow 0}|\psi_{\bm q}\rangle$:
\begin{equation}
\begin{aligned}
\delta E_{\boldsymbol{q}\rightarrow 0}&=\lim_{\bm q\rightarrow 0}\frac{\langle\psi_{\boldsymbol{q}}|\hat{H}_{\text{2bdy}} \left| \psi_{\boldsymbol{q}}\right. \rangle}{\left. \langle\psi_{\boldsymbol{q}} \right| \psi_{\boldsymbol{q}}\rangle}-E_{0}\\
&=\lim_{\bm q\rightarrow 0}\frac{\left\langle\psi_{0}\left| \left[\delta \hat{\rho}_{-\boldsymbol{q}},\left[\hat{H}_{\text{2bdy}}, \delta \hat{\rho}_{\boldsymbol{q}}\right]\right] \right| \psi_{0}\right\rangle}{2 S_{\boldsymbol{q}}}\\
&=\lim_{\bm q\rightarrow 0} \sum_{i \ne j}\int \frac{d^{2} \boldsymbol{q}^\prime}{(2 \pi)^2} V_{\boldsymbol{q}^\prime}\\ 
&\quad \times \frac{\left\langle\psi_{0}\left| \left[\delta \hat{\rho}_{-\boldsymbol{q}},\left[\hat{\rho}^i_{\boldsymbol{q}^\prime} \hat{\rho}^j_{\boldsymbol{-q}^\prime}, \delta \hat{\rho}_{\boldsymbol{q}}\right]\right] \right| \psi_{0}\right\rangle}{2 S_{\boldsymbol{q}}}
\end{aligned}
\end{equation}

To calculate the double commutator, recall the GMP algebra for the density operators\cite{girvin1986magneto}:
\begin{equation}
\left[\delta \hat{\rho}_{\boldsymbol{q}_{1}}, \delta \hat{\rho}_{\boldsymbol{q}_{2}}\right]  = 2 i \sin \frac{\boldsymbol{q}_{1} \times \boldsymbol{q}_{2}}{2} \delta \hat{\rho}_{\boldsymbol{q}_{1}+\boldsymbol{q}_{2}}
\end{equation}
This allows us to derive the final result, as shown in Ref.\cite{yang2020microscopic}:
\begin{equation}
\delta E_{\boldsymbol{q} \rightarrow 0}= \Gamma^{\text{2bdy}}_{m n} c^{m} d^{n} + O\left(q^{2}\right)
\label{e2b}
\end{equation}
where $m, n$ are both positive odd numbers because of the fermionic nature of electrons, and the characteristic matrix has an exact form:
\begin{equation}
\begin{aligned}
\Gamma^{\text{2bdy}}_{m n}
=&\frac{(-1)^{m}}{2^8 \eta \pi} \times [2(m^2+m+1)\delta_{m,n} \\
& -(m+1)(m+2)\delta_{m,n-2}-m(m-1) \delta_{m,n+2}]
\label{Gamma1}
\end{aligned}
\end{equation}
It is worth noticing that in this characteristic matrix, the expansion coefficients of the interaction and the structure factor are symmetric. Furthermore, only the ``nearest neighbors" ($n=m \pm 2$) are involved for any given $m$, which implies that the interplay between the components of different orders in the expansions of the interaction and the state is short-ranged. 

An immediate result is the variational energy of the graviton modes for the model Laughlin states at filling factor $\nu<1/3$:  for example, the variational energy of the graviton modes of the Laughlin-$1/5$ states (Laughlin-$1/5$ graviton mod for short) with respect to a generic two-body interaction can be written as ($d^1 = d^3 = 0$):
\begin{equation}
\begin{aligned}
\delta E_{\boldsymbol{q} \rightarrow 0}&=\sum_{m>3}\left(\Gamma^{\text{2bdy}}_{mm} c^m d^m+\Gamma^{\text{2bdy}}_{m\pm 2,m} c^{m\pm2} d^m\right)
\end{aligned}
\end{equation}
Thus with respect to the model Hamiltonian $V^{\text{2bdy}}_1$ (only $c^1 > 0$) the variational energy is exactly zero: the Laughlin-$1/5$ graviton mode is in the null space of $\hat V_1^{\text{2bdy}}$. 

\begin{figure}
\centering
\includegraphics[scale=0.28]{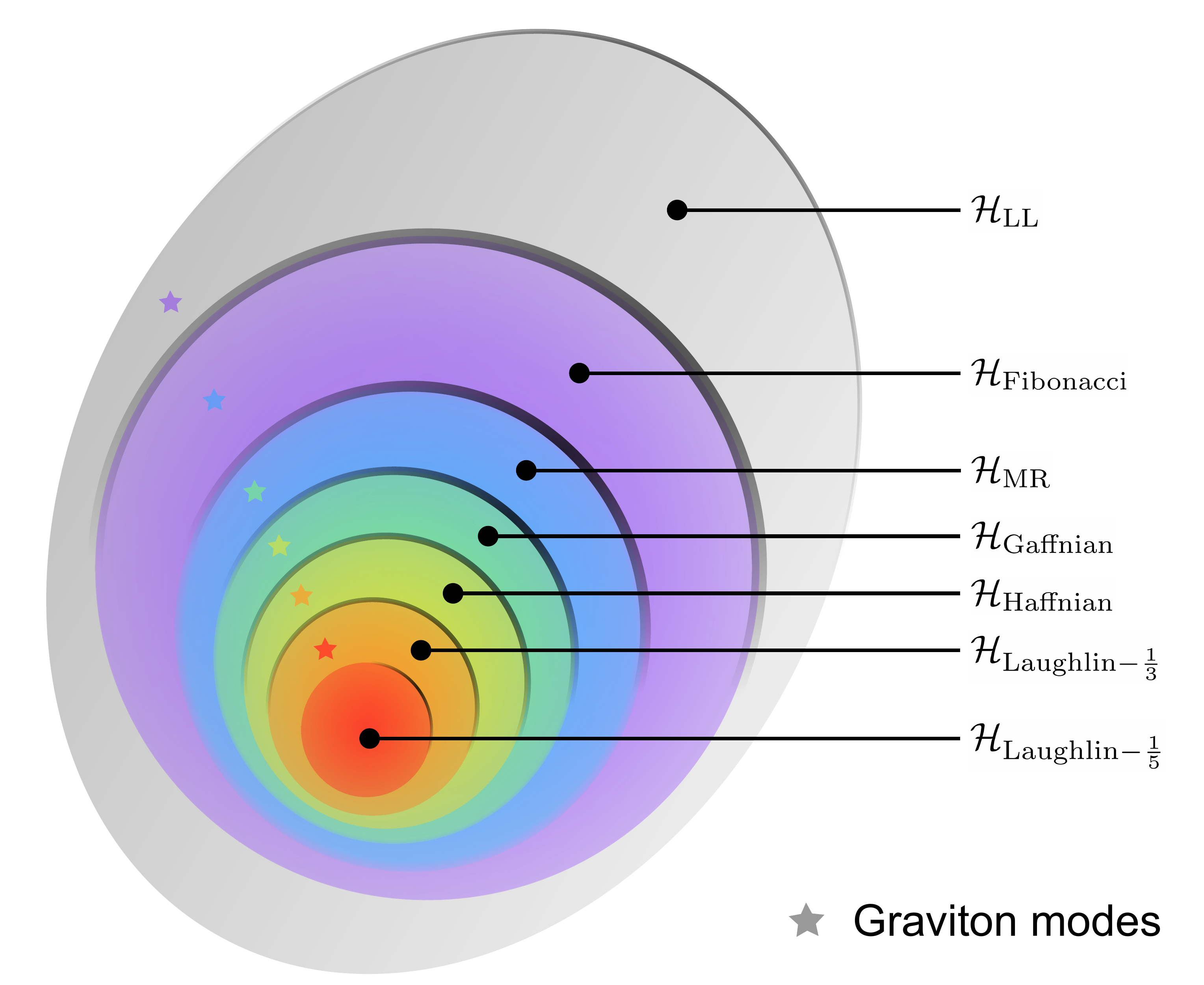}
\caption{\textbf{Hierarchy of the null spaces of different FQH states.} The null space of different model Hamiltonians can be organized as a hierarchical structure in the full Hilbert space. Meanwhile the density modes or more precisely, the graviton modes (denoted by the stars with the corresponding colors) constructed from model ground states live in the next null space. For example, the graviton modes of Laughlin-$1/3$ phase (orange star) live in the Haffnian null space (green circle), which contains the ground state and all the quasihole states of the Haffnian model Hamiltonian. It is efficient to verify this structure with the characteristic tensor formalism proposed in this work.}
\label{nullspace}
\end{figure}

\subsection{Model Hamiltonians for the graviton modes}

From both the theoretical and experimental point of view, it is useful to find the model Hamiltonians of which the graviton modes are the exact eigenstates. This will not only allow us to construct minimal models for phase transitions involving the softening of the graviton modes, but also help implement more systematic tuning of the realistic interactions in experiments for the observation of such transitions. For many topological phases, we can write down the root configurations of the graviton modes explicitly in the second quantized language. The model Hamiltonians can be identified accordingly using the numerically established local exclusion condition (LEC) formalism\cite{PhysRevB.100.241302, yang2021elementary}. We illustrate this useful procedure for the graviton modes of the Laughlin phase at filling factor $\nu=1/3$ and $\nu=1/5$. This will be followed by a rigorous proof for all the Laughlin states with $\nu < 1/3$. The proof for the $\nu = 1/3$ case involves three-body interactions, and will be given in the next section.

Let us start with a brief introduction of the Jack polynomial formalism on the spherical geometry. The pseudo-spin structure in the quantized Landau levels enables us to define second-quantized wavefunctions in an occupation-number-like basis for the many-body wavefunctions. The magnetic field can be introduced by putting a Dirac magnetic monopole at the center of the sphere with a total of $2S$ magnetic fluxes, rendering a spinor structure of the single particle the wavefunction  with total spin $S+N$, where $N$ is the Laudau level (LL) index\cite{PhysRevLett.51.605}. Without loss of generality we use the lowest LL (LLL) with $N=0$, so the total number of single particle orbitals in the LLL is $N_o=2S+1$. We can thus express a many-body state with a string of $N_o$ binary numbers, corresponding to the single particle orbitals sequentially from the north pole to the south pole. We use the integer $1$ to denote an occupied orbital, and $0$ for an empty orbital\cite{PhysRevLett.100.246802, PhysRevLett.101.246806, bernevig2008generalized}. For example, if $S=3$ and we have two electrons around the north pole of the sphere, this state can be denoted as $|1100000 \rangle$, with a total number of seven orbitals.

\begin{table}[]
\renewcommand\arraystretch{1.5}
\centering
\begin{tabular}{|c|c|}%
\hline
$\hat{\rho}^i_{\boldsymbol{q_1}}$ &\begin{tabular}[c]{@{}c@{}} Guiding center density operator\\ of the $i$-th electron \\ $\hat{\rho}^i_{\boldsymbol{q_1}} = e^{i q_{1a} \hat{R}^{a}_{i} }$\end{tabular}  \\[3pt] \hline
$\bar{S}_{\boldsymbol{q}_1,\boldsymbol{q}_2}$              & Reduced three-body structure factor  \\[3pt] \hline
$V_{\bm q_1 \bm q_2}$               & Three-body effective potential  \\[3pt] \hline
$\bm \tilde q_i$               & Momentum components in Jacobi coordinates  \\ \hline
$\tilde Q_i$               & Square of  $\bm \tilde q_i$ \\ \hline
$L^{(\alpha)}_{k}\left(q^2\right)$    & Generalized Laguerre polynomials  \\[3pt] \hline
$c^{m_1 m_2}$               &  \begin{tabular}[c]{@{}c@{}}Expansion coefficients of three-body\\ effective potential\end{tabular}\\ \hline
$\bar{d}^{n_1 n_2}$               &  \begin{tabular}[c]{@{}c@{}}Expansion coefficients of reduced\\ three-body structure factor\end{tabular}\\ \hline
$\tilde{\Gamma}_{m_1 m_2 n_1 n_2}^{\text{3bdy}}$              &  Three-body characteristic tensor\\ \hline
$\tilde{\Gamma}^{0/+/-}_{m_1 m_2 n_1 n_2}$              & \begin{tabular}[c]{@{}c@{}} Diagonal/Off-diagonal parts of\\ the three-body characteristic tensor\end{tabular}\\ \hline
\end{tabular}
\caption{Definition of various symbols used in the three-body calculations.}
\label{table3}
\end{table}

On the disk geometry, the occupation basis also corresponds to the first-quantized wavefunctions with the symmetric gauge (we ignore the unimportant Gaussian factor). Each digit in the occupation basis from left to right corresponds to the orbital from the origin ($z_k^0$) to the edge ($z_k^{N_o -1}$), where $z_k=x_k+iy_k$ is the holomorphic coordinate of the $k^{\text{th}}$ electron. The many-body wavefunctions for the FQH states are linear combinations of such monomials. For example, the monomial $|1101 \rangle$ in the first quantized form is given by:
\begin{equation}
\begin{aligned}
|1101 \rangle &\sim (z_1 - z_2) z^3_3 + (z_3 - z_1) z^3_2 + (z_2 - z_3) z^3_1 
\end{aligned}
\end{equation}

The model wavefunctions for the FQH states on the sphere in many cases are Jack polynomials (or Jacks), which is a family of symmetric polynomials characterised by the so-called generalized Pauli principle. One important characteristics of the Jack polynomial states is the existence of a root configuration, with all of the occupation basis of the state ``squeezed" from the root configuration\cite{PhysRevLett.100.246802,PhysRevLett.101.246806, bernevig2008generalized}. For instance, if one considers the Laughlin-$\frac{1}{3}$ state with 3 electrons, the wavefunction $(z_1 -z_2)^3(z_1 -z_3)^3(z_2 -z_3)^3$ is a Jack polynomial denoted by $J^{\alpha=-2}_{|1001001 \rangle}$. Here $1001001$ is the root configuration, and $\alpha=-2$ in the superscript is derived from the admission rule of the root configuration. All coefficients of the monomials in the Jack polynomial are determined by $\alpha$, and these monomials are ``squeezed" from the root configuration. We denote two monimials $m_1, m_2$ and that $m_2$ is squeezed from $m_1$ by $m_1 \succ m_2$. That implies $m_2$ is obtained from $m_1$ by repeatedly moving two electrons in the binary string towards each other, without changing the total angular momentum of the monomial. Explicitly for the Laughlin-$1/3$ state with three electrons we have the following:
\begin{equation}
\begin{aligned}
J^{-2}_{|1001001 \rangle} \sim& (z_1-z_2)^3(z_1-z_3)^3(z_2-z_3)^3\\
= & |1001001\rangle -3 |0110001\rangle -3 |1000110\rangle \\
& +6 |0101010\rangle -15 |0011100\rangle
\end{aligned}
\end{equation}
which consists of only the monomials squeezed from the root configuration $|1001001\rangle$, and the monomials are::
\begin{equation}
\begin{aligned}
|1001001\rangle=&  (z_1^3 - z_2^3) z_3^6 +  (z_3^3 - z_1^3) z_2^6 + (z_2^3 - z_3^3) z_1^6 \\ 
|0110001\rangle=&  (z_2^1 z_1^2 - z_1^1 z_2^2) z_3^6 + (z_1^1 z_3^2 - z_3^1 z_1^2) z_2^6\\ 
& + (z_3^1 z_2^2 - z^1_2 z_3^2) z_1^6\\
|1000110\rangle=&  (z_1^4 - z_2^4) z_3^5 +  (z_3^4 - z_1^4) z_2^5 + (z_2^4 - z_3^4) z_1^5 \\ 
|0101010\rangle=&  (z_2^1 z_1^3 - z_1^1 z_2^3) z_3^5 + (z_1^1 z_3^3 - z_3^1 z_1^3) z_2^5\\ 
& + (z_3^1 z_2^3 - z^1_2 z_3^3) z_1^5\\
|0011100\rangle=&  (z_2^2 z_1^3 - z_1^2 z_2^3 ) z_3^4 + (z_1^2 z_3^3 - z_3^2 z_1^3) z_2^4\\ 
& + (z_3^2 z_2^3 - z_2^2 z_3^3) z_1^4\\
\end{aligned}
\end{equation}

Let us now move on to the model wavefunction of the graviton mode, as constructed in Ref.\cite{PhysRevLett.108.256807}. In contrast to the FQH ground states, these wavefunctions are not Jack polynomials. Nevertheless, these wavefunctions have rich algebraic structure with a root configuration. Similar to the Jack polynomials, only monomials squeezed from the root configuration have non-zero coefficients in the basis expansion. We will also show later on they can be linear combinations of Jack polynomials of different admission rules in some cases. The root configurations of the graviton modes for the Laughlin-$1/3$ state and $1/5$ state for comparison are shown as follows:
\begin{eqnarray}
&1100001001001001001 \cdots \quad &(\text{Laughlin-}1/3)\label{LEC1}\\
&100100000010000100001\cdots \quad &(\text{Laughlin-}1/5)\label{LEC2}
\end{eqnarray}

Since all the occupation bases are squeezed from the root configurations above, we can immediately apply the local exclusion condition (LEC)\cite{PhysRevB.100.241302, yang2021elementary}. There are three parameters required to define an LEC, denoted by $\hat{n} = \{n, n_e, n_h\}$, giving the constraint that there can be no more than $n_e$ electrons or $n_h$ holes in a circular droplet containing $n$ fluxes anywhere in the quantum fluid. For a spherical geometry, one can simply look at the droplet at the north pole for the \emph{highest weight states}. Applying LEC to the graviton mode of Laughlin-$1/3$ state immediately leads to the conclusion that it is a ``Haffnian quasihole state" (thus a zero energy state of the Haffnian model Hamiltonian; more details will be given in the next section). This is because the LEC of the Haffnian state is given by $\hat{n} = \{4, 2, 4\}$ and the droplet at the north pole of the Laughlin-$1/3$ graviton mode does obey this condition as Eq.\ref{LEC1} shows. We would like to emphasize that the admissible rules for Jack polynomials cannot be used here, since the graviton modes are not Jack polynomials (note that the Haffnian states are also not Jack polynomials). Using the same reasoning, it is easy to see that the graviton modes of the Laughlin-$1/5$ state is the zero energy state of the Laughlin-$1/3$ model Hamiltonian, given that the LEC of the Laughlin-$1/3$ state is $\hat{n} = \{2, 1, 2\}$. 

It is worth noting that although using LEC offers a very simple way of determining if the graviton mode belongs to the null space of some model Hamiltonian, it only applies to the cases where the root configurations are easy to find, and fundamentally the LEC scheme is only ``proven" numerically. We will now proceed to analytically prove the more general cases for the Laughlin phases at $\nu_k =\frac{1}{2k+1}$, $k \in \mathrm{N}^+$. The $\nu=1/3$ case will be left to the next section focusing on the three-body model Hamiltonians. 

In the following discussion, we will denote the Laughlin state with the filling factor $\nu=\frac{1}{2k+1}$ by $| \psi_{\nu_k} \rangle$, the null space of the corresponding model Hamiltonian $\hat H_{\nu_k}$ by $\mathcal{H}_{\text{Laughlin}-\nu_k} $ ($\mathcal{H}_{\text{L}-\nu_k} $ for short) and the corresponding graviton mode $| \psi_{\nu_k, \boldsymbol{q}} \rangle = \delta \bar{\rho}_{\boldsymbol{q}} | \psi_{\nu_k} \rangle$ with $q\rightarrow 0$. The variational energy of $| \psi_{\nu_{k+1}, \boldsymbol{q}} \rangle$ with respect to the model Hamiltonian $\hat{H}_{\nu_k} $ is given by:
\begin{equation}
\begin{aligned}
\delta E_{\boldsymbol{q} \rightarrow 0}&= \Gamma^{\text{2bdy}}_{m n} c^{m} d^{n} \\
&= \sum_{m}^{2k+1} \frac{(-1)^{m} c^m}{2^8 \eta \pi} \times [2(m^2+m+1)d^{m}\\
& \quad -(m+1)(m+2)d^{m+2}-m(m-1) d^{m-2}]
\label{dmenergy}
\end{aligned}
\end{equation}
where $\hat H_{\nu_k}=\sum_{m=1}^{2k+1}c^mL_m\left(q^2\right)e^{-\frac{q^2}{2}}$, $c^m > 0$ and $m$ can only be odd. From Eq. \ref{gsenergy} we know that:
\begin{equation}
E_0 = c^m d^n\delta_{m,n}=\sum_{m}^{2k+3} c^m d^m =0
\end{equation}
Thus as long as $c^m>0$, the ground state energy should always be zero. This is only possible when all the $d^m = 0$ with $m < 2k + 3$. We thus obtain:
\begin{equation}
\begin{aligned}
\delta E_{\boldsymbol{q} \rightarrow 0}= 0
\end{aligned}
\end{equation}

Hence we proved that the graviton mode of the Laughlin state at $\nu = 1/(2k+3)$ (Laughlin-$1/(2k+3)$ graviton mode for short) is contained in the null space of $\hat H_{\nu_k}$. In particular, the graviton mode at $\nu=1/5$ is the exact zero-energy state of the $\hat V_1^{\text{2bdy}}$ pseudopotential. The latter can serve as the model Hamiltonian for the graviton mode. This is true for all the Laughlin states with $\nu<1/3$.

\section{Three-body interactions}\label{sec_3bdy}

To study the low-lying excitations and especially the graviton modes of the non-abelian FQH phases, it is important to extend the characteristic matrix formalism to the model Hamiltonians with three-body interactions. This is because the model Hamiltonians for the non-abelian phases consist of few-body pseudopotentials involving clusters of more than two electrons\cite{PhysRevB.75.075318}. Such interactions can also physically arise from LL mixing\cite{simon2013landau, sodemann2013landau, bishara2009effect, faugno2021unconventional}. Analogous to the two-body case, for three-body interactions we can define the following reduced structure factor for the unperturbed ground state :
\begin{equation}
\begin{aligned}
\bar{S}_{\boldsymbol{q}_1,\boldsymbol{q}_2}&=\sum_{i \ne j \ne k} \langle \psi_{0} \left| \hat{\rho}^i_{\boldsymbol{q_1}} \hat{\rho}^j_{\boldsymbol{q_2}} \hat{\rho}^k_{\boldsymbol{-q_1-q_2}}\right| \psi_{0}\rangle\\
&=\sum_{i \ne j \ne k}\left\langle\psi_{0}\left| e^{i q_{1a} \hat{R}^{a}_{i} } e^{i q_{2a} \hat{R}^{a}_{j} } e^{- i (q_{1a} + q_{2a})\hat{R}^{a}_{k} } \right| \psi_{0}\right\rangle
\label{eq4}
\end{aligned}
\end{equation}
where the indices $i$, $j$ and $k$ denote different electrons. A generic three-body Hamiltonian is given by:
\begin{equation}
\begin{aligned}
\hat{H}_{\text{3bdy}}
=&\int \frac{d^{2} \boldsymbol{q}_1 d^{2} \boldsymbol{q}_2}{(2 \pi)^4} V_{\boldsymbol{q_1, q_2}} \hat{\rho}_{\boldsymbol{q_1}} \hat{\rho}_{\boldsymbol{q_2}} \hat{\rho}_{\boldsymbol{-q_1-q_2}}\\
&-\int \frac{d^{2} \boldsymbol{q}_1 d^{2} \boldsymbol{q}_2}{(2 \pi)^4} V_{\boldsymbol{q_1, q_2}}  (\hat{\rho}_{\boldsymbol{q_1+q_2}} \hat{\rho}_{\boldsymbol{-q_1-q_2}}\\
&\quad + \hat{\rho}_{\boldsymbol{q_1}} \hat{\rho}_{\boldsymbol{-q_1}}+  \hat{\rho}_{\boldsymbol{q_2}} \hat{\rho}_{\boldsymbol{-q_2}})\\
&-N_e \int \frac{d^{2} \boldsymbol{q}_1 d^{2} \boldsymbol{q}_2}{(2 \pi)^4} V_{\boldsymbol{q_1, q_2}}\\
=&\sum_{i \ne j \ne k}\int \frac{d^{2} \boldsymbol{q}_1 d^{2} \boldsymbol{q}_2}{(2 \pi)^4} V_{\boldsymbol{q_1, q_2}} \hat{\rho}^i_{\boldsymbol{q_1}} \hat{\rho}^j_{\boldsymbol{q_2}} \hat{\rho}^k_{\boldsymbol{-q_1-q_2}}
\label{eq5}
\end{aligned}
\end{equation}
Here $V_{\boldsymbol{q}_1,\boldsymbol{q}_2}$ denotes the effective three-body interaction. The ground state energy is thus given by the expectation value of the Hamiltonian with respect to the ground state:
\begin{equation}
\begin{aligned}
E_0=&\sum_{i \ne j \ne k}\int \frac{d^{2} \boldsymbol{q}_1 d^{2} \boldsymbol{q}_2}{(2 \pi)^4} V_{\boldsymbol{q_1, q_2}} \langle \psi_{0} \left| \hat{\rho}^i_{\boldsymbol{q_1}} \hat{\rho}^j_{\boldsymbol{q_2}} \hat{\rho}^k_{\boldsymbol{-q_1-q_2}}\right| \psi_{0}\rangle\\
=& \int \frac{d^{2} \boldsymbol{q}_1 d^{2} \boldsymbol{q}_2}{(2 \pi)^4} V_{\boldsymbol{q_1, q_2}} \bar{S}_{\boldsymbol{q}_1,\boldsymbol{q}_2}
\label{3bge}
\end{aligned}
\end{equation}
which is analogous to the result in the two-body case. By constructing the SMA state as Eq.\ref{SMAstate} shows, the energy of the graviton mode can be written as:
\begin{equation}
\begin{aligned}
\delta E_{\boldsymbol{q}}&=\frac{\langle\psi_{\boldsymbol{q}}|\hat{H}_{\text{3bdy}} \left| \psi_{\boldsymbol{q}}\right. \rangle}{\left. \langle\psi_{\boldsymbol{q}} \right| \psi_{\boldsymbol{q}}\rangle}-E_{0}\\
&=\frac{\left\langle\psi_{0}\left| \left[\delta \hat{\rho}_{-\boldsymbol{q}},\left[\hat{H}_{\text{3bdy}}, \delta \hat{\rho}_{\boldsymbol{q}}\right]\right] \right| \psi_{0}\right\rangle}{2 S_{\boldsymbol{q}}}\\
&=\sum_{i \ne j \ne k}\int \frac{d^{2} \boldsymbol{q}_1 d^{2} \boldsymbol{q}_2}{(2 \pi)^4} V_{\boldsymbol{q_1, q_2}}\\ 
&\quad \times \frac{\left\langle\psi_{0}\left| \left[ \delta \hat{\rho}_{-\boldsymbol{q}},\left[\hat{\rho}^i_{\boldsymbol{q_1}} \hat{\rho}^j_{\boldsymbol{q_2}} \hat{\rho}^k_{\boldsymbol{-q_1-q_2}},  \delta \hat{\rho}_{\boldsymbol{q}}\right]\right] \right| \psi_{0}\right\rangle}{2 S_{\boldsymbol{q}}}
\label{eq6}
\end{aligned}
\end{equation}

For three-body interactions, it will be more convenient to use the Jacobi coordinates, which can separate the degree of freedom of the center-of-mass from other ones while maintaining the commutation relations between coordinates:
\begin{equation}
\begin{aligned}
\hat{R}_{i j}^{a} &=\frac{1}{\sqrt{2}}\left(\hat{R}_{i}^{a}-\hat{R}_{j}^{a}\right) \\
\hat{R}_{i j, k}^{a} &=\frac{1}{\sqrt{6}}\left(\hat{R}_{i}^{a}+\hat{R}_{j}^{a}-2 \hat{R}_{k}^{a}\right)\\
\hat{R}_{i j k}^{a}&=\frac{1}{\sqrt{3}}\left(\hat{R}_{i}^{a}+\hat{R}_{j}^{a}+\hat{R}_{k}^{a}\right)
\label{eq12}
\end{aligned}
\end{equation}
As expected, after coordinate transformation the Hamiltonian in Eq.\ref{eq5} contains no center-of-mass term and can be written as:
\begin{equation}
\hat{H}_{\text{3bdy}}= \sum_{i \ne j \ne k} \int \frac{d^{2} \tilde{\boldsymbol{q}}_1 d^{2} \tilde{\boldsymbol{q}}_2}{3 \times (2 \pi)^4}  V_{\tilde{\boldsymbol{q}}_{1} \tilde{\boldsymbol{q}}_{2}} e^{i \tilde{q}_{1a} \hat{R}_{ij}^{a}} e^{i \tilde{q}_{2a} \hat{R}_{ij,k}^{a}}
\end{equation}
which leads to a global linear transformation in the $k$-space:
\begin{equation}
\begin{aligned}
\tilde{\boldsymbol{q}}_1&=\frac{1}{\sqrt{2}}\left(\boldsymbol{q}_{1}-\boldsymbol{q}_{2}\right)\\
\tilde{\boldsymbol{q}}_2&=\sqrt{\frac{3}{2}}\left(\boldsymbol{q}_{1}+\boldsymbol{q}_{2}\right)
\end{aligned}
\end{equation}
The transformation to the Jacobi coordinates also enables us to properly decompose the three-body calculations into the product of two symmetric two-body ones, with the Laguerre-Gaussian expansions of the effective three-body potential given by:
\begin{equation}
\begin{aligned}
V_{\boldsymbol{\tilde{q}}_1,\boldsymbol{\tilde{q}}_2}&=\sum_{m_1 m_2} c^{m_1 m_2}  L^{(0)}_{m_1}\left(\frac{\tilde{Q}_1}{2}\right) L^{(0)}_{m_2}\left(\frac{\tilde{Q}_2}{2} \right)  e^{-\frac{1}{4}\left(\tilde{Q}_1+\tilde{Q}_2\right) }
\end{aligned}
\end{equation}
where $ L^{(i)}_{m}\left(x\right)$ is the generalised Laguerre polynomials, and $i=0$ gives the usual Laguerre polynomials. Without loss of generality only model Hamiltonians are considered here, so for the three-body interactions we take $i=0$. The reduced three-body structure factor can be expanded as:
\begin{equation}
\begin{aligned}
\bar{S}_{\boldsymbol{\tilde{q}}_1,\boldsymbol{\tilde{q}}_2}& = \sum_{i; n_1 n_2} d_{i}^{n_1 n_2} \cdot \sqrt{\frac{n_1 ! \cdot n_2 !}{(n_1+i) ! \cdot (n_2-i) !}} \cdot \left(\frac{\boldsymbol{\tilde{q}}_1}{\boldsymbol{\tilde{q}}_2}\right)^i\\
& \quad \times L^{(i)}_{n_1}\left(\frac{\tilde{Q}_1}{2}\right) L^{(-i)}_{n_2}\left(\frac{\tilde{Q}_2}{2} \right)  e^{-\frac{1}{4}\left(\tilde{Q}_1+\tilde{Q}_2\right) }
\label{expansion3b}
\end{aligned}
\end{equation}
where $i$ must be even integers because of the fermionic statistics, $|i| \le n_1+ n_2$ due to rotational invariance and $\tilde{Q}_j \equiv \left|\boldsymbol{\tilde{q}}_j \right|^2 $. Note that unlike the generic two-body case, the expansion of three-body interaction also contains the generalized Laguerre polynomials (when $i \ne 0$). Because the generalised Laguerre polynomials form a complete basis, any function of $\boldsymbol{\tilde{q}}_1$ and $\boldsymbol{\tilde{q}}_2$ can be expanded as in Eq.\ref{expansion3b}. Moreover, there seems to exist a singularity in the expansion when $|\boldsymbol{\tilde{q}}_j| \rightarrow 0$, but in fact this is not the case given that $\forall i \in \mathbb{Z} \backslash \mathbb{N}$, $|i| \le n$:
\begin{equation}
\begin{aligned}
\lim_{|\boldsymbol{\tilde{q}}_j| \rightarrow 0} \boldsymbol{\tilde{q}}_j^i \times L^{(i)}_{n}\left(\frac{\tilde{Q}_j}{2} \right) = 0
\end{aligned}
\end{equation}

It is useful to consider the expansion of a three-electron rotationally invariant state with zero center-of-mass angular momentum  in the magnetic field $| \psi_3 \rangle = \sum_{n_1, n_2} \alpha^{n_1, n_2} | n_1 ,n_2\rangle$, where the quantum number $n_1$ denotes the relative momentum between the first and the second electron, and $n_2$ represents the relative momentum between the center-of-mass of two electrons and the third electron\cite{laughlin1983quantized}. Due to the fermionic statistics, the coefficients of expansion $\alpha^{n_1, n_2} $ are fixed with an overall factor and their specific values can be found in Ref.\cite{yang2018three} and are also attached in the supplementary materials. Thus for any generic rotationally invariant many-body wavefunction with zero center-of-mass angular momentum, $d_i^{n_1 n_2}$ should be proportional to the product of two expansion coefficients of $| n_1 ,n_2\rangle$ based on Eq.\ref{eq4}:
\begin{equation}
d_i^{n_1 n_2} \propto \alpha^{* n_1, n_2} \alpha^{n_1+i, n_2-i}
\label{dn1n22}
\end{equation}
This immediately leads to the conclusion that if $d_0^{n_1 n_2}=0$, all the other $d_i^{n_1 n_2}$ with the same indices $n_1$ and $n_2$ vanish as well. By substituting Eq.\ref{expansion3b} to the ground state energy Eq.\ref{3bge}, a set of equations on the expansion coefficients can be derived from the orthogonality condition of the generalized Laguerre polynomials, which gives the ground state energy:
\begin{equation}
E_0 = \delta_{m_1, n_1} \delta_{m_2, n_2} c^{m_1 m_2} d_0^{n_1 n_2}
\label{3bgs}
\end{equation}

For model wavefunctions and the corresponding model Hamiltonians, the ground state energy is zero. Using the same technique as the two-body case, we can write down the zeroth-order term of the energy gap $\delta \tilde{E}_{\boldsymbol{q}}$ in the long-wavelength limit, given as follows:
\begin{equation}
\begin{aligned}
\delta \tilde{E}_{\boldsymbol{q} \rightarrow 0}
=&\Gamma_{m_1 m_2 n_1 n_2}^{i}  c^{m_1 m_2} d_{i}^{n_1 n_2}\\
\propto& \left( \Gamma_{m_1 m_2 n_1 n_2}^{i} \alpha_{n_1+i, n_2-i}\right) c^{m_1 m_2} \alpha^{* n_1, n_2}
\label{original_gamma}
\end{aligned}
\end{equation}
which means that the $i$-dependence in $d_{i}^{ n_1 n_2}$ can be absorbed into the characteristic tensor. The final result is given by:
\begin{equation}
\begin{aligned}
\delta \tilde{E}_{\boldsymbol{q} \rightarrow 0}
=& \tilde{\Gamma}_{m_1 m_2 n_1 n_2}^{\text{3bdy}} c^{m_1 m_2} \bar{d}^{n_1 n_2}\\
=&  (\Gamma_{m_1 m_2 n_1 n_2}^{0} +\Gamma_{m_1 m_2 n_1 n_2}^{+}+\Gamma_{m_1 m_2 n_1 n_2}^{-}) \\
& \quad \times c^{m_1 m_2} \bar{d}^{n_1 n_2} \times \mathcal{C}
\label{deltaEq}
\end{aligned}
\end{equation}
where the constant $\mathcal{C} = -1/(2^{6}\times3 \eta \pi^2)$ and we have defined the factor:
\begin{equation}
\begin{aligned}
\bar{d}^{n_1 n_2} \equiv& \scalebox{0.9}{$\int \frac{d^{2} \boldsymbol{\tilde{q}}_1 d^{2} \boldsymbol{\tilde{q}}_2}{(2 \pi)^4} \bar{S}_{\boldsymbol{\tilde{q}}_1,\boldsymbol{\tilde{q}}_2} L^{(0)}_{n_1}\left(\frac{\tilde{Q}_1}{2}\right) L^{(0)}_{n_2}\left(\frac{\tilde{Q}_2}{2} \right)  e^{-\frac{1}{4}\left(\tilde{Q}_1+\tilde{Q}_2\right)} $}\\
\propto & \alpha^{n_1, n_2} = \sum_{n_1, n_2} \langle n_1, n_2  | \psi_3\rangle
\label{d12}
\end{aligned}
\end{equation}
The explicit expressions of the characteristic tensor components are given by:
\begin{equation}
\begin{aligned}
\Gamma_{m_1 m_2 n_1 n_2}^{0}=&2(n_1^2 + n_2^2 + n_1 +n_2 +2) \delta_{n_1, m_1} \delta_{n_2, m_2}\\
&  -n_1 (n_1 -1) \delta_{n_1, m_1+2} \delta_{n_2, m_2}\\
&  - (n_1 +1) (n_1 + 2) \delta_{n_1, m_1-2} \delta_{n_2, m_2}  \\
&  -n_2 (n_2 -1)  \delta_{n_1, m_1} \delta_{n_2, m_2+2}\\
&  - (n_2 +1) (n_2 + 2)  \delta_{n_1, m_1} \delta_{n_2, m_2-2}
\label{3bgamma}
\end{aligned}
\end{equation}
which corresponds to the diagonal terms ($i=0$) in the Laguerre-Gaussian expansion of the structure factor. The expression is also reminiscent of two copies of two-body characteristic matrices given in Eq.\ref{Gamma1}. Meanwhile the contributions from the off-diagonal terms in the reduced ground state structure factor are given by:
\begin{equation}
\begin{aligned}
\Gamma_{m_1 m_2 n_1 n_2}^{+}=&\frac{\alpha_{n_1+2, n_2-2}}{\alpha_{n_1, n_2}} \times \sqrt{\frac{n_1 ! n_2! }{(n_1+2)! (n_2-2)!}} \\
& \times (n_1 +1)(n_1 + 2)\left( \delta_{m_1, n_1}  -\delta_{m_1, n_1+2} \right)\\
& \times (\delta_{m_2, n_2}-\delta_{m_2, n_2-2}  )
\label{3bomega2}
\end{aligned}
\end{equation}
and
\begin{equation}
\begin{aligned}
\Gamma_{m_1 m_2 n_1 n_2}^{-}=& \frac{\alpha_{n_1-2, n_2+2}}{\alpha_{n_1, n_2}} \times \sqrt{\frac{n_1 ! n_2! }{(n_1-2)! (n_2+2)!}} \\
& \times(n_2 +1)(n_2 + 2) \left( \delta_{m_2, n_2}  - \delta_{m_2, n_2+2} \right)\\
& \times (\delta_{m_1, n_1}-\delta_{m_1, n_1-2}  )
\label{3bomega3}
\end{aligned}
\end{equation}
which correspond to  $i = +2$ and $-2$ in Eq.\ref{original_gamma}. These terms physically captures the contributions to the energy of the graviton modes from different $| n_1 ,n_2\rangle$ components.  Based on the anti-symmetric property of fermionic wavefunctions, $n_1$ must be odd. A rigorous and detailed derivation on all the results we have got so far can be found in the supplementary material of this paper. 

Note that compared to the existing methods of studying the graviton modes, the characteristic tensor formalism can determine the neutral gap with just the information of the ground state, making the exact diagonalization to obtain excited energy spectra unnecessary. Since  $\Gamma^{\text{3bdy}}_{m_1 m_2 n_1 n_2}$ is universal and independent of the microscopic details, for a given FQH state with a model Hamiltonian or even realistic interaction, once the numerical properties of the expansion coefficients are determined, the behaviour of the graviton mode can be fully depicted by Eq.\ref{deltaEq}. 

\begin{figure}
\centering
\includegraphics[scale=0.25]{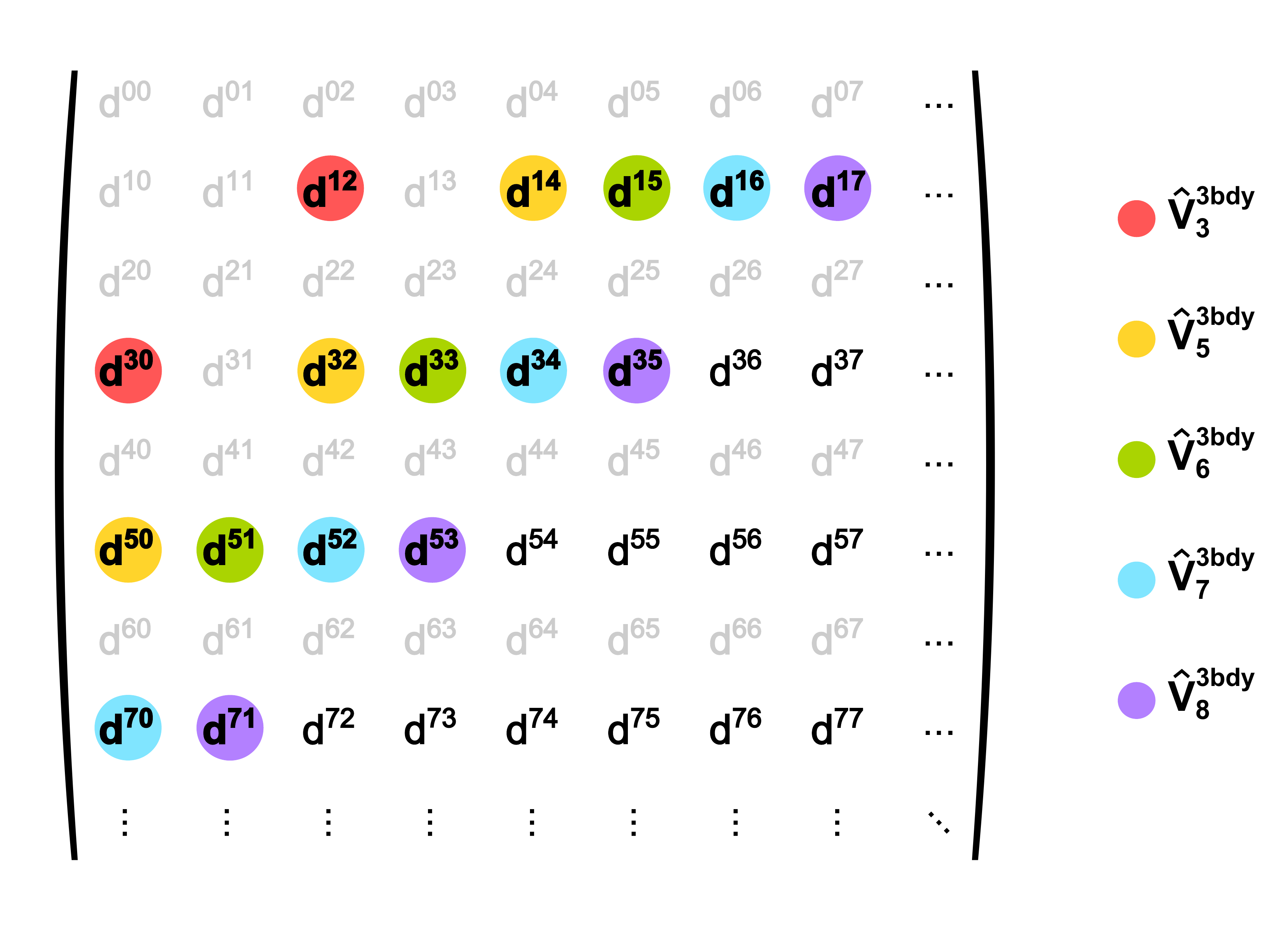}
\caption{\textbf{Structure of the $\bar{d}_{n_1 n_2}$ coefficients of the Laughlin-$1/3$ state.} Here gray coefficients are zero because of the Fermionic statistics. Bold coefficients are zero because of the vanishing ground state energy (contributions from different three-body model Hamiltonians are denoted by different colors). Black coefficients are unknown and not necessarily zero. Note that here we omit the bar over the $\bar{d}_{n_1 n_2}$ coefficients for simplicity.}
\label{3bodydn1n2}
\end{figure}

\subsection{Model Hamiltonian of Laughlin-$1/3$ graviton mode \label{sec:1/3_density}}

With the help of the three-body characteristic tensor formalism, now the behaviour of the Laughlin-$1/3$ graviton mode can be analytically discussed. From Eq.\ref{d12} and Eq.\ref{dn1n22}, we can decompose $\bar d^{n_1,n_2}$ as follows: 
\begin{equation}
\bar{d}^{n_1, n_2} = \sum^{k_{n_1 + n_2}}_{i=1} \lambda_{n_1+n_2,i}  A^{n_1, n_2}_{i} 
\label{dn1n2}
\end{equation}
where $k_{n_1+n_2}$ denotes the degeneracy of three-body pseudopotentials (or the highest-weight three-body eigenstates of the total angular momentum operator) with total relative angular momentum $n_1+n_2$; $\lambda_{n_1+n_2,i}$ depends on the ground state wavefunction,  but $A^{n_1, n_2}_{i} $ are well-defined as shown in Ref.\cite{yang2018three}. When there is no degeneracy in the highest weight state wavefunction (i.e. $k_{n_1+n_2}=1$, which is true for $n_1+n_2<9$),  $\bar{d}^{n_1, n_2} $ can be considered to be proportional to $A^{n_1, n_2}_{1}  =  |\alpha^{n_1, n_2}|^2 $. Thus from the ground state energy in Eq.\ref{3bge}, considering $d^1=0$ which gives $d^{1, n_2}=0$ for this state, we have:
\begin{equation}
\begin{aligned}
& \lambda_{1+n_2, 1} &= 0, &\quad \text{for} \quad n_2=2, 4, 5, 6, 7\\
\Rightarrow & \quad \lambda_{a, 1} &= 0, &\quad \text{for} \quad  a =3, 5, 6, 7, 8\\
\Rightarrow & \quad \bar{d}^{n_1 n_2} &= 0, &\quad \text{for} \quad  n_1 +n_2 =3, 5, 6, 7, 8
\end{aligned}
\end{equation}
We have thus proved that the graviton mode of the Laughlin-$1/3$ state lives in the null space of the following three-body Hamiltonian:
\begin{eqnarray}
\hat H^{\text{3bdy}}=\sum_{i=3}^8\tilde c_i\hat V_i^{\text{3bdy}}
\end{eqnarray}
where $\tilde c_4=0$ and $\tilde c_i>0$ otherwise. In particular, the graviton mode is in the null space of the Haffnian model Hamiltonian $\hat{H}_{\text{Haffnian}} = \hat{V}^{3bdy}_3 + \hat{V}^{3bdy}_5 +\hat{V}^{3bdy}_6$  as Fig.\ref{nullspace} shows. This Hamiltonian provides us with non-vanishing coefficients of $c_{m_1, m_2}$ and $m_1+m_2=3$, $5$ or $6$. Thus based on Eq.\ref{3bgamma}, the energy of the graviton mode of the Laughlin-$1/3$ state is given by:
\begin{equation}
\begin{aligned}
\delta \tilde{E}_{\boldsymbol{q}} =&  \tilde{\Gamma}_{m_1 m_2 n_1 n_2}^{\text{3bdy}} c^{m_1 m_2} \bar{d}^{n_1 n_2}=0 
\end{aligned}
\end{equation}
where the algebraic properties of $\tilde{\Gamma}_{m_1 m_2 n_1 n_2}^{\text{3bdy}}$ make sure that only the coefficients $\bar{d}^{n_1 n_2}$ with $n_1 +n_2 \le 8$ are involved. 

We also illustrate the behaviours of $\bar d^{n_1,n_2}$ in Fig.\ref{3bodydn1n2}, all the bold coefficients are zero given by the ground state energy in Eq.\ref{3bgs} of the corresponding model Hamiltonian, denoted by different colors. Meanwhile the fermionic statistics ensures that all the gray coefficients have to vanish. Although the black coefficients are generally not zero, none of them are involved in the result. Thus one can prove that the graviton mode energy of the Laughlin-$1/3$ state with the Haffnian Hamiltonian is zero.

\subsection{Moore-Read graviton mode}
Based on the same idea, only two coefficients $\bar{d}^{30}$ and $\bar{d}^{12}$ of the Moore-Read state vanish from the ground state energy, so the energy gap of the corresponding graviton mode with the Hamiltonian $\hat{H}_{\text{MR}}=\hat{V}^{3bdy}_{3}$ is given by:

\begin{equation}
\begin{aligned}
\delta \tilde{E}_{\boldsymbol{q}} =&  \tilde{\Gamma}_{m_1 m_2 n_1 n_2}^{\text{3bdy}} c^{m_1 m_2} \bar{d}^{n_1 n_2} \\
=& c^{30}(\tilde{\Gamma}_{3032}^{\text{3bdy}}  \bar{d}^{32} +\tilde{\Gamma}_{3050}^{\text{3bdy}}  \bar{d}^{50}) +c^{12}(\tilde{\Gamma}_{1232}^{\text{3bdy}}  \bar{d}^{32} + \tilde{\Gamma}_{1214}^{\text{3bdy}} \bar{d}^{14})\\
=& -\frac{1}{2^{8}\times3 \eta \pi^2} \times \left[ c^{30}\left( 8 \bar{d}^{32} -16 \bar{d}^{50} \right)  \right.\\
&\left. +c^{12} \left(-24 \bar{d}^{32} - 16 \bar{d}^{14}\right)\right]\\
\propto& \lambda_{5, 1} \times \left( c^{30}+ 3  c^{12}  \right)
\end{aligned}
\end{equation}
where $\lambda_{5, 1}$ is the coefficient of the structure factor of the states with the total angular momentum quantum number $n_1+n_2=5$ as in Eq.\ref{dn1n2}, which is proportional to the expectation value of $\hat{V}^{\text{3bdy}}_5$ with respect to the Moore-Read state and we have used the ratio of the coefficients given in Ref.\cite{yang2018three}:
\begin{equation}
\alpha^{50}:\alpha^{32}:\alpha^{14} = -\frac{\sqrt{5}}{4}:\frac{1}{2\sqrt{2}}:\frac{3}{4}
\end{equation}
Thus the graviton mode energy of the Moore-Read state can be determined by tuning $c^{30}$ and $c^{12}$. Furthermore, similar to the Laughlin-$1/3$ state, the graviton mode of the Moore-Read state should live in the null space of four-body interactions, which will not be discussed in detail here.

\subsection{Gaffnian graviton mode}

For the Gaffnian state, the coefficients $\bar{d}^{30}$, $\bar{d}^{12}$, $\bar{d}^{50}$, $\bar{d}^{32}$ and $\bar{d}^{14}$ vanish because of the model Hamiltonian $\hat{H}_{\text{Gaffnian}}=\hat{V}^{3bdy}_{3}+\hat{V}^{3bdy}_{5}$. Thus the energy gap of the corresponding graviton mode can be written as:
\begin{equation}
\begin{aligned}
\delta \tilde{E}_{\boldsymbol{q}} =&c_{50}(\tilde{\Gamma}_{5052}^{\text{3bdy}}  \bar{d}^{52}+\tilde{\Gamma}_{5070}^{\text{3bdy}} \bar{d}^{70})\\
&+c^{32}(\tilde{\Gamma}_{3234}^{\text{3bdy}} \bar{d}^{34} +\tilde{\Gamma}_{3252}^{\text{3bdy}} \bar{d}^{52})\\
&+c^{14} (\tilde{\Gamma}_{1416}^{\text{3bdy}} \bar{d}^{16}+\tilde{\Gamma}_{1434}^{\text{3bdy}} \bar{d}^{34})\\
=&-\frac{1}{2^{8}\times3 \eta \pi^2} \times \left[c^{50}(40 \bar{d}^{52} -40 \bar{d}^{70})\right.\\
&+c^{32}(-16 \bar{d}^{34} -80 \bar{d}^{52})\\
&\left. +c^{14} (-40 \bar{d}^{16} -24 \bar{d}^{34}) \right]\\
\propto& \lambda_{7, 1} \times \left(5 c^{50}+ 2 c^{32} + 9 c^{14}  \right)
\label{gaffnianneutral}
\end{aligned}
\end{equation}
where $\lambda_{7, 1}$ is the coefficient of the structure factor of the states with the total angular momentum quantum number $n_1+n_2=7$, which is proportional to the expectation value of $\hat{V}^{\text{3bdy}}_7$ with respect to the Gaffnian state, and the ratio of the coefficients given in Ref.\cite{yang2018three} has been used:
\begin{equation}
\alpha^{70}:\alpha^{52}:\alpha^{34}:\alpha^{16} = -\frac{\sqrt{21}}{8}:\frac{1}{8}:\frac{\sqrt{15}}{8}:\frac{3\sqrt{3}}{8}
\end{equation}

\begin{figure}
\centering
\includegraphics[scale=0.13]{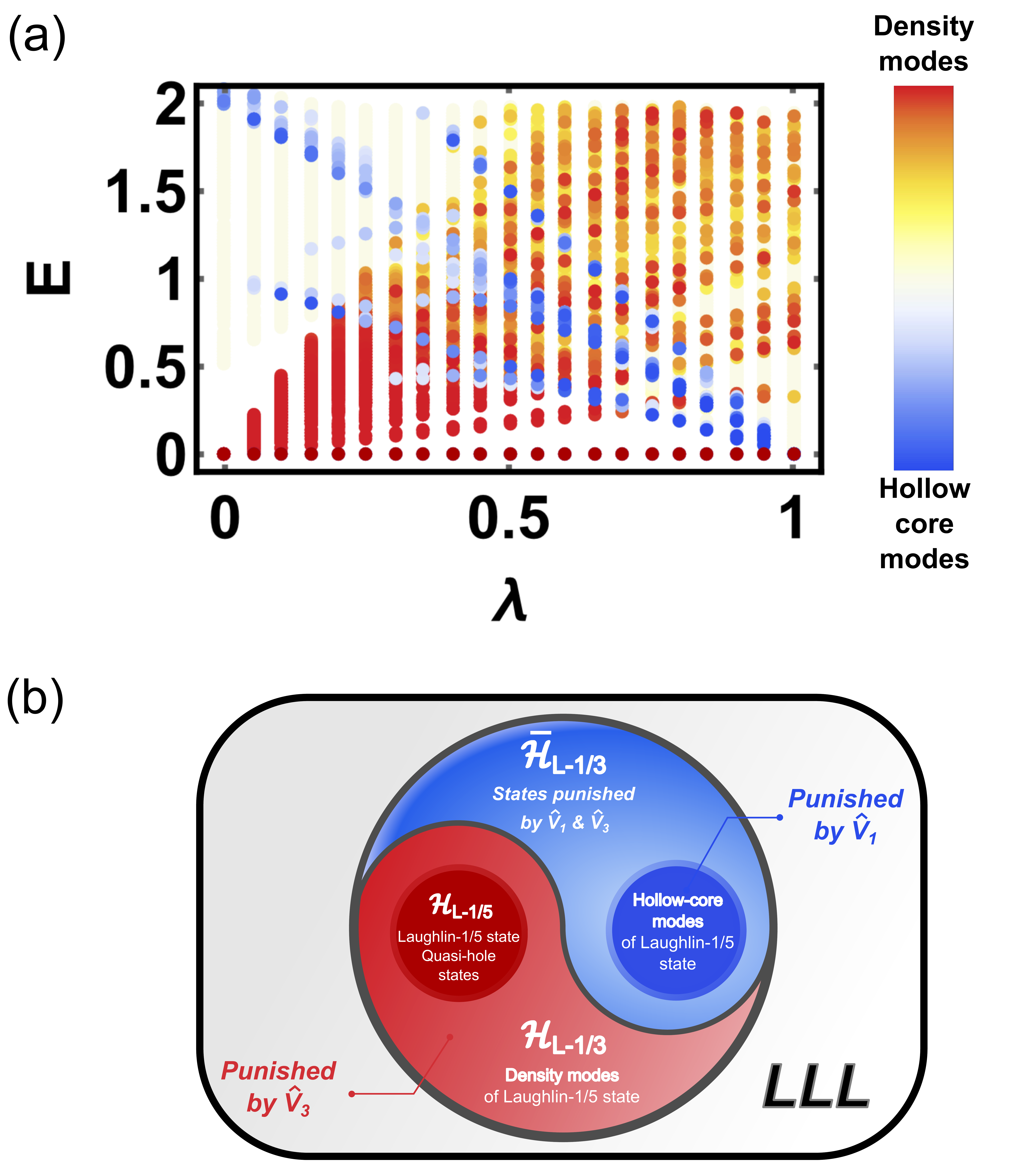}
\caption{\textbf{Spectrum of the toy Hamiltonian $\hat{H}_L$ with respect to 6 electrons and 26 orbitals.} (a) As the color bar shows, the red dots depict the spectrum of the density modes and the hollow-core modes are denoted by blue dots. The color of each dot is determined by calculating their collective overlap with all the density modes in the Laughlin-$1/3$ null space $\mathcal{H}_{\text{L-1/3}}$, or all the hollow-core modes in the complementary space $\mathcal{\bar{H}}_{\text{L-1/3}}$. As $\lambda$ increases, though the ground state (denoted by dark red) stays invariant, the low-lying states show a clear cross-over behaviour and transform from density modes to hollow-core modes. (b) illustrates the structure of the Hilbert space and the relationship between the states and the model Hamiltonian. The Laughlin-$1/5$ null space $\mathcal{H}_{\text{L-1/5}}$ (dark red circle) punished by neither $\hat V^{\text{2bdy}}_1$ nor $\hat V^{\text{2bdy}}_3$ is the sub-space of $\mathcal{H}_{\text{L-1/3}}$ (only punished by $\hat V^{\text{2bdy}}_3$). Meanwhile these exist the hollow-core modes (blue circle) only punished by $\hat V^{\text{2bdy}}_1$ in $\mathcal{\bar{H}}_{\text{L-1/3}}$. All of the other states are punished by both $\hat V^{\text{2bdy}}_1$ and $\hat V^{\text{2bdy}}_3$, living in the remaining part of $\mathcal{\bar{H}}_{\text{L-1/3}}$.}
\label{laughlin15}
\end{figure}

Thus the graviton mode energy of the Gaffnian state can be determined by tuning $c^{50}$, $c^{32}$ and $c^{14}$. Furthermore, there exist no $c^{30}$ and $c^{12}$ terms in Eq.\ref{gaffnianneutral}, which indicates the graviton mode of the Gaffnian state should live in the null space of the Moore-Read model Hamiltonian $\hat{V}^{3bdy}_{3}$. In particular, the variational energy of the Gaffnian graviton mode is independent of the strength of $\hat V_3^{\text{3bdy}}$ in the Hamiltonian.

\begin{figure}
\centering
\includegraphics[scale=0.128]{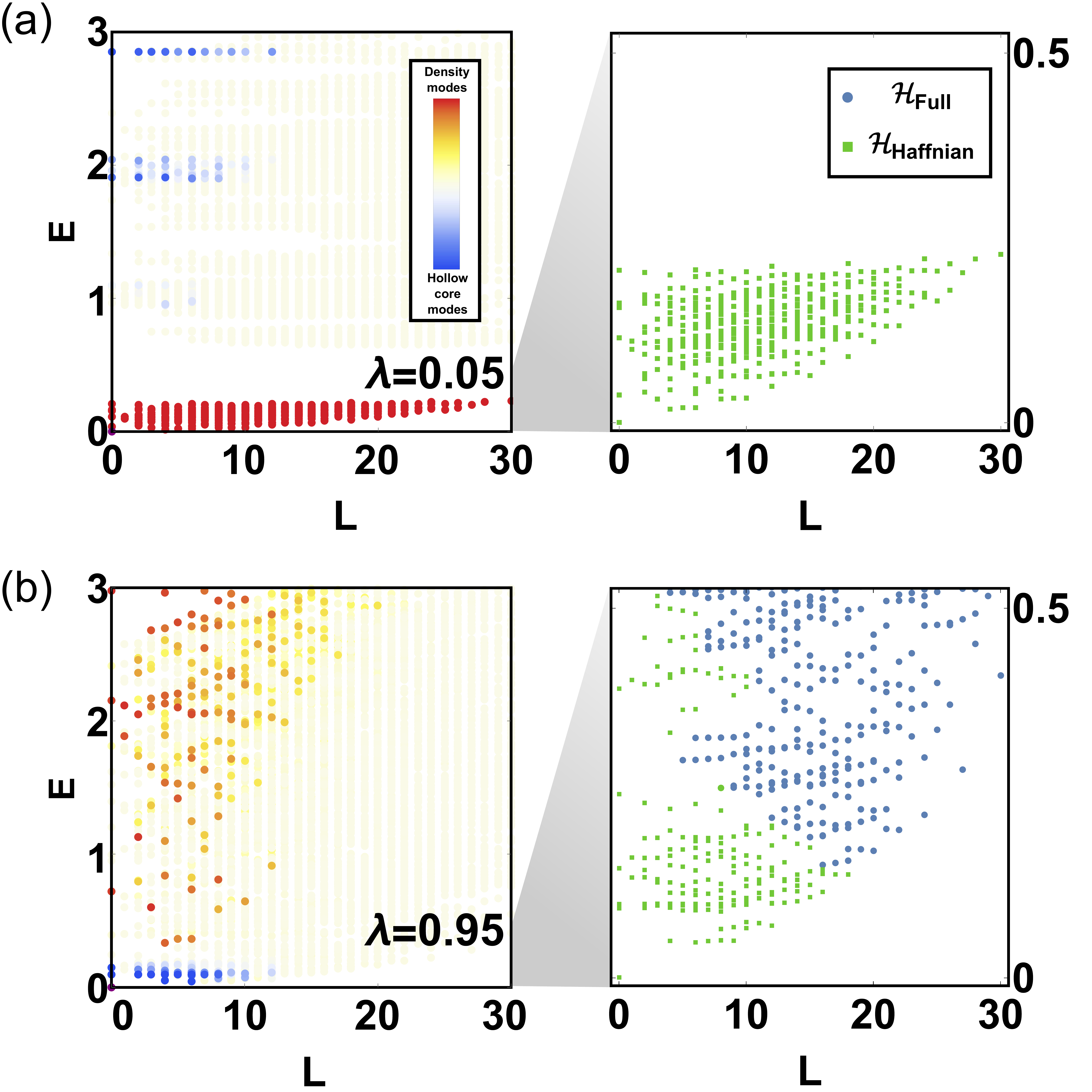}
\caption{\textbf{Spectra of the toy Hamiltonian $\hat{H}_L$ diagonalized in different Hilbert spaces.} The case with $\lambda = 0.05$ is shown in (a) and $\lambda = 0.95$ in (b). In the left panel, one can clear observe the transition of the low-lying states when $\lambda$ increases from $0.05$ to $0.95$. These states have different nature as explained in Fig.\ref{CF}. In the right panel, we zoom in on the spectra to the range $E \in [0, 0.5]$ and mark the states living in the Haffnian null space with green squares and other states with blue dots. As expected, by using the CF picture, both the density modes and the hollow-core modes live in the Haffnian null space. Note that the lowest angular momentum of the hollow-core modes is $L_{\text{min}}=4$ as can be seen in (b). Furthermore, the quantized energy of the hollow-core modes (especially when $\lambda \rightarrow 0$) can be understood by using the root configurations in Eq.\ref{hcroot}.}
\label{Laughlin_Gaffnian}
\end{figure}

\section{Numerical study}\label{sec_numerics}
As shown in Fig.\ref{nullspace}, the theoretical derivations have revealed the well-organized hierarchical structure of the null space of multiple model Hamiltonians, in which the graviton modes of different FQH states reside. For the graviton modes to be experimentally relevant, they have to be low-lying states in the excitation spectrum. For fully-gapped FQH phases, the graviton modes are also gapped. Moreover, given they are neutral excitations, for realistic interactions they may also become gapless without affecting the robustness of the Hall conductivity plateau\cite{kang2017neutral, jolicoeur2014absence, yang2021gaffnian}. We now proceed to perform numerical calculations for the dynamical properties of the graviton modes, using the theoretical tools developed in the previous sections. We focus in particular on the Laughlin-$1/5$ state and the Gaffnian state, and discuss possible theoretical and experimental consequences.

\begin{figure}[t]
\centering
\includegraphics[scale=0.2]{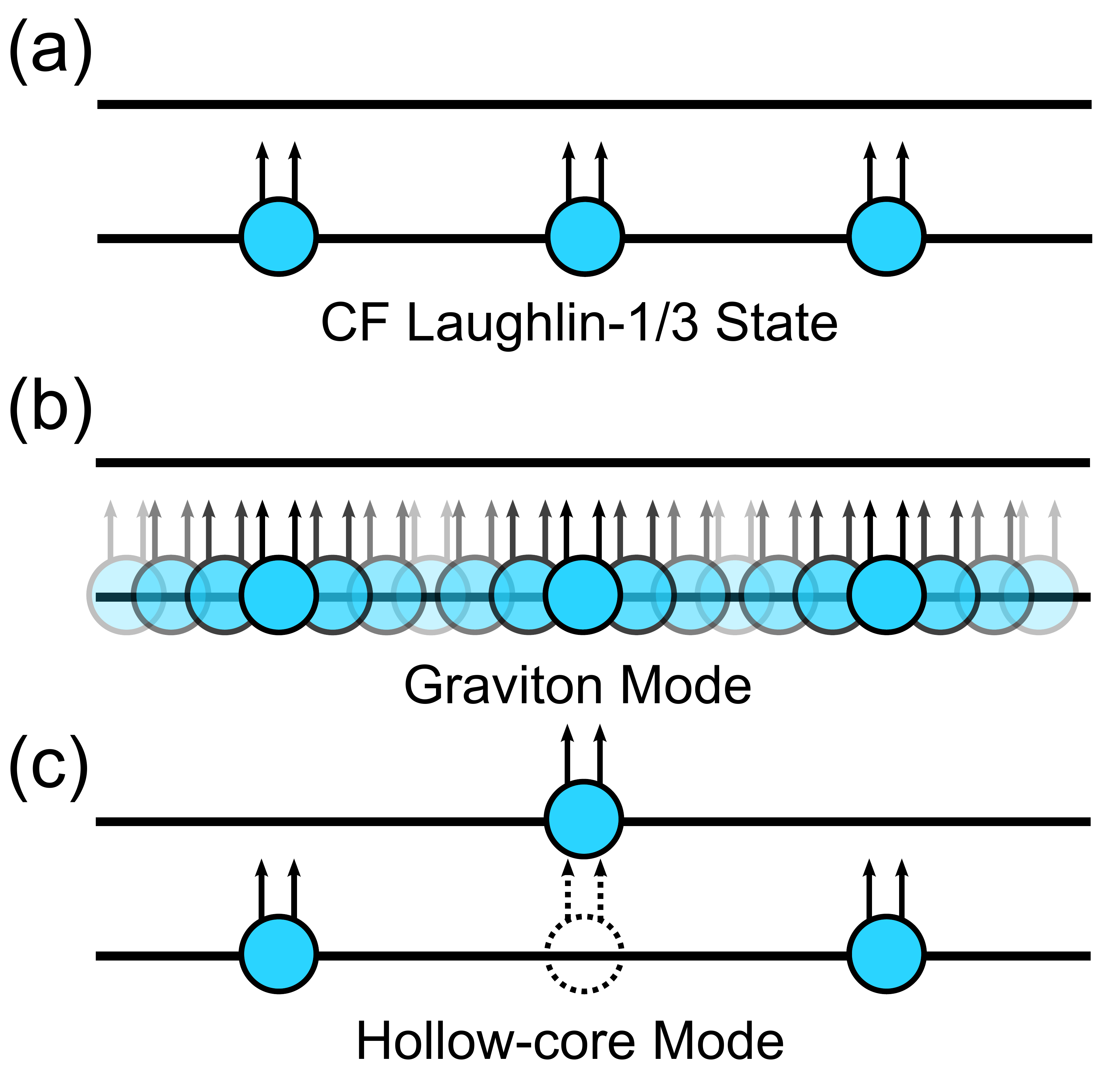}
\caption{\textbf{Nature of the low-lying states with different model Hamiltonian in the Composite Fermion picture.} (a) The Laughlin-$1/5$ state of the electrons can be reinterpreted as the Laughlin-$1/3$ state of CFs consisted of one electron and two fluxes. (b) The graviton modes can be understood as the excitations of CFs in the lowest $\Lambda$ level. (c) The hollow-core modes are created by exciting CFs to the second $\Lambda$ level, which still live in the Gaffnian null space. }
\label{CF}
\end{figure}

\subsection{Laughlin-$1/5$ graviton mode}
We have shown that the graviton mode of Laughlin-$1/5$ state lives within the null space of $\hat V_1^{\text{2bdy}}$, and it is a quantum fluid of Laughlin quasiholes. If we look at a short-range interaction with model Hamiltonians involving only $\hat V_1^{\text{2bdy}}$ and $\hat V_3^{\text{2bdy}}$, the dynamics of the graviton modes is completely controlled by $\hat V^{\text{2bdy}}_3$. It is thus useful to understand the low-lying excitations of the following toy model:
\begin{equation}
\hat{H}_L = (1-\lambda) \hat{V}^{\text{2bdy}}_1 + \lambda \hat{V}^{\text{2bdy}}_3
\label{Hamil1}
\end{equation}
The ground state is invariant when tuning the value of $\lambda$. In contrast, the low-lying excitations can be qualitatively different. In particular, when $\lambda$ is close to zero, the graviton mode and the magnetoroton modes should be the low-lying excitations. On the contrary, if there exist states that are punished by $\hat V^{\text{2bdy}}_1$ but not $\hat V^{\text{2bdy}}_3$, then they will become the low-lying states when $\lambda$ is close to unity. Thus one can expect to see the transition of the low-lying states when $\lambda$ is continuously increased from $0$ to $1$. 

\begin{figure}[t]
\centering
\includegraphics[scale=0.13]{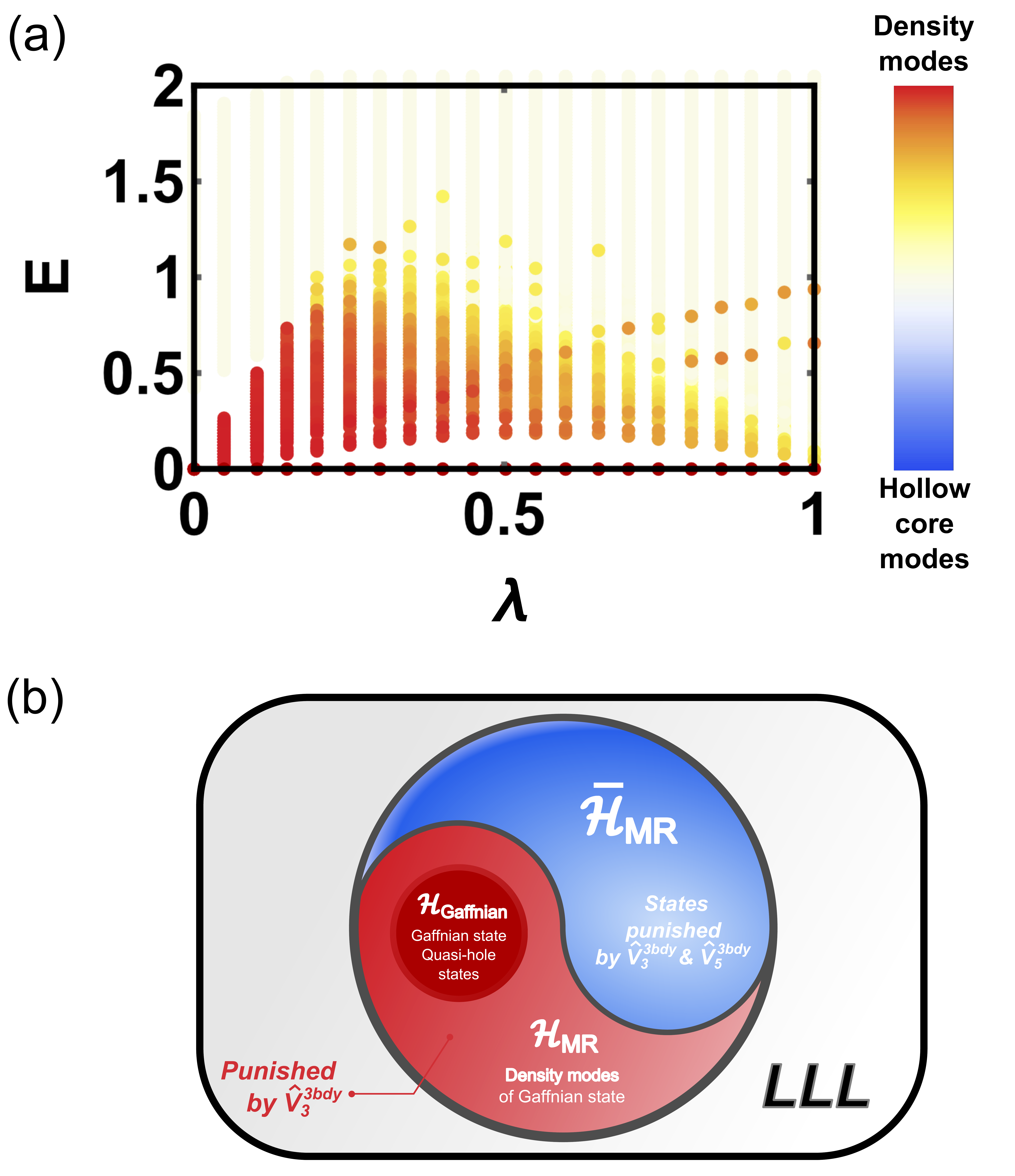}
\caption{\textbf{Spectrum of the toy Hamiltonian $\hat{H}_G$ with respect to 10 electrons and 22 orbitals.} (a) shows the spectra of the model Hamiltonian $\hat{H}_G$ in Eq.\ref{Hamil3}. The low-lying states are density modes even when $\lambda$ is quite large, and the absence of hollow-core modes is different from the Laughlin-$1/5$ state in Fig.\ref{laughlin15}. (b) illustrates the structure of the Hilbert space and the relationship between the states and the model Hamiltonian. All the states in the complementary space of the Moore-Read null space are punished by both $\hat{V}^{\text{3body}}_3$ and $\hat{V}^{\text{3body}}_5$.}
\label{gaff_10e}
\end{figure}

The results of the Laughlin states with $6$ electrons are shown in Fig.\ref{laughlin15}. While the ground state is invariant (Laughlin-$1/5$ state, denoted by the dark red color in Fig.\ref{laughlin15}), the low-lying excitations show a clear cross-over behaviour. When $\lambda \rightarrow 1$, the density modes including the graviton modes and the multi-magnetoroton modes, shown by red spectrum in Fig.\ref{laughlin15} (a), are no longer low-lying excitations. The structure of the Hilbert space in the LLL is illustrated in Fig.\ref{laughlin15} (b). The null space of the Laughlin-$1/5$ model Hamiltonian (Laughlin-$1/5$ null space for short) denoted by the red circle is a proper subspace of the Laughlin-$1/3$ null space(light-red part), the complement space of which contains the states either only punished by $\hat V^{\text{2bdy}}_1$ (blue circle), or punished by both $\hat V^{\text{2bdy}}_1$  and $\hat V^{\text{2bdy}}_3$ (light-blue part). We can refer to the blue states as the ``hollow-core" modes, since they live in the null space of  $\hat V^{\text{2bdy}}_3$ but out of the null space of $\hat V^{\text{2bdy}}_1$\cite{PhysRevLett.60.956, PhysRevB.38.3636, PhysRevB.102.245107}.

It is useful to look more closely at the spectra of $\mathcal{H}_L$ with $\lambda = 0.05$ and $\lambda = 0.95$ as shown in the left panel of Fig.\ref{Laughlin_Gaffnian}, where the density modes including the graviton modes make up the low-lying states when $\lambda$ is close to $0$. In contrast when $\lambda$ is close to $1$, the energy of these states significantly increases as expected so the low-lying excitations are replaced by the hollow-core modes. To understand better the nature of the low-lying states, one can also diagonalized $\hat{H}_L$ in different sub-Hilbert spaces, instead of the full Hilbert space of a single LL, and to check if the truncation of the Hilbert space affects the low-lying excitations. In the right panel of Fig.\ref{Laughlin_Gaffnian}, the spectra of the Hamiltonian diagonalized in the full Hilbert space are shown, where the states that live almost entirely within the Haffnian null space are denoted by green squares. For both cases, $\hat{H}_L$ with $\lambda = 0.05$ (low-lying excitations consisted of density modes) and $\lambda = 0.95$ (low-lying excitations consisted of hollow-core modes), numerical studies show strong evidence that all the low-lying states live in the Haffnian null space. On the other hand, the graviton and the magnetoroton modes live within the null space of $\hat V^{\text{2bdy}}_1$ (which itself is a subspace of Haffnian null space), while the hollow-core modes live outside of the $\hat V^{\text{2bdy}}_1$ null space.

We can also understand the differences between these two types of low-lying states, by appealing to the intuitive picture from the Composite Fermion(CF) theory\cite{jain1989composite, jain2007composite}. Fig.\ref{CF} illustrates the physical distinctions between the graviton modes (low-lying states when $\lambda = 0.05$) and the hollow-core modes (low-lying states when $\lambda = 0.95$). According to the CF theory, the Laughlin-$1/5$ state of electrons can be reinterpreted as the Laughlin-$1/3$ state of CFs as Fig.\ref{CF} (a) shows, because each CF contains one electron and two fluxes so the filling factor becomes $\nu^* = 1/(5-2) = 1/3$. Similar to the Landau levels of electrons, the discrete levels of CFs are sometimes named as ``$\Lambda$ levels"\cite{jain2007composite}. The Laughlin-$1/3$ null space only contains the states in the first CF level. The graviton and the magnetoroton modes are thus excitations within the partially filled first CF level, which are low-lying excitations for small $\lambda$. In contrast, when $\lambda$ approaches one, the hollow-core modes come from the excitations of the CFs into the second CF level, in some sense similar to the graviton modes of the Laughlin-$1/3$ state. Some of the root configurations containing one or more of the ``hollow-core" mode can be written as:
\begin{equation}
\begin{aligned}
&110000001000010000100001 \cdots \qquad L=4\\
&110000001100000000100001 \cdots \qquad L=8\\
&110000001100000000110000 \cdots \qquad L=12
\label{hcroot}
\end{aligned}
\end{equation}
with the lowest angular momentum $L=4$, agreeing well with Fig.\ref{Laughlin_Gaffnian} (b). From these root configurations one can also understand the quantized energy of the hollow-core modes. Indeed, each pair of electrons in the root configuration (corresponding to a CF in the second CF level) contribute a unit of energy. Thus for six electrons, the highest energy should be $\sim 3$ as shown in Fig.\ref{Laughlin_Gaffnian}(a). 

\begin{figure}[t]
\centering
\includegraphics[scale=0.128]{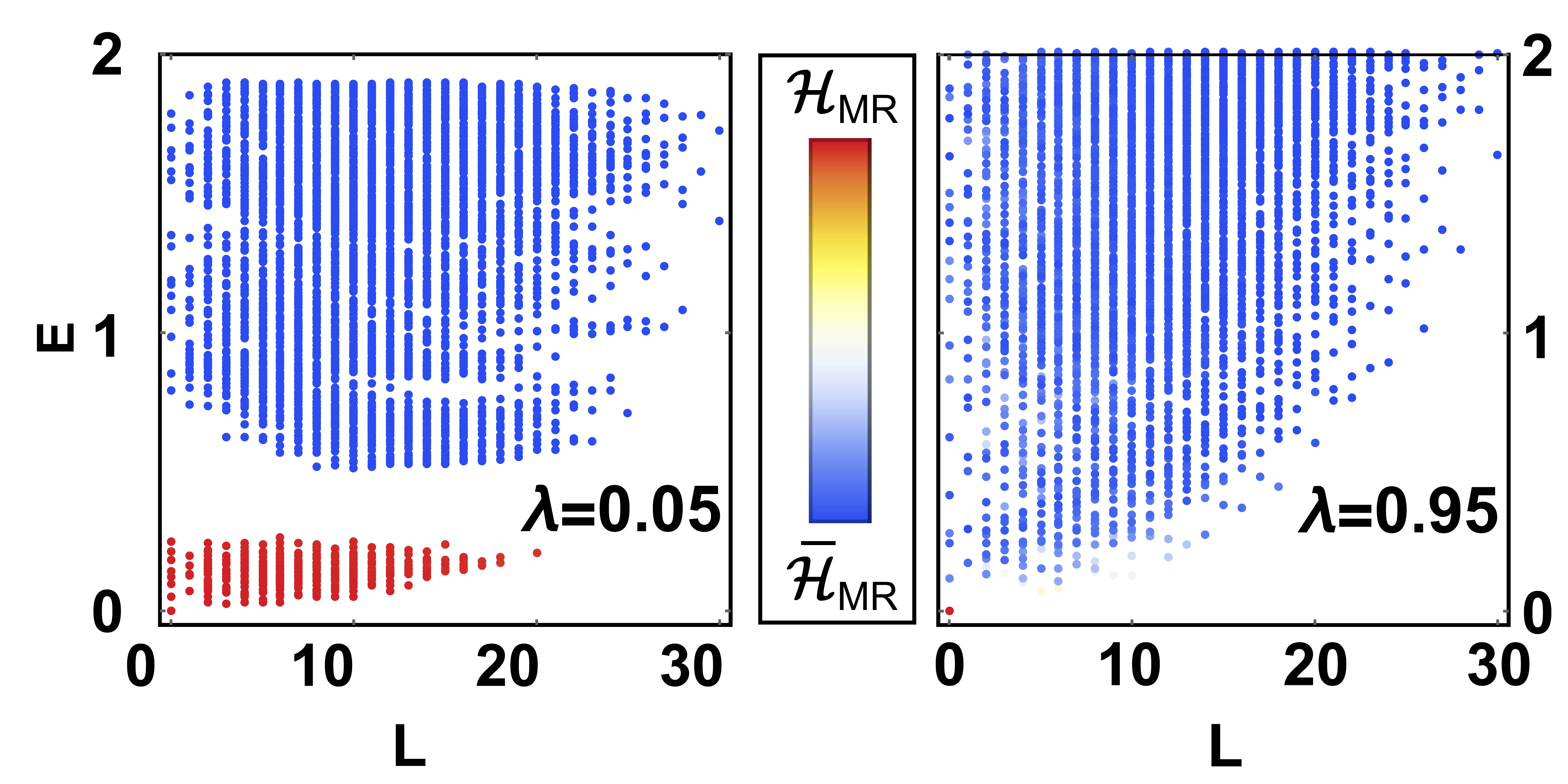}
\caption{\textbf{Spectra of the toy Hamiltonian $\hat{H}_G$ diagonalized in the full Hilbert spaces.} The spectrum with $\lambda = 0.05$ is shown in the left panel and $\lambda = 0.95$ in the right panel. As expected, when $\lambda = 0.05$ all the low-lying states (density modes) live within the Moore-Read null space $\mathcal{H}_{\text{MR}}$. Meanwhile when $\lambda = 0.95$, all the states except the ground state are in the complementary space $\bar{\mathcal{H}}_{\text{MR}}$.}
\label{Gaffnian_with_Laughlin}
\end{figure}

\subsection{Experimental significance}

While we analyse the graviton modes above with only toy models, they can give insights on the experimental measurements of low-lying neutral excitations in FQH phases, using for example Raman scattering or inelastic photon scattering\cite{pinczuk1994inelastic, pinczuk1998light, wurstbauer2013resonant, wurstbauer2015gapped}. For the Laughlin phase at $\nu=1/5$, a short-range realistic interaction (e.g. in the LLL, or with the Coulomb interaction renormalised by sample thickness or screening\cite{zhang1986excitation,park1999activation, peterson2008finite}), the graviton mode as well as the magnetoroton modes will be more prominent. However, since the realistic interaction cannot completely project out the complement of the null space of $\hat V_1^{\text{2bdy}}$, the graviton modes will always mix with the hollow-core modes, so their experimental signals will not be as clean as those from, for example, the Laughlin-$1/3$ phase. 

With longer-range interactions (e.g. in higher LLs), it is still possible for the Laughlin-$1/5$ state to be robust in the sense that the plateau of the Hall conductivity can be observed\cite{Balram_2017}. However, for such interactions, we do not expect clear experimental signals of the graviton modes due to the strong mixing with the hollow-core modes. If the realistic interaction is short-ranged, but dominated by $\hat V_3^{\text{3bdy}}$, there will be no graviton modes (or quadruple excitations) at low energy. Instead, the low-lying excitations are in the complement of the null space of $\hat V_1^{\text{2bdy}}$, and in particular the quasihole excitations can be fractionalised and carry the charge of $e/10$. This is analogous to the nematic FQH phase at $\nu=1/3$ observed in the experiments, and the fractionalisation of the Laughlin-$1/3$ quasiholes near the phase transition\cite{PhysRevLett.127.046402}. It would thus be very interesting if the hollow-core modes, characterised by fractionalised Laughlin-$1/5$ quasiholes, can be observed in experiments.

There was also recent interest in the possibility of the multiple graviton modes in FQH states. Here we show microscopically that at $\nu=1/5$, the Laughlin phase has only a single graviton mode living in the null space of $\hat V_1^{\text{2bdy}}$. In particular, all the density modes are excitations in the lowest CF level, and their coupling to higher CF levels are suppressed by the short-range interaction. It is important to note from our analytical proof that this is the direct consequence of the fact that the Laughlin model wavefunction has exact zero energy with respect to $\hat V_1^{\text{2bdy}}$ and $\hat V_3^{\text{2bdy}}$. On the other hand, the CF state at $\nu=2/7$, which can be understood as the particle-hole conjugate of the Laughlin-$1/5$ state within the lowest CF level, is no longer the exact zero-energy state with respect to $\hat V_1^{\text{2bdy}}$ and $\hat V_3^{\text{2bdy}}$. Thus the graviton mode of the $\nu=2/7$ state will have components both within the null space of $\hat V_1^{\text{2bdy}}$ and the complement of it. One can reinterpret this as multiple graviton modes\cite{PhysRevLett.126.076604,Nguyen2021, Liou2019, Haldane2021}: since the null space of $\hat V_1^{\text{2bdy}}$ corresponds to the lowest CF level, the two graviton modes indeed can be understood as geometric fluctuation within the lowest CF level, as well as the geometric fluctuation associated with the mixing between different CF levels. This will lead to two resonance peaks of opposite chirality with the Raman measurement, while the relative strength of the two peaks depends on the microscopic details of the electron-electron interaction.

\subsection{Gaffnian graviton mode}

The behaviours of the Gaffnian graviton modes at $\nu=2/5$ are not entirely the same as the Laughlin-$1/5$ case. Based on the same idea, one can study these modes by diagonalizing the following Hamiltonian:
\begin{equation}
\hat{H}_G = (1-\lambda) \hat{V}^{\text{3bdy}}_3 + \lambda \hat{V}^{\text{3bdy}}_5
\label{Hamil3}
\end{equation}
The spectrum of the droplet with 10 electrons is shown in Fig.\ref{gaff_10e}. The density modes are still behaving as predicted by the theoretical derivations, i.e. occupying the low-lying states of $\mathcal{H}_G$ with $\lambda \rightarrow 0$. However as shown in Fig.\ref{gaff_10e}, when $\lambda \rightarrow 1$ there exists no state in the Hilbert space that is only punished by $\hat V^{\text{3bdy}}_3$, so the null space of $\hat V^{\text{3bdy}}_5$ lies entirely within the Gaffnian null space (also see Fig.\ref{Gaffnian_with_Laughlin}). Thus there are no hollow-core modes here in contrast to the case for the Laughlin-$1/5$ phase. It would be interesting to see if this is related to the conjecture that the model Hamiltonian of Eq.{\ref{Hamil3}} is gapless in the thermodynamic limit at $\nu=2/5$, while the Laughlin-$\frac{1}{5}$ phase is gapped.

From Eq.\ref{gaffnianneutral}, we know the graviton mode gap of the Gaffnian state at $\nu = 2/5$ is  determined by the expectation value of just $\hat{V}^{\text{3bdy}}_7$ with respect to the ground state, denoted by $d_{\text{3bdy}}^7$ to be consistent with the two-body case in Fig.\ref{finitesize}, where the finite size scaling of the structure factor coefficients of different states is shown. Previous numerical calculations show evidence that in the thermodynamic limit, the gap of Eq.(\ref{Hamil3}) at $\lambda=0.5$ closes in the $L=2$ sector\cite{jolicoeur2014absence}. This is indeed the sector of the graviton mode, and we have shown it is in the null space of $\hat V_3^{\text{3bdy}}$ and its energy is entirely determined by $\hat V_5^{\text{3bdy}}$. Our numerical calculation is thus valid for a family of model Hamiltonian of Eq.(\ref{Hamil3}) parametrized by $\lambda$. It shows that the graviton mode of the Gaffnian phase will likely go soft in the thermodynamic limit, as its variational energy is an order of magnitude smaller than the graviton modes in the Moore-Read phase. It is however also important to note that the graviton mode energy gap of the Laughlin-$1/5$ phase is also an order of magnitude lower than that of the Laughlin-$1/3$ phase, as shown in Fig. (\ref{finitesize}a).

While the Gaffnian model Hamiltonian is a theoretical model that is conjectured to be gapless and thus describing a possibly critical point, it is also closely related to the gapped Jain $\nu=2/5$ phase from short-range two-body interactions\cite{jain1992hierarchy, Freedman2021, simon2007construction}. It would be interesting to see how the graviton modes at $\nu=2/5$ behave when we approach the critical point from the gapped Jain phase at $\nu=2/5$. The finite numerical analysis seems to suggest that both the charge gap and the neutral gap will close, but with realistic interactions, we can also entertain the possibility that the graviton modes of the Jain $\nu=2/5$ phase can close first while the charge gap remains open, in analogy to the nematic FQH phase that has been observed in experiments at $\nu=2+1/3$\cite{regnault2017evidence, maciejko2013field, feldman2016observation}.

\begin{figure}
\centering
\includegraphics[scale=0.21]{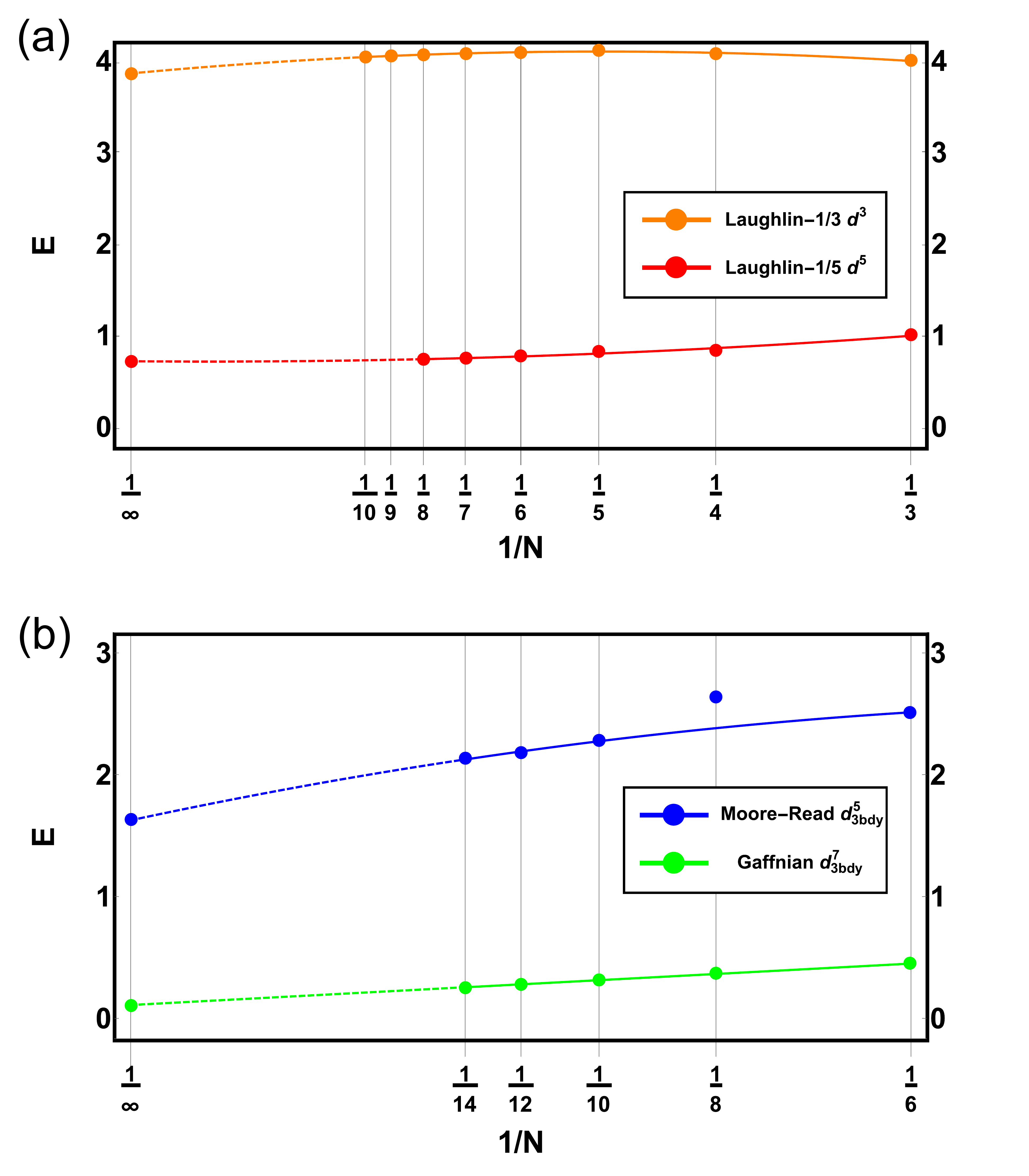}
\caption{\textbf{Finite size scaling of the structure factor coefficients of different states.}(a) The structure factor expansion coefficients of the Laughlin-$1/3$ and the Laughlin-$1/5$ state. (b) The expectation value of  $\hat{V}^{\text{3bdy}}_7$ with respect to the Gaffnian state ($\approx 0.04$ in the thermodynamic limit), denoted by $d_{\text{3bdy}}^7$, is significantly smaller than other coefficients in the plot, where the expectation value of  $\hat{V}^{\text{3bdy}}_5$ with respect to the Moore-Read state is denoted by $d_{\text{3bdy}}^5$ ($\approx 1.6$ in the thermodynamic limit).}
\label{finitesize}
\end{figure}

\section{Summary and outlooks}\label{sec_con}

In summary, we have presented a number of analytical results for the variational energies of the graviton modes in FQH phases. These results are rigorous in the thermodynamic limit, for FQH states with any arbitrary two-body or three-body interactions. In particular, we show that the variational energies of the graviton modes are fully determined by the ground state wavefunction. In addition for short-range interactions, only the leading terms of the ground state structure factor, when expanded in the proper Laguerre polynomial basis, are involved in the computation of the graviton mode energy. These analytical results allow us to construct model Hamiltonians for these graviton modes, which are exact zero-energy states of these Hamiltonians. We can thus determine analytically if the graviton mode lives entirely within a certain conformal Hilbert space, or if they have finite overlaps in different conformal Hilbert spaces. The latter gives microscopic understanding of the multiple graviton modes proposed in the effective field theory descriptions.

There are a number of proposals for the graviton modes to be detected in experiments, but in general it is a difficult task because of the high energy of such excitations. For many FQH phases with Coulomb-based interactions in simple experimental settings, the long-wavelength excitations are not the lowest energy ones. The graviton modes thus have to compete with multi-roton and other neutral excitations. The analytical results we have derived can be useful in understanding how the graviton energy can be affected by realistic interactions, and how we can tune such interactions to lower their variational energies. The generalisation to three-body interactions in this work also allows us to treat Landau level mixing in realistic systems\cite{haldane1997landau, wojs2010landau, sodemann2013landau, simon2013landau, bishara2009effect}, which can be significant in higher LLs. Detailed studies of the graviton modes in the context of real experimental parameters will be carried out elsewhere. The softening of the graviton modes can also allow us to understand potential ``phase transitions" in topological systems even when the ground state topological properties are invariant, as we explored numerically with the Laughlin-$1/5$ state and the Gaffnian state with toy Hamiltonians in this work.

Even with the methodology and the analytical tools developed in this work, the quantitative values of the graviton energy in the thermodynamic limit cannot be determined without numerical computations and finite size scaling. It is, however, a much simpler procedure requiring the computation of only the ground state and partial information about its static structure factor, in contrast to the conventional ways requiring the computation of many low-lying states. The universal characteristic tensors derived in this work show that the Hilbert space of the FQH states are highly structured, and this formalism can in principle be generalised to interactions involving more than three particles. The dispersion of the graviton mode can also be computed analytically by expanding the single mode approximation to higher orders in momentum. At this stage, both cases are algebraically very involved. It would be useful in the future to carry out a more general and systematic calculation of the graviton energy and its dispersion for any arbitrary Hamiltonians. This, combined with a numerically more efficient way to obtain information from the ground state static structure factor (or the density correlation functions), can lead to much better understandings of the collective neutral excitations in non-abelian FQH phases.

\begin{acknowledgments}
This work is supported by the Singapore National Research Foundation (NRF) under NRF fellowship award NRF-NRFF12-2020-0005.
\end{acknowledgments}


\newpage
\pagebreak
\widetext
\begin{center}
\textbf{\large Supplementary material of \\``Analytic exposition of the graviton modes in fractional quantum Hall effects\\ and its physical implications"}
\end{center}
\setcounter{equation}{0}
\setcounter{figure}{0}
\setcounter{table}{0}
\setcounter{page}{1}
\setcounter{section}{0}
\makeatletter
\renewcommand{\theequation}{S\arabic{equation}}
\renewcommand{\thefigure}{S\arabic{figure}}
\renewcommand{\bibnumfmt}[1]{[S#1]}
\renewcommand{\citenumfont}[1]{S#1}

The supplementary material provides the through and  self-contained derivations for the formulas and conclusions in the main text. To strike a balance between sufficient rigor and explicit content we chose to use the mathematical proof formalism to present the whole procedure. Important steps are specifically proved in the \textit{Propositions}. Moreover the \textit{Theorems} contain those formulas that are physically relevant but have not been proved by mathematicians. Lots of efficient \textit{Lemmas} are  used in the derivations, and for simplicity not all the proofs of the lemmas are complete but references will be offered in that case. In the first section, all the assumptions, notations and conventions will be clearly explained, with which we will show the complete derivations of the energy of the graviton mode with three-body interactions, . The two-body case is also reviewed in the following section, which was firstly proposed in Ref.\cite{yang2020microscopic}. The last part shows the table of the anti-symmetric FQH three-body wavefunction expansion in the basis of $| n_1, n_2 \rangle$, from which one can easily get the relationship between the expansion coefficients $\bar{d}^{n_1 n_2}$.

\section{Notations and conventions}
Firstly all the physical assumptions are listed below, which will be universal in the whole derivation:

\textbf{Assumption 1 (Rotational invariance)}

\quad $\cdot$\textit{The quantum numbers related to the angular momentum are thus good quantum numbers.} 

\quad$\cdot$\textit{The interactions used only have radial components.}

\textbf{Assumption 2 (Long-wavelength limit)}

\quad$\cdot$\textit{The momentum of the graviton mode should be  close to $0$, less than the roton minimum.}

\textbf{Assumption 3 (Thermodynamic limit)}

\quad$\cdot$\textit{The particle number in the system is considered as infinity.} 

\textbf{Assumption 4 (Lowest Landau level)}

\quad$\cdot$\textit{All the phsical quantities will be projected to the lowest Landau level.} 

\quad$\cdot$\textit{Only the guiding center coordinates are involved.}

\quad$\cdot$\textit{This is also the origin of many-body interactions.} 

\textbf{Assumption 5 (Single-mode approximation)}

\quad$\cdot$\textit{The graviton mode can be constructed by the regularized density operator acting on the ground state.} 

\textbf{Assumption 6 (Fermionic statistics)}

\quad$\cdot$\textit{All the particles are fermions in the system.}

\quad$\cdot$\textit{All the wavefunctions are anti-symmetric.}

The notations of most of the quantities used in the derivations are given in Table.\ref{notation_table}. Furthermore the magnetic length $l_B$ is set to be $1$. Bold symbols (e.g. $\boldsymbol{q}_i$) are used to denote two-dimensional vectors and the scalars will be plain (e.g. $q_i$). Also the Einstein summation convention are adopted in the results.

\begin{table}[]
\renewcommand\arraystretch{1.5}
\centering
\begin{tabular}{|c|c|}%
\hline
$\hat R_i^a$              & Guiding center operator  \\[3pt] \hline
$| \psi_0 \rangle$        & Ground state  \\ \hline
$\langle \hat{O} \rangle_0$  & Expectation value of the operator $\hat{O}$ acting on the ground state  \\ \hline
$\hat\rho_{\bm q}$        & Guiding center density operator  \\ \hline
$\delta \hat\rho_{\bm q}$ & Regularized guiding center density operator  \\ \hline
$V_{\bm q}$               & Effective potential  \\[3pt] \hline
$\tilde S_{\bm q}$        & Static structure factor  \\[3pt] \hline
$s_{\bm q}$               & Reduced structure factor  \\ \hline
$L_{k}\left(q^2\right)$    &  Laguerre polynomials  \\[3pt] \hline
$c^m$               &  Expansion coefficients of effective potential\\ \hline
$d^n$               &  Expansion coefficients of reduced structure factor\\ \hline
$\left| \psi_{\boldsymbol{q}}\right. \rangle$               &  Single-mode approximation(SMA) model wavefunction\\ \hline
$\Gamma^{\text{2bdy}}_{m n}$              &  Two-body characteristic matrix\\ \hline \hline
$\hat{\rho}^i_{\boldsymbol{q_1}}$ &\begin{tabular}[c]{@{}c@{}} Guiding center density operator of the $i$-th electron \\ $\hat{\rho}^i_{\boldsymbol{q_1}} = e^{i q_{1a} \hat{R}^{a}_{i} }$\end{tabular}  \\[3pt] \hline
$\bar{S}_{\boldsymbol{q}_1,\boldsymbol{q}_2}$              & Reduced three-body structure factor  \\[3pt] \hline
$\hat{b}_i^{\dagger}$ \& $\hat{b}_i$     & Ladder operators   \\ \hline
$\bar{S}_{\boldsymbol{p}_1,\boldsymbol{p}_2}$              & Fourier transform pair of  reduced three-body structure factor  \\[3pt] \hline
$V_{\bm q_1 \bm q_2}$               & Three-body effective potential  \\[3pt] \hline
$\bm \tilde q_i$               & Momentum components in Jacobi coordinates  \\ \hline
$\bm \tilde p_i$               & Fourier transform pair of $\bm \tilde q_i$\\ \hline
$\tilde Q_i$ \& $\tilde P_i$                & Square of  $\bm \tilde q_i$ \& $\bm \tilde p_i$\\ \hline
$J_n (x)$    & Bessel functions of the first kind  \\[3pt] \hline
$ _0\mathbf{F}_1(;n;x)$    & Hyper-geometric functions  \\[3pt] \hline
$L^{(\alpha)}_{k}\left( x \right)$    & Generalized Laguerre polynomials  \\[3pt] \hline
$c^{m_1 m_2}$               &  \begin{tabular}[c]{@{}c@{}}Expansion coefficients of three-body effective potential\end{tabular}\\ \hline
$ d_{i}^{n_1 n_2}$        &  Expansion coefficients of reduced three-body structure factor\\ \hline
$\bar{d}^{n_1 n_2}$        &  Redefined expansion coefficients of reduced three-body structure factor\\ \hline
$\alpha^{n_1 n_2}$        &  Expansion coefficients of three-body antisymmetric wavefunctions in the basis of $| n_1, n_2 \rangle$\\ \hline
$\tilde{\Gamma}_{m_1 m_2 n_1 n_2}^{\text{3bdy}}$              &  Three-body characteristic tensor\\ \hline
$\tilde{\Gamma}^{0/+/-}_{m_1 m_2 n_1 n_2}$              & \begin{tabular}[c]{@{}c@{}} Diagonal/Off-diagonal parts of the three-body characteristic tensor\end{tabular}\\[3pt] \hline
\end{tabular}
\caption{Notations of the physical quantities and the mathematical functions used in the text.}
\label{notation_table}
\end{table}

\section{Graviton-mode energy with three-body interactions}


\begin{definition}
The generic three-body Hamiltonian in a single Landau level (LL)
\begin{equation}
\begin{aligned}
\hat{H}_{\text{3bdy}}
=&\int \frac{d^{2} \boldsymbol{q}_1 d^{2} \boldsymbol{q}_2}{(2 \pi)^4} V_{\boldsymbol{q_1, q_2}} \hat{\rho}_{\boldsymbol{q_1}} \hat{\rho}_{\boldsymbol{q_2}} \hat{\rho}_{\boldsymbol{-q_1-q_2}} -\int \frac{d^{2} \boldsymbol{q}_1 d^{2} \boldsymbol{q}_2}{(2 \pi)^4} V_{\boldsymbol{q_1, q_2}}  (\hat{\rho}_{\boldsymbol{q_1+q_2}} \hat{\rho}_{\boldsymbol{-q_1-q_2}}+ \hat{\rho}_{\boldsymbol{q_1}} \hat{\rho}_{\boldsymbol{-q_1}}+  \hat{\rho}_{\boldsymbol{q_2}} \hat{\rho}_{\boldsymbol{-q_2}})\\
&-N_e \int \frac{d^{2} \boldsymbol{q}_1 d^{2} \boldsymbol{q}_2}{(2 \pi)^4} V_{\boldsymbol{q_1, q_2}}=\sum_{i \ne j \ne k}\int \frac{d^{2} \boldsymbol{q}_1 d^{2} \boldsymbol{q}_2}{(2 \pi)^4} V_{\boldsymbol{q_1, q_2}} \hat{\rho}^i_{\boldsymbol{q_1}} \hat{\rho}^j_{\boldsymbol{q_2}} \hat{\rho}^k_{\boldsymbol{-q_1-q_2}}
\end{aligned}
\end{equation}
where $N_e$ denotes the number of electrons.
\end{definition}

\begin{definition}
The reduced three-body structure factor for the unperturbed ground state
\begin{equation}
\begin{aligned}
\bar{S}_{\boldsymbol{q}_1,\boldsymbol{q}_2}&=\sum_{i \ne j \ne k} \langle \psi_{0} \left| \hat{\rho}^i_{\boldsymbol{q_1}} \hat{\rho}^j_{\boldsymbol{q_2}} \hat{\rho}^k_{\boldsymbol{-q_1-q_2}}\right| \psi_{0}\rangle=\sum_{i \ne j \ne k}\left\langle\psi_{0}\left| e^{i q_{1a} \hat{R}^{a}_{i} } e^{i q_{2a} \hat{R}^{a}_{j} } e^{- i (q_{1a} + q_{2a})\hat{R}^{a}_{k} } \right| \psi_{0}\right\rangle
\end{aligned}
\end{equation}
\end{definition}

Then the ground state energy should be given by:
\begin{equation}
\begin{aligned}
E_0=&\sum_{i \ne j \ne k}\int \frac{d^{2} \boldsymbol{q}_1 d^{2} \boldsymbol{q}_2}{(2 \pi)^4} V_{\boldsymbol{q_1, q_2}} \langle \psi_{0} \left| \hat{\rho}^i_{\boldsymbol{q_1}} \hat{\rho}^j_{\boldsymbol{q_2}} \hat{\rho}^k_{\boldsymbol{-q_1-q_2}}\right| \psi_{0}\rangle =\int \frac{d^{2} \boldsymbol{q}_1 d^{2} \boldsymbol{q}_2}{(2 \pi)^4} V_{\boldsymbol{q_1, q_2}} \bar{S}_{\boldsymbol{q}_1,\boldsymbol{q}_2}
\end{aligned}
\end{equation}

\begin{definition}
Single-mode approximation(SMA) model wavefunction
\begin{equation}
\left| \psi_{\boldsymbol{q}}\right. \rangle = \delta \hat{\rho}_{\boldsymbol{q}} \left| \psi_{\boldsymbol{0}}\right. \rangle
\label{3benergy}
\end{equation}
which is used to describe the graviton mode in the long-wavelength limit.
\end{definition}

\begin{corollary}\label{3bcommutator}
The energy of the graviton mode can be written as
\begin{equation}
\begin{aligned}
\delta E_{\boldsymbol{q}\rightarrow 0}&= \lim_{\bm q\rightarrow 0} \frac{\langle\psi_{\boldsymbol{q}}|\hat{H}_{\text{3bdy}} \left| \psi_{\boldsymbol{q}}\right. \rangle}{\left. \langle\psi_{\boldsymbol{q}} \right| \psi_{\boldsymbol{q}}\rangle}-E_{0}=\lim_{\bm q\rightarrow 0} \frac{\left\langle\psi_{0}\left| \left[\delta \hat{\rho}_{-\boldsymbol{q}},\left[\hat{H}_{\text{3bdy}}, \delta \hat{\rho}_{\boldsymbol{q}}\right]\right] \right| \psi_{0}\right\rangle}{2 S_{\boldsymbol{q}}}\\
&=\lim_{\bm q\rightarrow 0} \sum_{i \ne j \ne k}\int \frac{d^{2} \boldsymbol{q}_1 d^{2} \boldsymbol{q}_2}{(2 \pi)^4} V_{\boldsymbol{q_1, q_2}}\times \frac{\left\langle\psi_{0}\left| \left[ \delta \hat{\rho}_{-\boldsymbol{q}},\left[\hat{\rho}^i_{\boldsymbol{q_1}} \hat{\rho}^j_{\boldsymbol{q_2}} \hat{\rho}^k_{\boldsymbol{-q_1-q_2}},  \delta \hat{\rho}_{\boldsymbol{q}}\right]\right] \right| \psi_{0}\right\rangle}{2 S_{\boldsymbol{q}}}
\end{aligned}
\end{equation}
\end{corollary}
\begin{proof}
By considering the ground state energy is $E_0 = \left\langle\psi_{0} \left|\hat{H}_{\text{3bdy}} \right| \psi_{0}\right\rangle$, we have:
\begin{small}
\begin{equation}
\begin{aligned}
\delta \tilde{E}_{\boldsymbol{q} \rightarrow 0 }&=\lim_{\bm q\rightarrow 0}\frac{\left\langle\psi_{\boldsymbol{q}} \left|\hat{H}_{\text{3bdy}} \right| \psi_{\boldsymbol{q}}\right\rangle}{\left\langle\psi_{\boldsymbol{q}} | \psi_{\boldsymbol{q}}\right\rangle}-E_{0}=\lim_{\bm q\rightarrow 0}\frac{2\left\langle\psi_{\boldsymbol{q}}\left|\hat{H}_{\text{3bdy}} \right| \psi_{\boldsymbol{q}}\right\rangle-2 E_0 \left\langle\psi_{\boldsymbol{q}} | \psi_{\boldsymbol{q}}\right\rangle}{2 \left\langle\psi_{\boldsymbol{q}} | \psi_{\boldsymbol{q}}\right\rangle}\\
&=\lim_{\bm q\rightarrow 0}\frac{\left\langle\psi_{\boldsymbol{-q}} \left|\hat{H}_{\text{3bdy}} \right| \psi_{\boldsymbol{q}}\right\rangle+\left\langle\psi_{\boldsymbol{q}} \left|\hat{H}_{\text{3bdy}} \right| \psi_{\boldsymbol{-q}}\right\rangle- E_0 \left\langle\psi_{\boldsymbol{-q}} | \psi_{\boldsymbol{q}}\right\rangle- E_0 \left\langle\psi_{\boldsymbol{q}} | \psi_{\boldsymbol{-q}}\right\rangle}{2 \left\langle\psi_{\boldsymbol{q}} | \psi_{\boldsymbol{q}}\right\rangle}\\
&=\lim_{\bm q\rightarrow 0}\frac{\left\langle\psi_{0}\left|\delta \hat{\rho}_{-\boldsymbol{q}} \hat{H}_{\text{3bdy}} \delta \hat{\rho}_{\boldsymbol{q}}\right|\psi_{0}\right\rangle+\left\langle\psi_{0} \left|\delta \hat{\rho}_{\boldsymbol{q}}\hat{H}_{\text{3bdy}}\delta \hat{\rho}_{-\boldsymbol{q}} \right| \psi_{0}\right\rangle- \left\langle\psi_{0} \left|\delta \hat{\rho}_{-\boldsymbol{q}}\delta \hat{\rho}_{\boldsymbol{q}} \hat{H}_{\text{3bdy}} \right|\psi_{0}\right\rangle-  \left\langle\psi_{0} \left|\hat{H}_{\text{3bdy}}\delta \hat{\rho}_{\boldsymbol{q}}\delta \hat{\rho}_{-\boldsymbol{q}} \right| \psi_{0}\right\rangle}{2 \left\langle\psi_{\boldsymbol{q}} | \psi_{\boldsymbol{q}}\right\rangle}\\
&=\lim_{\bm q\rightarrow 0}\frac{\left\langle\psi_{0} \left| \left(\delta \hat{\rho}_{-\boldsymbol{q}} \hat{H}_{\text{3bdyy}} \delta \hat{\rho}_{\boldsymbol{q}}+\delta \hat{\rho}_{\boldsymbol{q}}\hat{H}_{\text{3bdy}}\delta \hat{\rho}_{-\boldsymbol{q}}-  \delta \hat{\rho}_{-\boldsymbol{q}}\delta \hat{\rho}_{\boldsymbol{q}}\hat{H}_{\text{3bdy}} - \hat{H}_{\text{3bdy}} \delta \hat{\rho}_{\boldsymbol{q}}\delta \hat{\rho}_{-\boldsymbol{q}} \right) \right| \psi_{0}\right\rangle}{2 \left\langle\psi_{\boldsymbol{q}} | \psi_{\boldsymbol{q}}\right\rangle}\\
&=\lim_{\bm q\rightarrow 0}\frac{\left\langle\psi_{0}\left|\left[\delta \hat{\rho}_{-\boldsymbol{q}},\left[\hat{H}_{\text{3bdy}}, \delta \hat{\rho}_{\boldsymbol{q}}\right]\right]\right| \psi_{0}\right\rangle}{2 S_{\boldsymbol{q}}} \qedhere
\end{aligned}
\end{equation}
\end{small}
\end{proof}

Note that \textbf{the particle indices have no effect on the commutation rules of the density operators} so for simplicity we will omit the index $i$, $j$ and $k$ in the following derivations. Firstly consider the commutator:
\begin{equation}
\begin{aligned}
\left[\hat E, \left[\hat A \hat B \hat C, \hat D \right]\right]&= \left[ \hat E, \hat A \hat B \left[\hat C, \hat D \right] + \hat A \left[\hat B, \hat D \right] \hat C + \left[\hat A, \hat D \right] \hat B \hat C \right]\\
&=\left[ \hat E, \hat A \hat B \left[ \hat C, \hat D \right] \right] + \left[ \hat E, \hat A \left[ \hat B, \hat D \right] \hat C \right]+\left[ \hat E, \left[ \hat A, \hat D \right] \hat B \hat C \right]
\end{aligned}
\end{equation}
by taking:
\begin{equation}
\begin{aligned}
E=\delta \hat{\rho}_{-\boldsymbol{q}},\quad A=\delta \hat{\rho}_{\boldsymbol{q}_{1}}, \quad B=\delta \hat{\rho}_{\boldsymbol{q}_{2}} \quad C=\delta \hat{\rho}_{-\boldsymbol{q}_{1}-\boldsymbol{q}_{2}}, \quad D=\delta \hat{\rho}_{\boldsymbol{q}}
\end{aligned}
\end{equation}
we can decompose the commutator in Eq.\ref{3benergy} into:
\begin{equation}
\begin{aligned}
&\left[\delta \hat{\rho}_{-\boldsymbol{q}},[\delta \hat{\rho}_{\boldsymbol{q}_1}\delta \hat{\rho}_{\boldsymbol{q}_2}\delta \hat{\rho}_{-\boldsymbol{q}_1-\boldsymbol{q}_2},\delta \hat{\rho}_{\boldsymbol{q}}]\right]\\
=&[\delta \hat{\rho}_{-\boldsymbol{q}},\delta \hat{\rho}_{\boldsymbol{q}_1}\delta \hat{\rho}_{\boldsymbol{q}_2}[\delta \hat{\rho}_{-\boldsymbol{q}_1-\boldsymbol{q}_2},\delta \hat{\rho}_{\boldsymbol{q}}]+\delta \hat{\rho}_{\boldsymbol{q}_1}[\delta \hat{\rho}_{\boldsymbol{q}_2},\delta \hat{\rho}_{\boldsymbol{q}}]\delta \hat{\rho}_{-\boldsymbol{q}_1-\boldsymbol{q}_2}+[\delta \hat{\rho}_{\boldsymbol{q}_1},\delta \hat{\rho}_{\boldsymbol{q}}]\delta \hat{\rho}_{\boldsymbol{q}_2}\delta \hat{\rho}_{-\boldsymbol{q}_1-\boldsymbol{q}_2}]\\
=&[\delta \hat{\rho}_{-\boldsymbol{q}},\delta \hat{\rho}_{\boldsymbol{q}_1}\delta \hat{\rho}_{\boldsymbol{q}_2}[\delta \hat{\rho}_{-\boldsymbol{q}_1-\boldsymbol{q}_2},\delta \hat{\rho}_{\boldsymbol{q}}]]+[\delta \hat{\rho}_{-\boldsymbol{q}},\delta \hat{\rho}_{\boldsymbol{q}_1}[\delta \hat{\rho}_{\boldsymbol{q}_2},\delta \hat{\rho}_{\boldsymbol{q}}]\delta \hat{\rho}_{-\boldsymbol{q}_1-\boldsymbol{q}_2}]+[\delta \hat{\rho}_{-\boldsymbol{q}},[\delta \hat{\rho}_{\boldsymbol{q}_1},\delta \hat{\rho}_{\boldsymbol{q}}]\delta \hat{\rho}_{\boldsymbol{q}_2}\delta \hat{\rho}_{-\boldsymbol{q}_1-\boldsymbol{q}_2}]\\
=&2 i \sin \frac{(-\boldsymbol{q}_{1}-\boldsymbol{q}_{2}) \times \boldsymbol{q}}{2}[\delta \hat{\rho}_{-\boldsymbol{q}},\delta \hat{\rho}_{\boldsymbol{q}_1}\delta \hat{\rho}_{\boldsymbol{q}_2}  \delta \hat{\rho}_{-\boldsymbol{q}_{1}-\boldsymbol{q}_{2}+\boldsymbol{q}}]+2 i \sin \frac{\boldsymbol{q}_{2} \times \boldsymbol{q}}{2}[\delta \hat{\rho}_{-\boldsymbol{q}},\delta \hat{\rho}_{\boldsymbol{q}_1}\delta \hat{\rho}_{\boldsymbol{q}_{2}+\boldsymbol{q}}\delta \hat{\rho}_{-\boldsymbol{q}_1-\boldsymbol{q}_2}]\\
&+2 i \sin \frac{\boldsymbol{q}_{1} \times \boldsymbol{q}}{2}[\delta \hat{\rho}_{-\boldsymbol{q}},\delta \hat{\rho}_{\boldsymbol{q}_{1}+\boldsymbol{q}}\delta \hat{\rho}_{\boldsymbol{q}_2}\delta \hat{\rho}_{-\boldsymbol{q}_1-\boldsymbol{q}_2}]
\label{commutator3b}
\end{aligned}
\end{equation}
where we have considered the following lemmas:
\begin{lemma}\label{gcocomm}
Commutation rule of the guiding center operators \cite{ezawa2008quantum}
\begin{equation}
\left[\hat R_i^a,\hat R_j^b \right]=-i\epsilon^{ab}\delta_{ij} 
\end{equation}
\end{lemma}

\begin{lemma}\label{BCH}
Baker–Campbell–Hausdorff formula \cite{rossmann2006lie}
\begin{equation}
e^{\hat{X}}e^{\hat{Y}} = e^{\hat{X}+\hat{Y}+\frac{1}{2}[\hat{X}, \hat{Y}]+\frac{1}{12}[\hat{X},[\hat{X}, \hat{Y}]]-\frac{1}{12}[\hat{Y},[\hat{X}, \hat{Y}]]+\cdots}
\end{equation}
\end{lemma}

\begin{corollary}\label{GMP}
Girvin-MacDonal-Platzman(GMP) algebra 
\begin{equation}
\left[\delta \hat{\rho}_{\boldsymbol{q}_{1}}, \delta \hat{\rho}_{\boldsymbol{q}_{2}}\right]  = 2 i \sin \frac{\boldsymbol{q}_{1} \times \boldsymbol{q}_{2}}{2} \delta \hat{\rho}_{\boldsymbol{q}_{1}+\boldsymbol{q}_{2}}
\end{equation}
\end{corollary}
\begin{proof}
Combining Lemma.\ref{gcocomm} and Lemma.\ref{BCH} can easily get this result, firstly proposed in Ref.\cite{girvin1986magneto}.
\end{proof}

Then by considering the next commutator:
\begin{equation}
\left[\hat A, \hat B \hat C \hat D \right]=\left[ \hat A, \hat B \right] \hat C \hat D + \hat B \left[ \hat A, \hat C \right] \hat D+ \hat B \hat C \left[ \hat A, \hat D \right]
\end{equation}
and substituting the corresponding operators we can write down:
\begin{equation}
\begin{aligned}
&\left[\delta \hat{\rho}_{-\boldsymbol{q}},\delta \hat{\rho}_{\boldsymbol{q}_1}\delta \hat{\rho}_{\boldsymbol{q}_2}  \delta \hat{\rho}_{-\boldsymbol{q}_{1}-\boldsymbol{q}_{2}+\boldsymbol{q}}\right]\\
=& \left[\delta \hat{\rho}_{-\boldsymbol{q}},\delta \hat{\rho}_{\boldsymbol{q}_1} \right]\delta \hat{\rho}_{\boldsymbol{q}_2}\delta \hat{\rho}_{-\boldsymbol{q}_{1}-\boldsymbol{q}_{2}+\boldsymbol{q}}+\delta \hat{\rho}_{\boldsymbol{q}_1}\left[\delta \hat{\rho}_{-\boldsymbol{q}},\delta \hat{\rho}_{\boldsymbol{q}_2} \right]\delta \hat{\rho}_{-\boldsymbol{q}_{1}-\boldsymbol{q}_{2}+\boldsymbol{q}}+\delta \hat{\rho}_{\boldsymbol{q}_1}\delta \hat{\rho}_{\boldsymbol{q}_2}[\delta \hat{\rho}_{-\boldsymbol{q}},\delta \hat{\rho}_{-\boldsymbol{q}_{1}-\boldsymbol{q}_{2}+\boldsymbol{q}}]\\
=&2 i \sin \frac{-\boldsymbol{q} \times \boldsymbol{q}_{1}}{2} \delta \hat{\rho}_{-\boldsymbol{q}+\boldsymbol{q}_1}\delta \hat{\rho}_{\boldsymbol{q}_2}\delta \hat{\rho}_{-\boldsymbol{q}_{1}-\boldsymbol{q}_{2}+\boldsymbol{q}}+2 i \sin \frac{-\boldsymbol{q} \times \boldsymbol{q}_{2}}{2} \delta \hat{\rho}_{\boldsymbol{q}_1}\delta \hat{\rho}_{-\boldsymbol{q}+\boldsymbol{q}_2}\delta \hat{\rho}_{-\boldsymbol{q}_{1}-\boldsymbol{q}_{2}+\boldsymbol{q}}\\
&+2 i \sin \frac{-\boldsymbol{q} \times (-\boldsymbol{q}_1-\boldsymbol{q}_2+\boldsymbol{q})}{2}\delta \hat{\rho}_{\boldsymbol{q}_1}\delta \hat{\rho}_{\boldsymbol{q}_2}  \delta \hat{\rho}_{-\boldsymbol{q}_1-\boldsymbol{q}_2}
\end{aligned}
\end{equation}
and
\begin{equation}
\begin{aligned}
&\left[\delta \hat{\rho}_{-\boldsymbol{q}},\delta \hat{\rho}_{\boldsymbol{q}_1}\delta \hat{\rho}_{\boldsymbol{q}_{2}+\boldsymbol{q}}\delta \hat{\rho}_{-\boldsymbol{q}_1-\boldsymbol{q}_2}\right]\\
=&[\delta \hat{\rho}_{-\boldsymbol{q}},\delta \hat{\rho}_{\boldsymbol{q}_1}]\delta \hat{\rho}_{\boldsymbol{q}_{2}+\boldsymbol{q}}\delta \hat{\rho}_{-\boldsymbol{q}_1-\boldsymbol{q}_2}+\delta \hat{\rho}_{\boldsymbol{q}_1}[\delta \hat{\rho}_{-\boldsymbol{q}},\delta \hat{\rho}_{\boldsymbol{q}_{2}+\boldsymbol{q}}]\delta \hat{\rho}_{-\boldsymbol{q}_1-\boldsymbol{q}_2} +\delta \hat{\rho}_{\boldsymbol{q}_1}\delta \hat{\rho}_{\boldsymbol{q}_{2}+\boldsymbol{q}}[\delta \hat{\rho}_{-\boldsymbol{q}},\delta \hat{\rho}_{-\boldsymbol{q}_1-\boldsymbol{q}_2}]\\
=&2 i \sin \frac{-\boldsymbol{q} \times \boldsymbol{q}_{1}}{2} \delta \hat{\rho}_{-\boldsymbol{q}+\boldsymbol{q}_1}\delta \hat{\rho}_{\boldsymbol{q}_{2}+\boldsymbol{q}}\delta \hat{\rho}_{-\boldsymbol{q}_1-\boldsymbol{q}_2}+2 i \sin \frac{-\boldsymbol{q} \times (\boldsymbol{q}_2+\boldsymbol{q})}{2} \delta \hat{\rho}_{\boldsymbol{q}_1}\delta \hat{\rho}_{\boldsymbol{q}_2}\delta \hat{\rho}_{-\boldsymbol{q}_1-\boldsymbol{q}_2}\\
&+2 i \sin \frac{-\boldsymbol{q} \times (-\boldsymbol{q}_1-\boldsymbol{q}_2)}{2}\delta \hat{\rho}_{\boldsymbol{q}_1}\delta \hat{\rho}_{\boldsymbol{q}_{2}+\boldsymbol{q}} \delta \hat{\rho}_{-\boldsymbol{q}-\boldsymbol{q}_1-\boldsymbol{q}_2}
\end{aligned}
\end{equation}
and
\begin{equation}
\begin{aligned}
&\left[\delta \hat{\rho}_{-\boldsymbol{q}},\delta \hat{\rho}_{\boldsymbol{q}_{1}+\boldsymbol{q}}\delta \hat{\rho}_{\boldsymbol{q}_2}\delta \hat{\rho}_{-\boldsymbol{q}_1-\boldsymbol{q}_2}\right]\\
=&[\delta \hat{\rho}_{-\boldsymbol{q}},\delta \hat{\rho}_{\boldsymbol{q}_{1}+\boldsymbol{q}}]\delta \hat{\rho}_{\boldsymbol{q}_2}\delta \hat{\rho}_{-\boldsymbol{q}_1-\boldsymbol{q}_2} +\delta \hat{\rho}_{\boldsymbol{q}_{1}+\boldsymbol{q}}[\delta \hat{\rho}_{-\boldsymbol{q}},\delta \hat{\rho}_{\boldsymbol{q}_2}]\delta \hat{\rho}_{-\boldsymbol{q}_1-\boldsymbol{q}_2}+\delta \hat{\rho}_{\boldsymbol{q}_{1}+\boldsymbol{q}}\delta \hat{\rho}_{\boldsymbol{q}_2}[\delta \hat{\rho}_{-\boldsymbol{q}},\delta \hat{\rho}_{-\boldsymbol{q}_1-\boldsymbol{q}_2}]\\
=&2 i \sin \frac{-\boldsymbol{q} \times (\boldsymbol{q}_1+\boldsymbol{q})}{2} \delta \hat{\rho}_{\boldsymbol{q}_1}\delta \hat{\rho}_{\boldsymbol{q}_2}\delta \hat{\rho}_{-\boldsymbol{q}_1-\boldsymbol{q}_2}+2 i \sin \frac{-\boldsymbol{q} \times \boldsymbol{q}_{2}}{2}\delta \hat{\rho}_{\boldsymbol{q}_{1}+\boldsymbol{q}} \delta \hat{\rho}_{-\boldsymbol{q}+\boldsymbol{q}_2}\delta \hat{\rho}_{-\boldsymbol{q}_1-\boldsymbol{q}_2}\\
&+2 i \sin \frac{-\boldsymbol{q} \times (-\boldsymbol{q}_1-\boldsymbol{q}_2)}{2}\delta \hat{\rho}_{\boldsymbol{q}_{1}+\boldsymbol{q}}\delta \hat{\rho}_{\boldsymbol{q}_2} \delta \hat{\rho}_{-\boldsymbol{q}-\boldsymbol{q}_1-\boldsymbol{q}_2}
\end{aligned}
\end{equation}
where we have used the GMP algebra in Corollary.\ref{GMP}.
Thus Eq.\ref{commutator3b} could be written as:
\begin{small}
\begin{equation}
\begin{aligned}
&\left[\delta \hat{\rho}_{-\boldsymbol{q}},[\delta \hat{\rho}_{\boldsymbol{q}_1}\delta \hat{\rho}_{\boldsymbol{q}_2}\delta \hat{\rho}_{-\boldsymbol{q}_1-\boldsymbol{q}_2},\delta \hat{\rho}_{\boldsymbol{q}}]\right]\\
=&-4 \sin \frac{(-\boldsymbol{q}_{1}-\boldsymbol{q}_{2}) \times \boldsymbol{q}}{2} \left[ \sin \frac{-\boldsymbol{q} \times \boldsymbol{q}_{1}}{2} \bar{S}_{-\boldsymbol{q}+\boldsymbol{q}_1,\boldsymbol{q}_2} + \sin \frac{-\boldsymbol{q} \times \boldsymbol{q}_{2}}{2} \bar{S}_{\boldsymbol{q}_1,-\boldsymbol{q}+\boldsymbol{q}_2} + \sin \frac{-\boldsymbol{q} \times (-\boldsymbol{q}_1-\boldsymbol{q}_2+\boldsymbol{q})}{2} \bar{S}_{\boldsymbol{q}_1,\boldsymbol{q}_2}\right]\\
&-4 \sin \frac{\boldsymbol{q}_{2} \times \boldsymbol{q}}{2} \left[ \sin \frac{-\boldsymbol{q} \times \boldsymbol{q}_{1}}{2} \bar{S}_{-\boldsymbol{q}+\boldsymbol{q}_1,\boldsymbol{q}_2+\boldsymbol{q}} + \sin \frac{-\boldsymbol{q} \times (\boldsymbol{q}_2+\boldsymbol{q})}{2} \bar{S}_{\boldsymbol{q}_1,\boldsymbol{q}_2} + \sin \frac{-\boldsymbol{q} \times (-\boldsymbol{q}_1-\boldsymbol{q}_2)}{2} \bar{S}_{\boldsymbol{q}_1,\boldsymbol{q}_2
+\boldsymbol{q}} \right]\\
&-4 \sin \frac{\boldsymbol{q}_{1} \times \boldsymbol{q}}{2} \left[ \sin \frac{-\boldsymbol{q} \times (\boldsymbol{q}_1+\boldsymbol{q})}{2} \bar{S}_{\boldsymbol{q}_1,\boldsymbol{q}_2} + \sin \frac{-\boldsymbol{q} \times \boldsymbol{q}_{2}}{2} \bar{S}_{\boldsymbol{q}_{1}+\boldsymbol{q}, -\boldsymbol{q}+\boldsymbol{q}_2} + \sin \frac{-\boldsymbol{q} \times (-\boldsymbol{q}_1-\boldsymbol{q}_2)}{2} \bar{S}_{\boldsymbol{q}_{1}+\boldsymbol{q},\boldsymbol{q}_2} \right]\\
=&-4 \left[\sin \frac{(\boldsymbol{q}_{1}+\boldsymbol{q}_{2}) \times \boldsymbol{q}}{2}\sin \frac{ (\boldsymbol{q}_1+\boldsymbol{q}_2-\boldsymbol{q})\times \boldsymbol{q}}{2} +\sin \frac{\boldsymbol{q}_{2} \times \boldsymbol{q}}{2}  \sin \frac{(\boldsymbol{q}_2+\boldsymbol{q}) \times \boldsymbol{q}}{2} +\sin \frac{\boldsymbol{q}_{1} \times \boldsymbol{q}}{2} \sin \frac{(\boldsymbol{q}_1+\boldsymbol{q}) \times \boldsymbol{q} }{2} \right] \bar{S}_{\boldsymbol{q}_1,\boldsymbol{q}_2}\\
&-4 \sin \frac{ \boldsymbol{q} \times (\boldsymbol{q}_{1}+\boldsymbol{q}_{2})}{2} \left[ \sin \frac{\boldsymbol{q}_{1} \times \boldsymbol{q}}{2} (\bar{S}_{\boldsymbol{q}_1-\boldsymbol{q},\boldsymbol{q}_2} +\bar{S}_{\boldsymbol{q}_{1}+\boldsymbol{q},\boldsymbol{q}_2})+ \sin \frac{\boldsymbol{q}_{2} \times \boldsymbol{q}}{2}( \bar{S}_{\boldsymbol{q}_1,\boldsymbol{q}_2-\boldsymbol{q}}+\bar{S}_{\boldsymbol{q}_1,\boldsymbol{q}_2
+\boldsymbol{q}}) \right]\\
&-4 \sin \frac{\boldsymbol{q}_{2} \times \boldsymbol{q}}{2} \sin \frac{\boldsymbol{q}_{1} \times \boldsymbol{q}}{2} \left(\bar{S}_{\boldsymbol{q}_1-\boldsymbol{q},\boldsymbol{q}_2+\boldsymbol{q}}+ \bar{S}_{\boldsymbol{q}_{1}+\boldsymbol{q}, \boldsymbol{q}_2-\boldsymbol{q}}\right) \equiv \mathcal S_1 + \mathcal S_2 + \mathcal S_3
\label{commutator3b2}
\end{aligned}
\end{equation}
\end{small}

Considering the well-known trigonometric identities:
\begin{equation}
\sin \alpha \sin(\alpha + \beta) = -\frac{1}{2} [\cos(2 \alpha + \beta) - \cos \beta]; \quad \cos x = 1 - 2 \sin^2 \frac{x}{2}
\end{equation}
we have:
\begin{equation}
\begin{aligned}
\mathcal S_1 = &2 \left[\cos \left((\boldsymbol{q}_{1}+\boldsymbol{q}_{2} - \frac{\boldsymbol{q}}{2} ) \times \boldsymbol{q}\right) +\cos \left( (\boldsymbol{q}_{2} + \frac{\boldsymbol{q}}{2}) \times \boldsymbol{q} \right)+\cos \left( (\boldsymbol{q}_{1} + \frac{\boldsymbol{q}}{2}) \times \boldsymbol{q} \right)- 3 \cos \left(\frac{\boldsymbol{q} \times \boldsymbol{q}}{2} \right) \right] \bar{S}_{\boldsymbol{q}_1,\boldsymbol{q}_2} \\
&+2 \left[\cos \left(\frac{\boldsymbol{q}_2 \times \boldsymbol{q}}{2} \right) - \cos \left( (\boldsymbol{q}_1+\frac{\boldsymbol{q}_2}{2}) \times \boldsymbol{q} \right) \right] (\bar{S}_{\boldsymbol{q}_1-\boldsymbol{q},\boldsymbol{q}_2} +\bar{S}_{\boldsymbol{q}_{1}+\boldsymbol{q},\boldsymbol{q}_2}) \\
=&2 \left[\cos \left((\boldsymbol{q}_{1}+\boldsymbol{q}_{2}) \times \boldsymbol{q} \right) +\cos \left( \boldsymbol{q}_{2}  \times \boldsymbol{q} \right)+\cos (\boldsymbol{q}_{1}  \times \boldsymbol{q})- 3 \right] \bar{S}_{\boldsymbol{q}_1,\boldsymbol{q}_2} \\
=&-4 \left[\sin^2 \left(\frac{1}{2}(\boldsymbol{q}_{1}+\boldsymbol{q}_{2}) \times \boldsymbol{q}\right) +\sin^2 \left(\frac{1}{2} \boldsymbol{q}_{2}  \times \boldsymbol{q} \right)+\sin^2 \left(\frac{1}{2}\boldsymbol{q}_{1}  \times \boldsymbol{q} \right) \right] \bar{S}_{\boldsymbol{q}_1,\boldsymbol{q}_2} \\
=&- [((\boldsymbol{q}_{1}+\boldsymbol{q}_{2}) \times \boldsymbol{q})^2 +( \boldsymbol{q}_{2}  \times \boldsymbol{q})^2+(\boldsymbol{q}_{1}  \times \boldsymbol{q})^2 +O(\boldsymbol{q}^3)] \bar{S}_{\boldsymbol{q}_1,\boldsymbol{q}_2} \\
=&-2 [(\boldsymbol{q}_1 \times \boldsymbol{q}) \cdot (\boldsymbol{q}_2 \times \boldsymbol{q}) +( \boldsymbol{q}_{2}  \times \boldsymbol{q})^2+(\boldsymbol{q}_{1}  \times \boldsymbol{q})^2 +O(\boldsymbol{q}^3)] \bar{S}_{\boldsymbol{q}_1,\boldsymbol{q}_2} 
\end{aligned}
\end{equation}
and
\begin{equation}
\begin{aligned}
\mathcal S_2
=&2 \left[\cos \left(\frac{\boldsymbol{q}_2 }{2}\times \boldsymbol{q} \right) - \cos \left( (\boldsymbol{q}_1+\frac{\boldsymbol{q}_2}{2} ) \times \boldsymbol{q} \right)\right] (\bar{S}_{\boldsymbol{q}_1-\boldsymbol{q},\boldsymbol{q}_2} +\bar{S}_{\boldsymbol{q}_{1}+\boldsymbol{q},\boldsymbol{q}_2}) +2 \left[ \boldsymbol{q}_1 \leftrightarrow \boldsymbol{q}_2 \right] ( \bar{S}_{\boldsymbol{q}_1,\boldsymbol{q}_2-\boldsymbol{q}}+\bar{S}_{\boldsymbol{q}_1,\boldsymbol{q}_2 +\boldsymbol{q}}) \\
=&-4 \left[\sin^2 \left(\frac{1}{2}\frac{\boldsymbol{q}_2 }{2}\times \boldsymbol{q} \right) - \sin^2 \left(\frac{1}{2} (\boldsymbol{q}_1+\frac{\boldsymbol{q}_2}{2}) \times \boldsymbol{q} \right)\right] (\bar{S}_{\boldsymbol{q}_1-\boldsymbol{q},\boldsymbol{q}_2} +\bar{S}_{\boldsymbol{q}_{1}+\boldsymbol{q},\boldsymbol{q}_2})-4 \left[ \boldsymbol{q}_1 \leftrightarrow \boldsymbol{q}_2 \right]( \bar{S}_{\boldsymbol{q}_1,\boldsymbol{q}_2-\boldsymbol{q}}+\bar{S}_{\boldsymbol{q}_1,\boldsymbol{q}_2 +\boldsymbol{q}}) \\
=& \left[ \left( (\boldsymbol{q}_1+\frac{\boldsymbol{q}_2}{2}) \times \boldsymbol{q} \right)^2  - \left(\frac{1}{2}\boldsymbol{q}_2 \times \boldsymbol{q} \right)^2 +O(\boldsymbol{q}^3)\right] (\bar{S}_{\boldsymbol{q}_1-\boldsymbol{q},\boldsymbol{q}_2} +\bar{S}_{\boldsymbol{q}_{1}+\boldsymbol{q},\boldsymbol{q}_2}) + \left[ \boldsymbol{q}_1 \leftrightarrow \boldsymbol{q}_2 \right]( \bar{S}_{\boldsymbol{q}_1,\boldsymbol{q}_2-\boldsymbol{q}}+\bar{S}_{\boldsymbol{q}_1,\boldsymbol{q}_2 +\boldsymbol{q}}) \\
=& \left[(\boldsymbol{q}_1\times \boldsymbol{q} )^2  +(\boldsymbol{q}_1 \times \boldsymbol{q}) \cdot (\boldsymbol{q}_2 \times \boldsymbol{q}) +O(\boldsymbol{q}^3) \right] (\bar{S}_{\boldsymbol{q}_1-\boldsymbol{q},\boldsymbol{q}_2} +\bar{S}_{\boldsymbol{q}_{1}+\boldsymbol{q},\boldsymbol{q}_2}) \\
&+ [(\boldsymbol{q}_2\times \boldsymbol{q} )^2  +(\boldsymbol{q}_1 \times \boldsymbol{q}) \cdot (\boldsymbol{q}_2 \times \boldsymbol{q})+O(\boldsymbol{q}^3)]( \bar{S}_{\boldsymbol{q}_1,\boldsymbol{q}_2-\boldsymbol{q}}+\bar{S}_{\boldsymbol{q}_1,\boldsymbol{q}_2+\boldsymbol{q}}) 
\end{aligned}
\end{equation}
and
\begin{equation}
\begin{aligned}
\mathcal S_3
=&2 \left[\cos \left(\frac{\boldsymbol{q}_1 +\boldsymbol{q}_2}{2}\times \boldsymbol{q} \right) - \cos \left( \frac{\boldsymbol{q}_1 -\boldsymbol{q}_2}{2}\times \boldsymbol{q} \right)\right] (\bar{S}_{\boldsymbol{q}_1-\boldsymbol{q},\boldsymbol{q}_2+\boldsymbol{q}}+ \bar{S}_{\boldsymbol{q}_{1}+\boldsymbol{q}, \boldsymbol{q}_2-\boldsymbol{q}}) \\
=&-4 \left[\sin^2 \left(\frac{1}{2}\frac{\boldsymbol{q}_1 +\boldsymbol{q}_2}{2}\times \boldsymbol{q} \right) - \sin^2 \left(\frac{1}{2}\frac{\boldsymbol{q}_1 -\boldsymbol{q}_2}{2}\times \boldsymbol{q} \right)\right] (\bar{S}_{\boldsymbol{q}_1-\boldsymbol{q},\boldsymbol{q}_2+\boldsymbol{q}}+ \bar{S}_{\boldsymbol{q}_{1}+\boldsymbol{q}, \boldsymbol{q}_2-\boldsymbol{q}}) \\
=& \left[ \left(\frac{1}{2}(\boldsymbol{q}_1 -\boldsymbol{q}_2)\times \boldsymbol{q} \right)^2 -  \left(\frac{1}{2}(\boldsymbol{q}_1 +\boldsymbol{q}_2)\times \boldsymbol{q} \right)^2+O(\boldsymbol{q}^3) \right] (\bar{S}_{\boldsymbol{q}_1-\boldsymbol{q},\boldsymbol{q}_2+\boldsymbol{q}}+ \bar{S}_{\boldsymbol{q}_{1}+\boldsymbol{q}, \boldsymbol{q}_2-\boldsymbol{q}})\\
=&- [(\boldsymbol{q}_1 \times \boldsymbol{q}) \cdot (\boldsymbol{q}_2 \times \boldsymbol{q})+O(\boldsymbol{q}^3)] (\bar{S}_{\boldsymbol{q}_1-\boldsymbol{q},\boldsymbol{q}_2+\boldsymbol{q}}+ \bar{S}S_{\boldsymbol{q}_{1}+\boldsymbol{q}, \boldsymbol{q}_2-\boldsymbol{q}}) 
\label{commutator3b3}
\end{aligned}
\end{equation}
or we can also directly throw away $\boldsymbol{q} \times \boldsymbol{q} = 0$ in the first place. Here we have also considered:
\begin{equation}
[(\boldsymbol{x} +\boldsymbol{y})\times \boldsymbol{z}]^2 = (\boldsymbol{x} \times \boldsymbol{z})^2+(\boldsymbol{y}\times \boldsymbol{z})^2+ 2  (\boldsymbol{x} \times \boldsymbol{z})\cdot (\boldsymbol{y}\times \boldsymbol{z})
\end{equation}
and the limit:
\begin{equation}
\lim_{\bm q \rightarrow 0} \left[2 \sin \left(\frac{1}{2} \boldsymbol{q}^{\prime} \times \boldsymbol{q}\right)\right]^{2}=(\boldsymbol{q}^{\prime} \times \boldsymbol{q})^2+O(\boldsymbol{q}^3)
\end{equation}

Then we expand the structure factors to the second order, which gives:
\begin{equation}
\begin{aligned}
&\bar{S}_{\boldsymbol{q}_1-\boldsymbol{q},\boldsymbol{q}_2} +\bar{S}_{\boldsymbol{q}_{1}+\boldsymbol{q},\boldsymbol{q}_2}\\
=&\bar{S}_{\boldsymbol{q}_1,\boldsymbol{q}_2} -\boldsymbol{q}_x \frac{\partial \bar{S}_{\boldsymbol{q}_1,\boldsymbol{q}_2}}{\partial \boldsymbol{q}_{1x}}-\boldsymbol{q}_y \frac{\partial \bar{S}_{\boldsymbol{q}_1,\boldsymbol{q}_2}}{\partial \boldsymbol{q}_{1y}}+\frac{1}{2!} \left(\boldsymbol{q}^2_x \frac{\partial^2 \bar{S}_{\boldsymbol{q}_1,\boldsymbol{q}_2}}{\partial \boldsymbol{q}^{2}_{1x}}+2\boldsymbol{q}_x \boldsymbol{q}_y \frac{\partial^2 \bar{S}_{\boldsymbol{q}_1,\boldsymbol{q}_2}}{\partial \boldsymbol{q}_{1x} \partial \boldsymbol{q}_{1y}}+\boldsymbol{q}^2_y \frac{\partial^2 \bar{S}_{\boldsymbol{q}_1,\boldsymbol{q}_2}}{\partial \boldsymbol{q}^{2}_{1y}} \right)+O(\boldsymbol{q}^3)\\
&+\bar{S}_{\boldsymbol{q}_1,\boldsymbol{q}_2}+\boldsymbol{q}_x \frac{\partial \bar{S}_{\boldsymbol{q}_1,\boldsymbol{q}_2}}{\partial \boldsymbol{q}_{1x}}+\boldsymbol{q}_y \frac{\partial \bar{S}_{\boldsymbol{q}_1,\boldsymbol{q}_2}}{\partial \boldsymbol{q}_{1y}}+\frac{1}{2!}\left(\boldsymbol{q}^2_x \frac{\partial^2 \bar{S}_{\boldsymbol{q}_1,\boldsymbol{q}_2}}{\partial \boldsymbol{q}^{2}_{1x}}+2\boldsymbol{q}_x \boldsymbol{q}_y \frac{\partial^2 \bar{S}_{\boldsymbol{q}_1,\boldsymbol{q}_2}}{\partial \boldsymbol{q}_{1x} \partial \boldsymbol{q}_{1y}}+\boldsymbol{q}^2_y \frac{\partial^2 \bar{S}_{\boldsymbol{q}_1,\boldsymbol{q}_2}}{\partial \boldsymbol{q}^{2}_{1y}} \right)+O(\boldsymbol{q}^3)\\
=&2 \bar{S}_{\boldsymbol{q}_1,\boldsymbol{q}_2}+ \left(\boldsymbol{q}_x \frac{\partial}{\partial \boldsymbol{q}_{1x}}+\boldsymbol{q}_y \frac{\partial}{\partial \boldsymbol{q}_{1y}}\right)^2 \bar{S}_{\boldsymbol{q}_1,\boldsymbol{q}_2}+O(\boldsymbol{q}^3) \sim [2 + (\boldsymbol{q} \cdot \boldsymbol{\nabla}_1)^2] \bar{S}_{\boldsymbol{q}_1,\boldsymbol{q}_2}
\end{aligned}
\end{equation}
and
\begin{equation}
\begin{aligned}
&\bar{S}_{\boldsymbol{q}_1,\boldsymbol{q}_2-\boldsymbol{q}} +\bar{S}_{\boldsymbol{q}_{1},\boldsymbol{q}_2+\boldsymbol{q}}\\
=& \bar{S}_{\boldsymbol{q}_1,\boldsymbol{q}_2}-\boldsymbol{q}_x \frac{\partial \bar{S}_{\boldsymbol{q}_1,\boldsymbol{q}_2}}{\partial \boldsymbol{q}_{2x}}-\boldsymbol{q}_y \frac{\partial \bar{S}_{\boldsymbol{q}_1,\boldsymbol{q}_2}}{\partial \boldsymbol{q}_{2y}}+\frac{1}{2!} \left(\boldsymbol{q}^2_x \frac{\partial^2 \bar{S}_{\boldsymbol{q}_1,\boldsymbol{q}_2}}{\partial \boldsymbol{q}^{2}_{2x}}+2\boldsymbol{q}_x \boldsymbol{q}_y \frac{\partial^2 \bar{S}_{\boldsymbol{q}_1,\boldsymbol{q}_2}}{\partial \boldsymbol{q}_{2x} \partial \boldsymbol{q}_{2y}}+\boldsymbol{q}^2_y \frac{\partial^2 \bar{S}_{\boldsymbol{q}_1,\boldsymbol{q}_2}}{\partial \boldsymbol{q}^{2}_{2y}} \right)+O(\boldsymbol{q}^3)\\
&+\bar{S}_{\boldsymbol{q}_1,\boldsymbol{q}_2}+\boldsymbol{q}_x \frac{\partial \bar{S}_{\boldsymbol{q}_1,\boldsymbol{q}_2}}{\partial \boldsymbol{q}_{2x}}+\boldsymbol{q}_y \frac{\partial \bar{S}_{\boldsymbol{q}_1,\boldsymbol{q}_2}}{\partial \boldsymbol{q}_{2y}}+\frac{1}{2!} \left(\boldsymbol{q}^2_x \frac{\partial^2 \bar{S}_{\boldsymbol{q}_1,\boldsymbol{q}_2}}{\partial \boldsymbol{q}^{2}_{2x}}+2\boldsymbol{q}_x \boldsymbol{q}_y \frac{\partial^2 \bar{S}_{\boldsymbol{q}_1,\boldsymbol{q}_2}}{\partial \boldsymbol{q}_{2x} \partial \boldsymbol{q}_{2y}}+\boldsymbol{q}^2_y \frac{\partial^2 \bar{S}_{\boldsymbol{q}_1,\boldsymbol{q}_2}}{\partial \boldsymbol{q}^{2}_{2y}} \right)+O(\boldsymbol{q}^3)\\
=&2 \bar{S}_{\boldsymbol{q}_1,\boldsymbol{q}_2}+ \left(\boldsymbol{q}_x \frac{\partial}{\partial \boldsymbol{q}_{2x}}+\boldsymbol{q}_y \frac{\partial}{\partial \boldsymbol{q}_{2y}} \right)^2 \bar{S}_{\boldsymbol{q}_1,\boldsymbol{q}_2}+O(\boldsymbol{q}^3) \sim [2 + (\boldsymbol{q} \cdot \boldsymbol{\nabla}_2)^2] \bar{S}_{\boldsymbol{q}_1,\boldsymbol{q}_2}
\end{aligned}
\end{equation}
and
\begin{equation}
\begin{aligned}
&\bar{S}_{\boldsymbol{q}_1-\boldsymbol{q},\boldsymbol{q}_2+\boldsymbol{q}}+ \bar{S}_{\boldsymbol{q}_{1}+\boldsymbol{q}, \boldsymbol{q}_2-\boldsymbol{q}}\\
=& \bar{S}_{\boldsymbol{q}_1,\boldsymbol{q}_2}-\boldsymbol{q}_x \left(\frac{\partial \bar{S}_{\boldsymbol{q}_1,\boldsymbol{q}_2}}{\partial \boldsymbol{q}_{1x}}-\frac{\partial \bar{S}_{\boldsymbol{q}_1,\boldsymbol{q}_2}}{\partial \boldsymbol{q}_{2x}} \right)-\boldsymbol{q}_y \left(\frac{\partial \bar{S}_{\boldsymbol{q}_1,\boldsymbol{q}_2}}{\partial \boldsymbol{q}_{1y}}-\frac{\partial \bar{S}_{\boldsymbol{q}_1,\boldsymbol{q}_2}}{\partial \boldsymbol{q}_{2y}} \right)\\
&+\frac{1}{2!} \left[\boldsymbol{q}^2_x \frac{\partial^2 \bar{S}_{\boldsymbol{q}_1,\boldsymbol{q}_2}}{\partial \boldsymbol{q}^{2}_{1x}}+2\boldsymbol{q}_x \boldsymbol{q}_y \frac{\partial^2 \bar{S}_{\boldsymbol{q}_1,\boldsymbol{q}_2}}{\partial \boldsymbol{q}_{1x} \partial \boldsymbol{q}_{1y}}+\boldsymbol{q}^2_y \frac{\partial^2 \bar{S}_{\boldsymbol{q}_1,\boldsymbol{q}_2}}{\partial \boldsymbol{q}^{2}_{1y}}+\boldsymbol{q}^2_x \frac{\partial^2 \bar{S}_{\boldsymbol{q}_1,\boldsymbol{q}_2}}{\partial \boldsymbol{q}^{2}_{2x}}+2\boldsymbol{q}_x \boldsymbol{q}_y \frac{\partial^2 \bar{S}_{\boldsymbol{q}_1,\boldsymbol{q}_2}}{\partial \boldsymbol{q}_{2x} \partial \boldsymbol{q}_{2y}}+\boldsymbol{q}^2_y \frac{\partial^2 \bar{S}_{\boldsymbol{q}_1,\boldsymbol{q}_2}}{\partial \boldsymbol{q}^{2}_{2y}} \right.\\
&\left. - \left(2\boldsymbol{q}_x \boldsymbol{q}_y \frac{\partial^2 \bar{S}_{\boldsymbol{q}_1,\boldsymbol{q}_2}}{\partial \boldsymbol{q}_{1x} \partial \boldsymbol{q}_{2y}}+2\boldsymbol{q}_x \boldsymbol{q}_y \frac{\partial^2 \bar{S}_{\boldsymbol{q}_1,\boldsymbol{q}_2}}{\partial \boldsymbol{q}_{2x} \partial \boldsymbol{q}_{1y}}+2\boldsymbol{q}^2_x \frac{\partial^2 \bar{S}_{\boldsymbol{q}_1,\boldsymbol{q}_2}}{\partial \boldsymbol{q}_{1x} \partial \boldsymbol{q}_{2x}}+2 \boldsymbol{q}^2_y \frac{\partial^2 \bar{S}_{\boldsymbol{q}_1,\boldsymbol{q}_2}}{\partial \boldsymbol{q}_{1y} \partial \boldsymbol{q}_{2y}} \right) \right]+O(\boldsymbol{q}^3)\\
&+\bar{S}_{\boldsymbol{q}_1,\boldsymbol{q}_2}+\boldsymbol{q}_x \left(\frac{\partial \bar{S}_{\boldsymbol{q}_1,\boldsymbol{q}_2}}{\partial \boldsymbol{q}_{1x}}-\frac{\partial \bar{S}_{\boldsymbol{q}_1,\boldsymbol{q}_2}}{\partial \boldsymbol{q}_{2x}} \right)+\boldsymbol{q}_y \left(\frac{\partial \bar{S}_{\boldsymbol{q}_1,\boldsymbol{q}_2}}{\partial \boldsymbol{q}_{1y}}-\frac{\partial \bar{S}_{\boldsymbol{q}_1,\boldsymbol{q}_2}}{\partial \boldsymbol{q}_{2y}} \right)\\
&+\frac{1}{2!} \left[\boldsymbol{q}^2_x \frac{\partial^2 \bar{S}_{\boldsymbol{q}_1,\boldsymbol{q}_2}}{\partial \boldsymbol{q}^{2}_{1x}}+2\boldsymbol{q}_x \boldsymbol{q}_y \frac{\partial^2 \bar{S}_{\boldsymbol{q}_1,\boldsymbol{q}_2}}{\partial \boldsymbol{q}_{1x} \partial \boldsymbol{q}_{1y}}+\boldsymbol{q}^2_y \frac{\partial^2 \bar{S}_{\boldsymbol{q}_1,\boldsymbol{q}_2}}{\partial \boldsymbol{q}^{2}_{1y}}+\boldsymbol{q}^2_x \frac{\partial^2 \bar{S}_{\boldsymbol{q}_1,\boldsymbol{q}_2}}{\partial \boldsymbol{q}^{2}_{2x}}+2\boldsymbol{q}_x \boldsymbol{q}_y \frac{\partial^2 \bar{S}_{\boldsymbol{q}_1,\boldsymbol{q}_2}}{\partial \boldsymbol{q}_{2x} \partial \boldsymbol{q}_{2y}}+\boldsymbol{q}^2_y \frac{\partial^2 \bar{S}_{\boldsymbol{q}_1,\boldsymbol{q}_2}}{\partial \boldsymbol{q}^{2}_{2y}} \right.\\
& \left. -\left(2\boldsymbol{q}_x \boldsymbol{q}_y \frac{\partial^2 \bar{S}_{\boldsymbol{q}_1,\boldsymbol{q}_2}}{\partial \boldsymbol{q}_{1x} \partial \boldsymbol{q}_{2y}}+2\boldsymbol{q}_x \boldsymbol{q}_y \frac{\partial^2 \bar{S}_{\boldsymbol{q}_1,\boldsymbol{q}_2}}{\partial \boldsymbol{q}_{2x} \partial \boldsymbol{q}_{1y}}+2\boldsymbol{q}^2_x \frac{\partial^2 \bar{S}_{\boldsymbol{q}_1,\boldsymbol{q}_2}}{\partial \boldsymbol{q}_{1x} \partial \boldsymbol{q}_{2x}}+2 \boldsymbol{q}^2_y \frac{\partial^2 \bar{S}_{\boldsymbol{q}_1,\boldsymbol{q}_2}}{\partial \boldsymbol{q}_{1y} \partial \boldsymbol{q}_{2y}}\right) \right] +O(\boldsymbol{q}^3)\\
=&2\bar{S}_{\boldsymbol{q}_1,\boldsymbol{q}_2}+ \left[\boldsymbol{q}^2_x \frac{\partial^2}{\partial \boldsymbol{q}^{2}_{1x}}+2\boldsymbol{q}_x \boldsymbol{q}_y \frac{\partial^2 }{\partial \boldsymbol{q}_{1x} \partial \boldsymbol{q}_{1y}}+\boldsymbol{q}^2_y \frac{\partial^2 }{\partial \boldsymbol{q}^{2}_{1y}}+\boldsymbol{q}^2_x \frac{\partial^2 }{\partial \boldsymbol{q}^{2}_{2x}}+2\boldsymbol{q}_x \boldsymbol{q}_y \frac{\partial^2 }{\partial \boldsymbol{q}_{2x} \partial \boldsymbol{q}_{2y}}+\boldsymbol{q}^2_y \frac{\partial^2 }{\partial \boldsymbol{q}^{2}_{2y}} \right.\\
& \left. -\left(2\boldsymbol{q}_x \boldsymbol{q}_y \frac{\partial^2 }{\partial \boldsymbol{q}_{1x} \partial \boldsymbol{q}_{2y}}+2\boldsymbol{q}_x \boldsymbol{q}_y \frac{\partial^2 }{\partial \boldsymbol{q}_{2x} \partial \boldsymbol{q}_{1y}}+2\boldsymbol{q}^2_x \frac{\partial^2 }{\partial \boldsymbol{q}_{1x} \partial \boldsymbol{q}_{2x}}+2 \boldsymbol{q}^2_y \frac{\partial^2 }{\partial \boldsymbol{q}_{1y} \partial \boldsymbol{q}_{2y}}\right) \right]\bar{S}_{\boldsymbol{q}_1,\boldsymbol{q}_2}+O(\boldsymbol{q}^3)\\
=&2\bar{S}_{\boldsymbol{q}_1,\boldsymbol{q}_2}+ \left[ \left(\boldsymbol{q}_x \frac{\partial}{\partial \boldsymbol{q}_{1x}}+\boldsymbol{q}_y \frac{\partial}{\partial \boldsymbol{q}_{1y}}\right)^2+ \left(\boldsymbol{q}_x \frac{\partial}{\partial \boldsymbol{q}_{2x}}+\boldsymbol{q}_y \frac{\partial}{\partial \boldsymbol{q}_{2y}} \right)^2 \right.\\
& \left. -2 \left(\boldsymbol{q}_x \frac{\partial}{\partial \boldsymbol{q}_{1x}}+\boldsymbol{q}_y \frac{\partial}{\partial \boldsymbol{q}_{1y}} \right) \left(\boldsymbol{q}_x \frac{\partial}{\partial \boldsymbol{q}_{2x}}+\boldsymbol{q}_y \frac{\partial}{\partial \boldsymbol{q}_{2y}} \right) \right]\bar{S}_{\boldsymbol{q}_1,\boldsymbol{q}_2}+O(\boldsymbol{q}^3)\\
=&2\bar{S}_{\boldsymbol{q}_1,\boldsymbol{q}_2}+ \left(\boldsymbol{q}_x \frac{\partial}{\partial \boldsymbol{q}_{1x}}+\boldsymbol{q}_y \frac{\partial}{\partial \boldsymbol{q}_{1y}}-\boldsymbol{q}_x \frac{\partial}{\partial \boldsymbol{q}_{2x}}-\boldsymbol{q}_y \frac{\partial}{\partial \boldsymbol{q}_{2y}} \right)^2 \bar{S}_{\boldsymbol{q}_1,\boldsymbol{q}_2} +O(\boldsymbol{q}^3) \sim \left[2 + (\boldsymbol{q} \cdot (\boldsymbol{\nabla}_1 - \boldsymbol{\nabla}_2))^2 \right] \bar{S}_{\boldsymbol{q}_1,\boldsymbol{q}_2}
\end{aligned}
\end{equation}

Note that the $\boldsymbol{\nabla}_i$ here is not a rigorously-defined operator, which is only used for convenience. By taking the three equations above back to Eq.\ref{commutator3b2}, we have:
\begin{equation}
\begin{aligned}
&\left[\delta \hat{\rho}_{-\boldsymbol{q}},[\delta \hat{\rho}_{\boldsymbol{q}_1}\delta \hat{\rho}_{\boldsymbol{q}_2}\delta \hat{\rho}_{-\boldsymbol{q}_1-\boldsymbol{q}_2},\delta \hat{\rho}_{\boldsymbol{q}}]\right]\\
=&\left\{-2 [(\boldsymbol{q}_1 \times \boldsymbol{q}) \cdot (\boldsymbol{q}_2 \times \boldsymbol{q}) +( \boldsymbol{q}_{2}  \times \boldsymbol{q})^2+(\boldsymbol{q}_{1}  \times \boldsymbol{q})^2)] + [(\boldsymbol{q}_1\times \boldsymbol{q} )^2  +(\boldsymbol{q}_1 \times \boldsymbol{q}) \cdot (\boldsymbol{q}_2 \times \boldsymbol{q}) ] [2 + (\boldsymbol{q} \cdot \boldsymbol{\nabla}_1)^2] \right.\\
& \left. + [(\boldsymbol{q}_2\times \boldsymbol{q} )^2  +(\boldsymbol{q}_1 \times \boldsymbol{q}) \cdot (\boldsymbol{q}_2 \times \boldsymbol{q})][2 + (\boldsymbol{q} \cdot \boldsymbol{\nabla}_2)^2]  - [(\boldsymbol{q}_1 \times \boldsymbol{q}) \cdot (\boldsymbol{q}_2 \times \boldsymbol{q})] \left[2 + (\boldsymbol{q} \cdot (\boldsymbol{\nabla}_1 - \boldsymbol{\nabla}_2))^2 \right] \right\} \bar{S}_{\boldsymbol{q}_1,\boldsymbol{q}_2} + O(\boldsymbol{q}^5)\\
=&[(\boldsymbol{q}_1\times \boldsymbol{q} )^2  +(\boldsymbol{q}_1 \times \boldsymbol{q}) \cdot (\boldsymbol{q}_2 \times \boldsymbol{q})] (\boldsymbol{q} \cdot \boldsymbol{\nabla}_1)^2 \bar{S}_{\boldsymbol{q}_1,\boldsymbol{q}_2}+ [(\boldsymbol{q}_2\times \boldsymbol{q} )^2  +(\boldsymbol{q}_1 \times \boldsymbol{q}) \cdot (\boldsymbol{q}_2 \times \boldsymbol{q}))](\boldsymbol{q} \cdot \boldsymbol{\nabla}_2)^2 \bar{S}_{\boldsymbol{q}_1,\boldsymbol{q}_2}\\
&- [(\boldsymbol{q}_1 \times \boldsymbol{q}) \cdot (\boldsymbol{q}_2 \times \boldsymbol{q})] (\boldsymbol{q} \cdot (\boldsymbol{\nabla}_1 - \boldsymbol{\nabla}_2))^2 \bar{S}_{\boldsymbol{q}_1,\boldsymbol{q}_2} +O(\boldsymbol{q}^5)\\
=&[(\boldsymbol{q}_1\times \boldsymbol{q} )^2  (\boldsymbol{q} \cdot \boldsymbol{\nabla}_1)^2 + (\boldsymbol{q}_2\times \boldsymbol{q} )^2 (\boldsymbol{q} \cdot \boldsymbol{\nabla}_2)^2 +2 (\boldsymbol{q}_1 \times \boldsymbol{q}) \cdot (\boldsymbol{q}_2 \times \boldsymbol{q}) (\boldsymbol{q} \cdot \boldsymbol{\nabla}_1)(\boldsymbol{q} \cdot \boldsymbol{\nabla}_2)] \bar{S}_{\boldsymbol{q}_1,\boldsymbol{q}_2}+O(\boldsymbol{q}^5)\\
=&[(\boldsymbol{q}_1\times \boldsymbol{q} ) (\boldsymbol{q} \cdot \boldsymbol{\nabla}_1) + (\boldsymbol{q}_2 \times \boldsymbol{q}) (\boldsymbol{q} \cdot \boldsymbol{\nabla}_2))]^2 \bar{S}_{\boldsymbol{q}_1,\boldsymbol{q}_2} +O(\boldsymbol{q}^5)
\label{commutator3b4}
\end{aligned}
\end{equation}

\begin{definition}The non-isometric basis transformation of the momentum
\begin{equation}
\begin{aligned}
\tilde{\boldsymbol{q}}_1&=\frac{1}{\sqrt{2}}\left(\boldsymbol{q}_{1}-\boldsymbol{q}_{2}\right), \quad
\tilde{\boldsymbol{q}}_2=\sqrt{\frac{3}{2}}\left(\boldsymbol{q}_{1}+\boldsymbol{q}_{2}\right)
\end{aligned}
\end{equation}
thus:
\begin{equation}
\begin{aligned}
\boldsymbol{q}_{1}&=\frac{1}{\sqrt{2}}\tilde{\boldsymbol{q}}_1+\frac{1}{\sqrt{6}}\tilde{\boldsymbol{q}}_2, \quad \boldsymbol{q}_{2}=-\frac{1}{\sqrt{2}}\tilde{\boldsymbol{q}}_1+\frac{1}{\sqrt{6}}\tilde{\boldsymbol{q}}_2
\end{aligned}
\end{equation}
\end{definition}

Note that for constructing the commuting ladder operators, we need to transform the integral over $ \boldsymbol{q_1}$ and $\boldsymbol{q_2}$ to the one over $\boldsymbol{\tilde{q}_1}$ and $ \boldsymbol{\tilde{q}_2}$, and the Jacobian is given by:
\begin{equation}
J=\left |\begin{array}{cccc}
\frac{\partial q_{1x}}{\partial \tilde{q}_{1x}} & \frac{\partial q_{1x}}{\partial \tilde{q}_{1y}} & \frac{\partial q_{1x}}{\partial \tilde{q}_{2x}} & \frac{\partial q_{1x}}{\partial \tilde{q}_{2y}}  \\
\frac{\partial q_{1y}}{\partial \tilde{q}_{1x}} & \frac{\partial q_{1y}}{\partial \tilde{q}_{1y}}  & \frac{\partial q_{1y}}{\partial \tilde{q}_{2x}} & \frac{\partial q_{1y}}{\partial \tilde{q}_{2y}}  \\
\frac{\partial q_{2x}}{\partial \tilde{q}_{1x}} & \frac{\partial q_{2x}}{\partial \tilde{q}_{1y}} & \frac{\partial q_{2x}}{\partial \tilde{q}_{2x}} & \frac{\partial q_{2x}}{\partial \tilde{q}_{2y}}  \\
\frac{\partial q_{2y}}{\partial \tilde{q}_{1x}} & \frac{\partial q_{2y}}{\partial \tilde{q}_{1y}}  & \frac{\partial q_{2y}}{\partial \tilde{q}_{2x}} & \frac{\partial q_{2y}}{\partial \tilde{q}_{2y}}  \\
\end{array}\right|
=\left |\begin{array}{cccc}
\frac{1}{\sqrt{2}} & 0 & \frac{1}{\sqrt{6}} & 0  \\
0 & \frac{1}{\sqrt{2}} & 0 & \frac{1}{\sqrt{6}}  \\
-\frac{1}{\sqrt{2}} & 0 & \frac{1}{\sqrt{6}} & 0  \\
0 & -\frac{1}{\sqrt{2}} & 0 & \frac{1}{\sqrt{6}}  \\
\end{array}\right|
=\frac{1}{3} 
\end{equation}

\begin{definition}
The regularized momentum  
\begin{equation}
\mathbf{q}=\frac{1}{\sqrt{2}} | \tilde{\mathbf{q}}_{1} | \cdot e^{i \tilde{\theta}_{1}}, \quad
\mathbf{q}^{\prime}=\frac{1}{\sqrt{2}}\left|\tilde{\mathbf{q}}_{2}\right| \cdot e^{i \tilde{\theta}_{2}}
\end{equation}
\end{definition}

\begin{definition}\label{qndef}
The ladder operators
\begin{equation}
\left\{\begin{array}{l}
\hat{b}_{1}^{\dagger}=\frac{1}{\sqrt{2}}\left(\hat{R}_{12}^{x}+i \hat{R}_{12}^{y}\right) \\
\hat{b}_{1}=\frac{1}{\sqrt{2}}\left(\hat{R}_{12}^{x}-i \hat{R}_{12}^{y}\right)
\end{array}\right. , \quad
\left\{\begin{array}{l}
\hat{b}_{2}^{\dagger}=\frac{1}{\sqrt{2}}\left(\hat{R}_{12, 3}^{x}+i \hat{R}_{12, 3}^{y}\right) \\
\hat{b}_{2}=\frac{1}{\sqrt{2}}\left(\hat{R}_{12, 3}^{x}-i \hat{R}_{12, 3}^{y}\right)
\end{array}\right.
\end{equation}
The corresponding quantum numbers are denoted by $n_1$ and $n_2$, with the differences
\begin{equation}
\Delta n_1 = n_1^\prime - n_1, \quad \Delta n = n^\prime - n_2 
\end{equation}
where the quantum number $n_1$ denotes the relative momentum between the first and the second electron, and $n_2$ represents the relative momentum between the center-of-mass of two electrons and the third electron. Also $n_1$ can only be odd due to the fermionic statistics.
\end{definition}

\begin{corollary}
Due to the rotational invariance in the system, we have
\begin{equation}
\Delta n_1 = - \Delta n_2 
\end{equation}
\end{corollary}

\begin{definition}\label{def_3bwf}
For a three-electron rotationally invariant state $|\psi_3 \rangle$ with zero center of mass angular momentum in the magnetic field can be expanded with the complete basis $| n_1 ,n_2\rangle$:
\begin{equation}
|\psi_3 \rangle = \sum_{n_1, n_2} \alpha^{n_1, n_2} | n_1 ,n_2\rangle
\end{equation}
And the expansion coefficients are defined as $\alpha^{n_1, n_2}$.
\end{definition}

\begin{proposition}
Expectation value of the ladder operators' product 
\begin{equation}
\begin{aligned}
\left\langle n_1, n_2\left| (\hat{b}_{1}^{\dagger})^{k_1} (\hat{b}_{1})^{k_2} (\hat{b}_{2}^{\dagger})^{k_3} (\hat{b}_{2})^{k_4} \right| n_1^{\prime}, n_2^{\prime}\right\rangle= \frac{\sqrt{n_1! \cdot n_2! \cdot n_1^\prime ! \cdot n_2^\prime!} }{(n_1 - k_1)! \cdot (n_2 - k_3)!} \delta_{n_1 -k_1,  n_1^{\prime} - k_2 }, \delta_{n_2-k_3, n_2^{\prime} - k_4 }
\end{aligned}
\end{equation}
\end{proposition}

\begin{proof}
By considering Eq.\ref{ladder_commut}, we have:
\begin{equation}
\begin{aligned}
&\left\langle n_1, n_2\left| (\hat{b}_{1}^{\dagger})^{k_1} (\hat{b}_{1})^{k_2} (\hat{b}_{2}^{\dagger})^{k_3} (\hat{b}_{2})^{k_4} \right| n_1^{\prime}, n_2^{\prime}\right\rangle= \left\langle n_1, n_2\left| (\hat{b}_{1}^{\dagger})^{k_1} (\hat{b}_{2}^{\dagger})^{k_3}  \right| n_1^{\prime} - k_2, n_2^{\prime} - k_4\right\rangle \sqrt{\frac{n_1^\prime !}{(n_1^\prime - k_2) !}} \sqrt{\frac{n_2^\prime !}{(n_2^\prime - k_4) !}}\\
=& \left\langle n_1, n_2 | n_1^{\prime} - k_2 +k_1, n_2^{\prime} - k_4 +k_3\right\rangle \sqrt{\frac{n_1^\prime !}{(n_1^\prime - k_2) !}} \sqrt{\frac{n_2^\prime !}{(n_2^\prime - k_4) !}} \sqrt{\frac{(n_1^\prime -k_2 +k_1 ) !}{(n_1^\prime - k_2) !}} \sqrt{\frac{(n_2^\prime - k_4 +k_3) !}{(n_2^\prime - k_4) !}}\\
=& \delta_{n_1 -k_1,  n_1^{\prime} - k_2 }, \delta_{n_2-k_3, n_2^{\prime} - k_4 } \sqrt{\frac{n_1^\prime !}{(n_1^\prime - k_2) !}} \sqrt{\frac{n_2^\prime !}{(n_2^\prime - k_4) !}} \sqrt{\frac{(n_1^\prime -k_2 +k_1 ) !}{(n_1^\prime - k_2) !}} \sqrt{\frac{(n_2^\prime - k_4 +k_3) !}{(n_2^\prime - k_4) !}}\\
=&\frac{\sqrt{n_1! \cdot n_2! \cdot n_1^\prime ! \cdot n_2^\prime!} }{(n_1 - k_1)! \cdot (n_2 - k_3)!} \delta_{n_1 -k_1,  n_1^{\prime} - k_2 }, \delta_{n_2 - k_3, n^{\prime} - k_4 }
\end{aligned}
\end{equation}
\end{proof}

\begin{proposition}\label{prop2}
Matrix elements of three-body interaction Hamiltonians
\begin{equation}
\begin{aligned}
\left\langle n_1, n_2\left|e^{i \tilde{q}_{1a} \hat{R}_{12}^{a}} e^{i \tilde{q}_{2a} \hat{R}_{12, 3}^{a}}\right| n_1^{\prime}, n_2^{\prime}\right\rangle =\sqrt{\frac{n_1 ! \cdot n_2 !}{n_1^{\prime} ! \cdot n_2^{\prime} !}}  (i \mathbf{q})^{\Delta n_1}\left(i \mathbf{q}^{\prime}\right)^{\Delta n_2} L_{n_1}^{(\Delta n_1)}\left(|\mathbf{q}|^{2}\right) L_{n_2}^{(\Delta n_2)}\left(\left|\mathbf{q}^{\prime}\right|^{2}\right)e^{-\frac{1}{2}\left(|\mathbf{q}|^{2}+\left|\mathbf{q}^{\prime}\right|^{2}\right)}
\end{aligned}
\end{equation}
\end{proposition}

\begin{proof}
From Def.\ref{qndef} and Lemma.\ref{gcocomm}, we have:
\begin{equation}
\left[\hat{b}_{1}, \hat{b}_{1}^{\dagger}\right]=\left[\hat{b}_{2}, \hat{b}_{2}^{\dagger}\right]=1, \quad \left[\hat{b}_{1}, \hat{b}_{2}\right]=\left[\hat{b}_{1}, \hat{b}_{2}^{\dagger}\right]=0
\label{ladder_commut}
\end{equation}
and
\begin{equation}
\left\{\begin{array}{l}
\hat{R}_{12}^{x}=\frac{1}{\sqrt{2}}\left(\hat{b}_{1}^{\dagger}+\hat{b}_{1}\right) \\
\hat{R}_{12}^{y}=-\frac{i}{\sqrt{2}}(\hat{b}_{1}^{\dagger} - \hat{b}_{1})
\end{array}\right., \quad
\left\{\begin{array}{l}
\hat{R}_{12, 3}^{x}=\frac{1}{\sqrt{2}}\left(\hat{b}_{2}^{\dagger}+\hat{b}_{2}\right) \\
\hat{R}_{12, 3}^{y}=-\frac{i}{\sqrt{2}}(\hat{b}_{2}^{\dagger} - \hat{b}_{2})
\end{array}\right.
\end{equation}

\begin{equation}
\begin{aligned}
&\left\langle n_1, n_2\left|e^{i \tilde{q}_{1a} \hat{R}_{12}^{a}} e^{i \tilde{q}_{2a} \hat{R}_{12, 3}^{a}}\right| n_1^{\prime}, n_2^{\prime}\right\rangle\\
=& \left\langle n_1, n_2\left|e^{i \left[ \frac{\tilde{q}_{1x}}{\sqrt{2}}\left(\hat{b}_{1}^{\dagger}+\hat{b}_{1}\right) - i \frac{\tilde{q}_{1y}}{\sqrt{2}}\left(\hat{b}_{1}^{\dagger}-\hat{b}_{1}\right) \right]} e^{i \left[ \frac{\tilde{q}_{2x}}{\sqrt{2}}\left(\hat{b}_{2}^{\dagger}+\hat{b}_{2}\right) - i \frac{\tilde{q}_{2y}}{\sqrt{2}}\left(\hat{b}_{2}^{\dagger}-\hat{b}_{2}\right) \right]} \right| n_1^{\prime}, n_2^{\prime}\right\rangle\\
=& \left\langle n_1, n_2\left|e^{i \left[\left( \frac{\tilde{q}_{1x}}{\sqrt{2}} - i \cdot \frac{\tilde{q}_{1y}}{\sqrt{2}} \right) \hat{b}_{1}^{\dagger} + \left( \frac{\tilde{q}_{1x}}{\sqrt{2}} + i \cdot \frac{\tilde{q}_{1y}}{\sqrt{2}} \right)\hat{b}_{1} \right]} e^{i \left[\left( \frac{\tilde{q}_{2x}}{\sqrt{2}} - i \cdot \frac{\tilde{q}_{2y}}{\sqrt{2}} \right) \hat{b}_{2}^{\dagger} + \left( \frac{\tilde{q}_{2x}}{\sqrt{2}} + i \cdot \frac{\tilde{q}_{2y}}{\sqrt{2}} \right)\hat{b}_{2} \right]} \right| n_1^{\prime}, n_2^{\prime}\right\rangle\\
=& \left\langle n_1, n_2\left|e^{i \left(\mathbf{q}^* \hat{b}_{1}^{\dagger} + \mathbf{q} \hat{b}_{1} \right)} e^{i \left(\mathbf{q}^{\prime *} \hat{b}_{2}^{\dagger} + \mathbf{q}^{\prime}\hat{b}_{2} \right)} \right| n_1^{\prime}, n_2^{\prime}\right\rangle = \left\langle n_1, n_2\left|e^{i \mathbf{q}^* \hat{b}_{1}^{\dagger}} e^{i \mathbf{q} \hat{b}_{1} } e^{i \mathbf{q}^{\prime *} \hat{b}_{2}^{\dagger} } e^{i \mathbf{q}^{\prime}\hat{b}_{2} } \right| n_1^{\prime}, n_2^{\prime}\right\rangle \cdot e^{\frac{|\mathbf{q}|^2 + |\mathbf{q}^\prime|^2}{2}}\\
=& \sum_{k_1, k_2, k_3, k_4} \frac{(i \mathbf{q}^*)^{k_1} (i \mathbf{q})^{k_2} (i \mathbf{q}^{\prime *})^{k_3} (i \mathbf{q}^{\prime})^{k_4}}{k_1! \cdot k_2! \cdot k_3! \cdot k_4!} \left\langle n_1, n_2\left| (\hat{b}_{1}^{\dagger})^{k_1} (\hat{b}_{1})^{k_2} (\hat{b}_{2}^{\dagger})^{k_3} (\hat{b}_{2})^{k_4} \right| n_1^{\prime}, n_2^{\prime}\right\rangle \cdot e^{\frac{|\mathbf{q}|^2 + |\mathbf{q}^\prime|^2}{2}}\\
=& \sum_{k_1, k_2, k_3, k_4} \frac{(i \mathbf{q}^*)^{k_1} (i \mathbf{q})^{k_2} (i \mathbf{q}^{\prime *})^{k_3} (i \mathbf{q}^{\prime})^{k_4}}{k_1! \cdot k_2! \cdot k_3! \cdot k_4!} \cdot e^{\frac{|\mathbf{q}|^2 + |\mathbf{q}^\prime|^2}{2}} \cdot \delta_{n_1 -k_1,  n_1^{\prime} - k_2 }, \delta_{n_2 -k_3, n_2^{\prime} - k_4 } \cdot \frac{\sqrt{n_1! \cdot n_2! \cdot n_1^\prime ! \cdot n_2^\prime!} }{(n_1 - k_1)! \cdot (n_2 - k_3)!} \quad(*)
\end{aligned}
\end{equation}

Consider the quantum number differences defined in Def.\ref{qndef}, we have $\Delta n_1 = n_1^\prime - n_1 = k_2 - k_1$, $\Delta n_2 = n_2^\prime - n_2 = k_4 - k_3$:
\begin{equation}
\begin{aligned}
(*)=&\sum_{k_1, k_3} \frac{(i \mathbf{q}^*)^{k_1} (i \mathbf{q})^{k_1+\Delta n_1} (i \mathbf{q}^{\prime *})^{k_3} (i \mathbf{q}^{\prime})^{k_3 +\Delta n_2}}{k_1! \cdot (k_1+\Delta n_1)! \cdot k_3! \cdot (k_2+\Delta n_2)!} \cdot e^{\frac{|\mathbf{q}|^2 + |\mathbf{q}^\prime|^2}{2}} \cdot \frac{n_1^\prime ! \cdot n^\prime! }{(n_1 - k_1)! \cdot (n_2 - k_3)!} \cdot \sqrt{\frac{n_1! \cdot n_2!  }{n_1^\prime ! \cdot n_2^\prime!}}\\
=&\sqrt{\frac{n_1! \cdot n_2!}{n_1^\prime ! \cdot n_2^\prime!}} (i  \mathbf{q})^{\Delta n_1}\left(i \mathbf{q}^{\prime}\right)^{\Delta n_2} \left[\sum_{k_1}  \frac{(-1)^{k_1} \cdot n_1^\prime ! \cdot \left(|\mathbf{q}|^2 \right)^{k_1}}{(k_1 + \Delta n_1)! \cdot (n_1 - k_1)! \cdot k_1!} \right] \left[\sum_{k_3}  \frac{(-1)^{k_3} \cdot n_2^\prime ! \cdot \left(|\mathbf{q^\prime}|^2 \right)^{k_3}}{(k_3 + \Delta n_2)! \cdot (n_2 - k_3)! \cdot k_3!} \right]  \cdot e^{\frac{|\mathbf{q}|^2 + |\mathbf{q}^\prime|^2}{2}} \\
=&\sqrt{\frac{n_1 ! \cdot n_2 !}{n_1^{\prime} ! \cdot n_2^{\prime} !}}  (i  \mathbf{q})^{\Delta n_1}\left(i \mathbf{q}^{\prime}\right)^{\Delta n_2} L_{n_1}^{(\Delta n_1)}\left(|\mathbf{q}|^{2}\right) L_{n_2}^{(\Delta n_2)}\left(\left|\mathbf{q}^{\prime}\right|^{2}\right) e^{-\frac{1}{2}\left(|\mathbf{q}|^{2}+\left|\mathbf{q}^{\prime}\right|^{2}\right)}
\end{aligned}
\end{equation}
\end{proof}

\begin{lemma}\label{Fourier}
Let the operators $\hat{F}^1$ and $\hat{F}^{-1}$ denote the two-dimensional Fourier transform and its inverse transform, then the eigenvalue equation for operator $\hat F^1$ is given by \cite{yu1998laguerre}:
\begin{equation}
\hat{F}^{1}\left[\Psi_{s m}\left(\boldsymbol{q}\right) \right]=(-1)^{s} (i)^{m} \Psi_{s m}\left(\boldsymbol{p}\right)
\end{equation}
with the two-dimensional Laguerre–Gaussian (LG) function of vector $\boldsymbol{q}$ defined by:
\begin{equation}
\Psi_{s m}(\boldsymbol{q})=\left[\frac{2(s !)}{(s+m) !}\right]^{1 / 2} \boldsymbol{q}^{m} L_{s}^{(m)}\left(|\boldsymbol{q}|^{2}\right) \exp \left(-\frac{1}{2} |\boldsymbol{q}|^{2}\right)
\end{equation}
\end{lemma}

\begin{proposition}
The reduced structure factor $\bar{S}_{\boldsymbol{\tilde{q}}_1,\boldsymbol{\tilde{q}}_2}$ can be expanded with generalized Laguerre-Gaussian functions as:
\begin{equation}
\begin{aligned}
\bar{S}_{\boldsymbol{\tilde{q}}_1,\boldsymbol{\tilde{q}}_2} &= \sum_{i; n_1 n_2} d_{i}^{n_1 n_2} \sqrt{\frac{n_1 ! \cdot n_2 !}{(n_1+i) ! \cdot (n_2-i) !}} \left(\frac{\boldsymbol{\tilde{q}}_1}{\boldsymbol{\tilde{q}}_2}\right)^i \times L^{(i)}_{n_1}\left(\frac{\tilde{Q}_1}{2}\right) L^{(-i)}_{n_2}\left(\frac{\tilde{Q}_2}{2} \right)  e^{-\frac{1}{4}\left(\tilde{Q}_1+\tilde{Q}_2\right) }\\
& = \int  \frac{d^2 \tilde{\boldsymbol{p}}_1 d^2 \tilde{\boldsymbol{p}}_2}{4 \times 4 \pi^2} e^{-\frac{i}{2}(\tilde{\boldsymbol{q}}_1 \times \tilde{\boldsymbol{p}}_1 + \tilde{\boldsymbol{q}}_2 \times \tilde{\boldsymbol{p}}_2)} \bar{S}_{\boldsymbol{\tilde{p}}_1,\boldsymbol{\tilde{p}}_2}
\end{aligned}
\end{equation}
and the Fourier transform can be written as:
\begin{equation}\label{sn1n2fourier}
\begin{aligned}
\bar{S}_{\boldsymbol{\tilde{p}}_1,\boldsymbol{\tilde{p}}_2}
&=\sum_{i; n_1 n_2} (-1)^{n_1+n_2} d_{i}^{n_1 n_2} \sqrt{\frac{n_1 ! \cdot n_2 !}{(n_1+i) ! \cdot (n_2-i) !}} \left(\frac{\boldsymbol{\tilde{p}}_1}{\boldsymbol{\tilde{p}}_2}\right)^i \times L^{(i)}_{n_1}\left(\frac{\tilde{P}_1}{2}\right) L^{(-i)}_{n_2}\left(\frac{\tilde{P}_2}{2} \right)  e^{-\frac{1}{4}\left(\tilde{P}_1+\tilde{P}_2\right) }
\end{aligned}
\end{equation}
where we have defined:
\begin{equation}
\tilde{Q}_1=\tilde{\boldsymbol{q}}_1^2=|\tilde{\boldsymbol{q}}_1|^2, \quad \tilde{Q}_2=\tilde{\boldsymbol{q}}_2^2=|\tilde{\boldsymbol{q}}_2|^2, \quad\tilde{P}_1=\tilde{\boldsymbol{p}}_1^2=|\tilde{\boldsymbol{p}}_1|^2, \quad \tilde{P}_2=\tilde{\boldsymbol{p}}_2^2=|\tilde{\boldsymbol{p}}_2|^2
\end{equation}
\end{proposition}

\begin{proof}
For any three particles indexed by $s$, $j$ and $k$ in the system, the conclusions in Proposition.\ref{prop2} are always true. So by considering about the definition of the reduced structure factor, we can rewrite it as: 
\begin{equation}
\bar{S}_{\boldsymbol{q}_1,\boldsymbol{q}_2} = \sum_{s \ne j \ne k} \mathcal{N}^{sjk} \cdot \sum_{n_1, n_2} \sum_{n_1^\prime, n_2^\prime } \alpha^{* n_1 n_2} \alpha^{n^\prime_1 n^\prime_2}  \left(\left\langle n_1, n_2 \left|e^{i \tilde{q}_{s a} \hat{R}_{sj}^{a}} e^{i \tilde{q}_{j a} \hat{R}_{sj, k}^{a}}\right| n_1^{\prime}, n_2^{\prime}  \right\rangle \right)_{sjk}
\end{equation}
where $\alpha^{n_1 n_2}$ is defined in Definition.\ref{def_3bwf} the tensor $\mathcal{N}^{sjk}$ describes the overall factor related to all the other particles:
\begin{equation}
\mathcal{N}^{sjk} = \mathcal{N}\left(\boldsymbol{R}_{1}, \cdots, \boldsymbol{R}_{s-1}, \boldsymbol{R}_{s+1}, \cdots,  \boldsymbol{R}_{j-1}, \boldsymbol{R}_{j+1}, \cdots, \boldsymbol{R}_{k-1}, \boldsymbol{R}_{k+1}, \cdots, \boldsymbol{R}_{N}\right) 
\end{equation}
Thus by defining:
\begin{equation}
d_{i}^{n_1 n_2} = \sum_{s \ne j \ne k} \mathcal{N}^{sjk} \alpha^{* n_1 n_2} \alpha^{n^\prime_1 n^\prime_2} = \sum_{s \ne j \ne k} \mathcal{N}^{sjk} \alpha^{* n_1, n_2} \alpha^{n_1+i, n_2-i}
\end{equation}
we can write the reduced structure factor as:
\begin{equation}
\begin{aligned}
\bar{S}_{\boldsymbol{\tilde{q}}_1,\boldsymbol{\tilde{q}}_2} &= \sum_{i; n_1 n_2} d_{i}^{n_1 n_2} \sqrt{\frac{n_1 ! \cdot n_2 !}{(n_1+i) ! \cdot (n_2-i) !}} \left(\frac{\boldsymbol{\tilde{q}}_1}{\boldsymbol{\tilde{q}}_2}\right)^i \times L^{(i)}_{n_1}\left(\frac{\tilde{Q}_1}{2}\right) L^{(-i)}_{n_2}\left(\frac{\tilde{Q}_2}{2} \right)  e^{-\frac{1}{4}\left(\tilde{Q}_1+\tilde{Q}_2\right) }
\end{aligned}
\end{equation}

Note that $n_1^\prime$ and $n_2^\prime$ have been represented by $i \equiv \Delta n_1$ here. Then according to Lemma.\ref{Fourier}, generalized Laguerre polynomials are the eigenfunctions of two-dimensional Fourier transform, so Eq.\ref{sn1n2fourier} obviously holds.
\end{proof}

As we can see, the problem with the definition above is that we need to do lots of numerical calculations to get the values of $d_{i}^{n_1 n_2}$. However it is worth noticing that the ratio between the coefficients $d_{i}^{n_1 n_2}$ with different $i$ only depends on $\alpha^{* n_1, n_2}$ and $\alpha^{n_1+i, n_2-i}$, which can be calculated easily and the results are shown in Table.\ref{tabexpansion}. So one can take a specific $d_{i}^{n_1 n_2}$ as a reference to efficiently get the other expansion coefficients with the same $n_1$ and $n_2$ but different $i$.
\begin{definition}
The reduced-structure-factor expansion coefficient of reference
\begin{equation}
\begin{aligned}
\bar{d}^{n_1 n_2} \equiv& \int \frac{d^{2} \boldsymbol{\tilde{q}}_1 d^{2} \boldsymbol{\tilde{q}}_2}{(2 \pi)^4} \bar{S}_{\boldsymbol{\tilde{q}}_1,\boldsymbol{\tilde{q}}_2} L^{(0)}_{n_1}\left(\frac{\tilde{Q}_1}{2}\right) L^{(0)}_{n_2}\left(\frac{\tilde{Q}_2}{2} \right)  e^{-\frac{1}{4}\left(\tilde{Q}_1+\tilde{Q}_2\right)} = d_0^{n_1 n_2}
\propto  \alpha^{n_1, n_2} = \sum_{n_1, n_2} \langle n_1, n_2| \psi_3\rangle
\end{aligned}
\end{equation}
\end{definition}
\begin{corollary}\label{regularized_d}
Any expansion coefficient can also be expressed as 
\begin{equation}
d_{i}^{n_1 n_2} = \frac{\alpha^{n_1+i, n_2-i}}{\alpha^{n_1, n_2}} \cdot \bar{d}^{n_1 n_2} 
\end{equation}
\end{corollary}

Similarly the rotationally-invariant interaction $V_{\tilde{Q}_1,\tilde{Q}_2}$ can also be expanded with Laguerre-Gaussian polynomials (also called Haldane pseudo-potentials in physics):
\begin{equation}
\begin{aligned}
V_{\boldsymbol{\tilde{q}}_1,\boldsymbol{\tilde{q}}_2}&=\sum_{m_1 m_2} c^{m_1 m_2}  L^{(0)}_{m_1}\left(\frac{\tilde{Q}_1}{2}\right) L^{(0)}_{m_2}\left(\frac{\tilde{Q}_2}{2} \right)  e^{-\frac{1}{4}\left(\tilde{Q}_1+\tilde{Q}_2\right) }
\end{aligned}
\end{equation}
Note that the form of the interactions can be freely chosen as long as it obeys the basic assumptions in the last section. In this work we will not take the more complicated forms involving generalized Laguerre polynomials into account. 

Then by differentiating the exponential functions, we have:
\begin{equation}
\begin{aligned}
(\boldsymbol{q} \cdot \boldsymbol{\nabla}_1)^2 e^{-\frac{i}{2}(\tilde{\boldsymbol{q}}_1  \times \tilde{\boldsymbol{p}}_1 + \tilde{\boldsymbol{q}}_2  \times \tilde{\boldsymbol{p}}_2)}=& \left[q_x (\frac{\partial \tilde{q}_{1x} }{\partial q_{1x}} \frac{\partial}{\partial \tilde{q}_{1x}}+\frac{\partial \tilde{q}_{2x} }{\partial q_{1x}} \frac{\partial}{\partial \tilde{q}_{2x}})+q_y (\frac{\partial \tilde{q}_{1y}}{\partial q_{1y}}\frac{\partial}{\partial \tilde{q}_{1y}}+\frac{\partial \tilde{q}_{2y}}{\partial q_{1y}}\frac{\partial}{\partial \tilde{q}_{2y}})\right]^2 e^{-\frac{i}{2}(\tilde{\boldsymbol{q}}_1  \times\tilde{\boldsymbol{p}}_1 + \tilde{\boldsymbol{q}}_2  \times \tilde{\boldsymbol{p}}_2)}\\
=& - \left[ \frac{\boldsymbol{q}}{2} \times (\frac{1}{\sqrt{2}} \tilde{\boldsymbol{p}}_1 + \frac{\sqrt{3}}{\sqrt{2}} \tilde{\boldsymbol{p}}_2)  \right]^2 e^{-\frac{i}{2}(\tilde{\boldsymbol{q}}_1  \times \tilde{\boldsymbol{p}}_1 + \tilde{\boldsymbol{q}}_2  \times \tilde{\boldsymbol{p}}_2)} 
\end{aligned}
\end{equation}
and:
\begin{equation}
\begin{aligned}
(\boldsymbol{q} \cdot \boldsymbol{\nabla}_2)^2 e^{-\frac{i}{2}(\tilde{\boldsymbol{q}}_1 \times \tilde{\boldsymbol{p}}_1 + \tilde{\boldsymbol{q}}_2 \times \tilde{\boldsymbol{p}}_2)}=& \left[q_x (\frac{\partial \tilde{q}_{1x} }{\partial q_{2x}} \frac{\partial}{\partial \tilde{q}_{1x}}+\frac{\partial \tilde{q}_{2x} }{\partial q_{2x}} \frac{\partial}{\partial \tilde{q}_{2x}})+q_y (\frac{\partial \tilde{q}_{1y}}{\partial q_{2y}}\frac{\partial}{\partial \tilde{q}_{1y}}+\frac{\partial \tilde{q}_{2y}}{\partial q_{2y}}\frac{\partial}{\partial \tilde{q}_{2y}})\right]^2 e^{-\frac{i}{2}(\tilde{\boldsymbol{q}}_1  \times\tilde{\boldsymbol{p}}_1 + \tilde{\boldsymbol{q}}_2  \times \tilde{\boldsymbol{p}}_2)}\\
=&- \left[ \frac{\boldsymbol{q}}{2} \times (-\frac{1}{\sqrt{2}} \tilde{\boldsymbol{p}}_1 + \frac{\sqrt{3}}{\sqrt{2}} \tilde{\boldsymbol{p}}_2)  \right]^2 e^{-\frac{i}{2}(\tilde{\boldsymbol{q}}_1 \times \tilde{\boldsymbol{p}}_1 + \tilde{\boldsymbol{q}}_2 \times \tilde{\boldsymbol{p}}_2)}
\end{aligned}
\end{equation}
and also:
\begin{equation}
\begin{aligned}
(\boldsymbol{q} \cdot \boldsymbol{\nabla}_1)(\boldsymbol{q} \cdot \boldsymbol{\nabla}_2) e^{-\frac{i}{2}(\tilde{\boldsymbol{q}}_1 \times \tilde{\boldsymbol{p}}_1 + \tilde{\boldsymbol{q}}_2 \times \tilde{\boldsymbol{p}}_2)} =& \left[q_x (\frac{\partial \tilde{q}_{1x} }{\partial q_{1x}} \frac{\partial}{\partial \tilde{q}_{1x}}+\frac{\partial \tilde{q}_{2x} }{\partial q_{1x}} \frac{\partial}{\partial \tilde{q}_{2x}})+q_y (\frac{\partial \tilde{q}_{1y}}{\partial q_{1y}}\frac{\partial}{\partial \tilde{q}_{1y}}+\frac{\partial \tilde{q}_{2y}}{\partial q_{1y}}\frac{\partial}{\partial \tilde{q}_{2y}})\right]\\
&\quad \left[q_x (\frac{\partial \tilde{q}_{1x} }{\partial q_{2x}} \frac{\partial}{\partial \tilde{q}_{1x}}+\frac{\partial \tilde{q}_{2x} }{\partial q_{2x}} \frac{\partial}{\partial \tilde{q}_{2x}})+q_y (\frac{\partial \tilde{q}_{1y}}{\partial q_{2y}}\frac{\partial}{\partial \tilde{q}_{1y}}+\frac{\partial \tilde{q}_{2y}}{\partial q_{2y}}\frac{\partial}{\partial \tilde{q}_{2y}})\right] e^{-\frac{i}{2}(\tilde{\boldsymbol{q}}_1 \times \tilde{\boldsymbol{p}}_1 + \tilde{\boldsymbol{q}}_2 \times \tilde{\boldsymbol{p}}_2)}\\
=&- \left[ \frac{\boldsymbol{q}}{2} \times (\frac{1}{\sqrt{2}} \tilde{\boldsymbol{p}}_1 + \frac{\sqrt{3}}{\sqrt{2}} \tilde{\boldsymbol{p}}_2)  \right] \left[ \frac{\boldsymbol{q}}{2} \times (-\frac{1}{\sqrt{2}} \tilde{\boldsymbol{p}}_1 + \frac{\sqrt{3}}{\sqrt{2}} \tilde{\boldsymbol{p}}_2)  \right] e^{-\frac{i}{2}(\tilde{\boldsymbol{q}}_1 \times \tilde{\boldsymbol{p}}_1 + \tilde{\boldsymbol{q}}_2 \times \tilde{\boldsymbol{p}}_2)} 
\end{aligned}
\end{equation}

\begin{theorem}\label{long-wavelength_limit}
In the long-wavelength limit, the wave number $\boldsymbol{q} \rightarrow 0$, and we have the limit of the structure factor as:
\begin{equation}
\lim_{|\bm q|\rightarrow 0} S_{\boldsymbol{q}}= \eta \cdot \boldsymbol{q}^4 = \frac{N_e \kappa}{2} \cdot \boldsymbol{q}^4
\end{equation}
where $\kappa$ is no more than the Hall viscosity of the ground state $| \psi_0\rangle$\cite{yang2020microscopic}.
\end{theorem}

Thus the graviton mode gap turns out to be:
\begin{small}
\begin{equation}
\begin{aligned}
\delta \tilde{E}_{\boldsymbol{q} \rightarrow 0 }
=&-\frac{1}{6 S_{\boldsymbol{q}}} \int \frac{d^{2} \tilde{\boldsymbol{q}}_1 d^{2} \tilde{\boldsymbol{q}}_2}{(2 \pi)^4} V_{\tilde{\boldsymbol{q}_1},\tilde{\boldsymbol{q}_2}} \int  \frac{d^2 \tilde{\boldsymbol{p}}_1 d^2 \tilde{\boldsymbol{p}}_2}{2^2 \times 4 \pi^2} e^{-\frac{i}{2}(\tilde{\boldsymbol{q}}_1 \times \tilde{\boldsymbol{p}}_1 + \tilde{\boldsymbol{q}}_2 \times \tilde{\boldsymbol{p}}_2)} \bar{S}_{\boldsymbol{\tilde{p}}_1,\boldsymbol{\tilde{p}}_2} \\
& \quad \times \left\{(\boldsymbol{q}_1\times \boldsymbol{q} )^2 \left[  (\frac{1}{\sqrt{2}} \tilde{\boldsymbol{p}}_1 + \frac{\sqrt{3}}{\sqrt{2}} \tilde{\boldsymbol{p}}_2) \times \frac{\boldsymbol{q}}{2} \right]^2  + (\boldsymbol{q}_2\times \boldsymbol{q} )^2  \left[  (-\frac{1}{\sqrt{2}} \tilde{\boldsymbol{p}}_1 + \frac{\sqrt{3}}{\sqrt{2}} \tilde{\boldsymbol{p}}_2) \times \frac{\boldsymbol{q}}{2} \right]^2 \right. \\
& + \left. 2 \times (\boldsymbol{q}_1 \times \boldsymbol{q}) \cdot (\boldsymbol{q}_2 \times \boldsymbol{q}) \cdot \left[ (\frac{1}{\sqrt{2}} \tilde{\boldsymbol{p}}_1 + \frac{\sqrt{3}}{\sqrt{2}} \tilde{\boldsymbol{p}}_2) \times \frac{\boldsymbol{q}}{2} \right] \left[ (-\frac{1}{\sqrt{2}} \tilde{\boldsymbol{p}}_1 + \frac{\sqrt{3}}{\sqrt{2}} \tilde{\boldsymbol{p}}_2) \times \frac{\boldsymbol{q}}{2} \right] \right\} \\
=&-\frac{1}{6 S_{\boldsymbol{q}}} \int \frac{d^{2} \tilde{\boldsymbol{q}}_1 d^{2} \tilde{\boldsymbol{q}}_2}{(2 \pi)^4} V_{\tilde{\boldsymbol{q}_1},\tilde{\boldsymbol{q}_2}} \int  \frac{d^2 \tilde{\boldsymbol{p}}_1 d^2 \tilde{\boldsymbol{p}}_2}{2^4 \times 4 \pi^2} e^{-\frac{i}{2}(\tilde{\boldsymbol{q}}_1 \times \tilde{\boldsymbol{p}}_1 + \tilde{\boldsymbol{q}}_2 \times \tilde{\boldsymbol{p}}_2)} \bar{S}_{\boldsymbol{\tilde{p}}_1,\boldsymbol{\tilde{p}}_2}\\
&\times \left\{ \left[\left(\frac{1}{\sqrt{2}}\tilde{\boldsymbol{q}}_1+\frac{1}{\sqrt{6}}\tilde{\boldsymbol{q}}_2\right) \times \boldsymbol{q} \right] \left[ (\frac{1}{\sqrt{2}} \tilde{\boldsymbol{p}}_1 + \frac{\sqrt{3}}{\sqrt{2}} \tilde{\boldsymbol{p}}_2) \times \boldsymbol{q} \right] +  \left[\left(-\frac{1}{\sqrt{2}}\tilde{\boldsymbol{q}}_1+\frac{1}{\sqrt{6}}\tilde{\boldsymbol{q}}_2\right) \times \boldsymbol{q} \right]  \left[ (-\frac{1}{\sqrt{2}} \tilde{\boldsymbol{p}}_1 + \frac{\sqrt{3}}{\sqrt{2}} \tilde{\boldsymbol{p}}_2) \times \boldsymbol{q} \right] \right\}^2 
\end{aligned}
\end{equation}
\end{small}

We can set $\theta_q = \frac{\pi}{2}$ without losing any generality and get:
\begin{small}
\begin{equation}
\begin{aligned}
\delta \tilde{E}_{\boldsymbol{q}}=&-\frac{1}{6 \eta} \int \frac{d^{2} \tilde{\boldsymbol{q}}_1 d^{2} \tilde{\boldsymbol{q}}_2}{(2 \pi)^4} V_{\boldsymbol{\tilde{q}_1},\boldsymbol{\tilde{q}_2}} \int  \frac{d^2 \tilde{\boldsymbol{p}}_1 d^2 \tilde{\boldsymbol{p}}_2}{2^4 \times 4 \pi^2} e^{- \frac{i}{2} (\tilde{\boldsymbol{q}}_1 \times \tilde{\boldsymbol{p}}_1 + \tilde{\boldsymbol{q}}_2 \times \tilde{\boldsymbol{p}}_2)} \bar{S}_{\boldsymbol{\tilde{p}}_1,\boldsymbol{\tilde{p}}_2} \left[|\tilde{\boldsymbol{q}}_1| |\tilde{\boldsymbol{p}}_1 | \cos(\theta_{\tilde{q}_1})\cos(\theta_{\tilde{p}_1}) + |\tilde{\boldsymbol{q}}_2||\tilde{\boldsymbol{p}}_2 | \cos(\theta_{\tilde{q}_2})\cos(\theta_{\tilde{p}_2}) \right]^2\\
=&-\frac{1}{6 \eta}\int \frac{2|\tilde{q}_1|d |\tilde{q}_1| \times 2|\tilde{q}_2|d |\tilde{q}_2|}{4 \times (2 \pi)^4} V_{\boldsymbol{\tilde{q}_1},\boldsymbol{\tilde{q}_2}} \int\frac{2|\tilde{p}_1|d |\tilde{p}_1| \times 2|\tilde{p}_2|d |\tilde{p}_2|}{ 4 \times 2^4 \times 4 \pi^2}  \sum_{\Delta n_1}\left|\bar{S}^{(\Delta n_1)}_{\boldsymbol{\tilde{p}}_1,\boldsymbol{\tilde{p}}_2} \right| \Theta^{(\Delta n_1)}\left(\tilde{q}_1, \tilde{q}_2, \tilde{p}_1, \tilde{p}_2\right)
\label{deltaEq1}
\end{aligned}
\end{equation}
\end{small}
where $\left|\bar{S}^{(\Delta n_1)}_{\boldsymbol{\tilde{p}}_1,\boldsymbol{\tilde{p}}_2} \right|$ denotes the radial part of each term in the structure factor expansion. Moreover the angular integral is defined by
\begin{equation}
\begin{aligned}
\Theta^{(\Delta n_1)}\left(\tilde{q}_1, \tilde{q}_2, \tilde{p}_1, \tilde{p}_2\right) = & \iiiint_{-\pi}^{\pi}   d \theta_{\tilde{q}_1} d \theta_{\tilde{q}_2} d \theta_{\tilde{p}_1} d \theta_{\tilde{p}_2} e^{-\frac{i}{2}[|\tilde{q}_1| |\tilde{p}_1| \sin(\theta_{\tilde{q}_1} - \theta_{\tilde{p}_1})+ |\tilde{q}_2| | \tilde{p}_2|  \sin(\theta_{\tilde{q}_2} - \theta_{\tilde{p}_2})]} \\
& \quad \times e^{i \Delta n_1 (\theta_{\tilde{p}_1} - \theta_{\tilde{p}_2})} \times  [|\tilde{\boldsymbol{q}}_1| |\tilde{\boldsymbol{p}}_1 | \cos(\theta_{\tilde{q}_1})\cos(\theta_{\tilde{p}_1})+ |\tilde{\boldsymbol{q}}_2||\tilde{\boldsymbol{p}}_2 | \cos(\theta_{\tilde{q}_2})\cos(\theta_{\tilde{p}_2})]^2 
\end{aligned}
\end{equation}

\begin{theorem}\label{besseltheorem}
Bessel function from angular integrals
\begin{equation}
\begin{aligned}
\int_{-\pi}^{\pi}\int_{-\pi}^{\pi} d \theta_1  d \theta_2 e^{i[x \cos (\theta_1-\theta_2)-\alpha (\theta_1-\theta_2)+\beta \theta_2]} = 4 \pi^2 J_\alpha(x) \delta(\beta) 
\end{aligned}
\end{equation}
\end{theorem}

\begin{proof}
The formula can be proved by:
\begin{equation}
\begin{aligned}
&\int_{-\pi}^{\pi}\int_{-\pi}^{\pi} d \theta_1  d \theta_2 e^{i[x \cos (\theta_1-\theta_2)-\alpha (\theta_1-\theta_2)+\beta \theta_2]} = \int_{-\pi}^{\pi}e^{i(\beta \theta_2)} d \theta_2 \int_{-\pi}^{\pi} d \theta_1 e^{i [x \cos (\theta_1-\theta_2)-\alpha (\theta_1-\theta_2)]}\\
=&\int_{-\pi}^{\pi}e^{i(\beta \theta_2)} d \theta_2 \int_{-\pi-\theta_2}^{\pi-\theta_2} d (\theta_1-\theta_2) e^{i[x \sin (\frac{\pi}{2}-\theta_1+\theta_2)-\alpha (\theta_1-\theta_2)]}= \int_{-\pi}^{\pi}e^{i(\beta \theta_2)} d \theta_2 \int_{-\pi-\theta_2}^{\pi-\theta_2} d \theta_{12} e^{i[ x \sin (\theta_{12})-\alpha \theta_{12}]}\\
=& \int_{-\pi}^{\pi}e^{i(\beta \theta_2)} d \theta_2 \int_{-\pi}^{\pi} d \theta_{12} e^{i[x \sin (\theta_{12})-\alpha \theta_{12}]}= 2 \pi J_\alpha(x) \int_{-\pi}^{\pi}e^{i(\beta \theta_2)} d \theta_2 = 4 \pi^2 J_\alpha(x) \delta(\beta)  \qedhere
\end{aligned}
\end{equation}
\end{proof}

\begin{proposition}
Domain of the quantum number differences is given by
\begin{equation}
\Delta n_1 = - \Delta n_2 \in \{\pm 2, 0\}
\end{equation}
\end{proposition}

\begin{proof}
The integrand in $\Theta^{(\Delta n_1)}\left(\tilde{q}_1, \tilde{q}_2, \tilde{p}_1, \tilde{p}_2\right)$ can be written as the linear combination of:
\begin{equation}
\begin{aligned}
& \left[e^{-\frac{i}{2}[|\tilde{q}_1| |\tilde{p}_1| \sin(\theta_{\tilde{q}_1} - \theta_{\tilde{p}_1})+ |\tilde{q}_2| | \tilde{p}_2|  \sin(\theta_{\tilde{q}_2} - \theta_{\tilde{p}_2})]}  |\tilde{\boldsymbol{q}}_i| |\tilde{\boldsymbol{p}}_i | |\tilde{\boldsymbol{q}}_j| |\tilde{\boldsymbol{p}}_j |\right] e^{i \Delta n_1 (\theta_{\tilde{p}_1} - \theta_{\tilde{p}_2})}  \cos (\theta_{\tilde{q}_i})\cos(\theta_{\tilde{p}_i})\cos (\theta_{\tilde{q}_j})\cos(\theta_{\tilde{p}_j}) \\
=&\frac{\left[ \cdots \right]}{2^4} \left(e^{i \theta_{\tilde{q}_i}}+e^{-i  \theta_{\tilde{q}_i}} \right) \left( e^{i  \theta_{\tilde{q}_j}}+e^{-i  \theta_{\tilde{q}_j}} \right) \left[e^{i \Delta n_1 \theta_{\tilde{p}_1}} e^{-i \Delta n_1   \theta_{\tilde{p}_2}}   \left( e^{i  \theta_{\tilde{p}_i}}+e^{-i \theta_{\tilde{p}_i}} \right)  \left( e^{i \theta_{\tilde{p}_j}}+e^{-i  \theta_{\tilde{p}_j}} \right) \right]
\end{aligned}
\end{equation}
where $i, j \in \{ 1, 2\}$. Thus according to Theorem.\ref{besseltheorem}, when $i = j$, $ \Delta n$ must vanish for getting a non-zero term in $\Theta^{(\Delta n_1)}\left(\tilde{q}_1, \tilde{q}_2, \tilde{p}_1, \tilde{p}_2\right)$ otherwise integrating either $\theta_{\tilde{q}_1}$ or $\theta_{\tilde{q}_2}$ will give a zero. As for $i \ne j$, only when
\begin{equation}
\begin{aligned}
1 \pm \Delta n_1 = \pm1 
\end{aligned}
\end{equation}
can we have a non-zero term in $\Theta^{(\Delta n_1)}\left(\tilde{q}_1, \tilde{q}_2, \tilde{p}_1, \tilde{p}_2\right)$. Hence $\Delta n_1$ can only be $\pm 2$ or $0$.
\end{proof}

\subsection{When $i = j$ and $\Delta n_1 = 0$}
Based on Theorem.\ref{besseltheorem}, when $i = j$, $\Theta^{(0)}\left(\tilde{q}_1, \tilde{q}_2, \tilde{p}_1, \tilde{p}_2\right)$ can be written as:
\begin{equation}
\begin{aligned}
& \frac{1}{2^4} \iiiint_{-\pi}^{\pi}   d \theta_{\tilde{q}_1} d \theta_{\tilde{q}_2} d \theta_{\tilde{p}_1} d \theta_{\tilde{p}_2} e^{-\frac{i}{2}[|\tilde{q}_1| |\tilde{p}_1| \sin(\theta_{\tilde{q}_1} - \theta_{\tilde{p}_1})+ |\tilde{q}_2| | \tilde{p}_2|  \sin(\theta_{\tilde{q}_2} - \theta_{\tilde{p}_2})]} \tilde{Q}_i \tilde{P}_i    \left[ 4+ e^{i 2 (\theta_{\tilde{q}_i}-\theta_{\tilde{p}_i})}+e^{-i 2( \theta_{\tilde{q}_i}-\theta_{\tilde{p}_i})}  \right]  \\
=&\frac{\tilde{Q}_i \tilde{P}_i}{2^4} \times \left[  4 \times 4 \pi^2 J_0 \left(\frac{|\tilde{q}_1| |\tilde{p}_1|}{2} \right) \times 4 \pi^2 J_0 \left(\frac{|\tilde{q}_2| |\tilde{p}_2|}{2} \right) + 2 \times4 \pi^2 J_2 \left( \frac{|\tilde{q}_i| |\tilde{p}_i|}{2} \right) \times 4 \pi^2 J_0 \left(\frac{|\tilde{q}_{3-i}| |\tilde{p}_{3-i}|}{2} \right) \right] \\
=&4 \pi^4 \tilde{Q}_i \tilde{P}_i \times \left[   J_0 \left( \frac{|\tilde{q}_1| |\tilde{p}_1|}{2} \right)  J_0 \left(\frac{|\tilde{q}_2| |\tilde{p}_2|}{2} \right) + \frac{1}{2} J_2 \left(\frac{|\tilde{q}_i| |\tilde{p}_i|}{2} \right)  J_0 \left( \frac{|\tilde{q}_{3-i}| |\tilde{p}_{3-i}|}{2} \right) \right]
\end{aligned}
\end{equation}

\begin{lemma}\label{Bessel_Hypergeometric}
Bessel functions and Hypergeometric functions \cite{abramowitz1948handbook}
\begin{equation}
J_\alpha(x)=\frac{(\frac{x}{2})^\alpha}{\Gamma(\alpha + 1)} \  _0\mathbf{F}_1(;\alpha+1;-\frac{x^2}{4})
\end{equation}
\end{lemma}

\begin{lemma}
Recurrence relations of Bessel functions \cite{abramowitz1948handbook}
\begin{equation}
\frac{2 \alpha}{x} J_{\alpha}(x) = J_{\alpha-1}(x) + J_{\alpha+1}(x)
\end{equation}
\end{lemma}

\begin{corollary}
For $\alpha = 0, 1, 2$, we have
\begin{equation}
J_0(x)= \ _0\mathbf{F}_1(;1;-\frac{x^2}{4}); \quad
\frac{2}{x} J_1(x) =\  _0\mathbf{F}_1(;2;-\frac{x^2}{4}); \quad
J_2(x)= \frac{2}{x} J_{1}(x)- J_0(x)
\end{equation}
\end{corollary}

\begin{lemma}\label{Hardy-Hille}
Hardy-Hille Formula \cite{abramowitz1948handbook}
\begin{equation}
\sum_{n=0}^{\infty} \frac{n ! \Gamma(\alpha+1)}{\Gamma(n+\alpha+1)} L_{n}^{(\alpha)}(x) L_{n}^{(\alpha)}(y) t^{n}=\frac{1}{(1-t)^{\alpha+1}} e^{-(x+y) t /(1-t)} \ _{0}\mathbf{F}_{1}\left(; \alpha+1 ; \frac{x y t}{(1-t)^{2}}\right)
\label{HHF}
\end{equation}
\end{lemma}

\begin{corollary}
Substituting $\alpha=0, t=-1$ to Eq.\ref{HHF}, we have 
\begin{equation}
\sum_{n=0}^{\infty}  L_{n}^{(0)}(x) L_{n}^{(0)}(y) (-1)^{n}=\frac{1}{2} e^{(x+y)/2} \ _{0} \mathbf{F}_{1}\left(; 1 ; -\frac{x y }{4}\right)
\end{equation}

Substituting $\alpha=1, t=-1$ to Eq.\ref{HHF}, we have
\begin{equation}
\sum_{n=0}^{\infty} \frac{1}{n+1} L_{n}^{(1)}(x) L_{n}^{(1)}(y) (-1)^{n}=\frac{1}{4} e^{(x+y)/2} \ _{0} \mathbf{F}_{1}\left(; 2 ; -\frac{x y }{4}\right)
\end{equation}
\end{corollary}

Thus we can express $\Theta^{(\Delta n_1)}\left(\tilde{q}_1, \tilde{q}_2, \tilde{p}_1, \tilde{p}_2\right)$ with the Hypergeometric functions:
\begin{small}
\begin{equation}
\begin{aligned}
&\Theta^{(0)}(\tilde{Q}_1,\tilde{Q}_2,\tilde{P}_1,\tilde{P}_2)=4 \pi^4 \left[  (\tilde{Q}_1 \tilde{P}_1  +  \tilde{Q}_2 \tilde{P}_2  )  \times \ _0\mathbf{F}_1 (;1;-\frac{\tilde{Q}_1 \tilde{P}_1}{16} ) \ _0\mathbf{F}_1 (;1;-\frac{\tilde{Q}_2 \tilde{P}_2}{16} ) \right.\\
&\qquad \qquad \qquad \qquad \qquad \qquad+  \frac{1}{2}\tilde{Q}_1 \tilde{P}_1  \left( \ _0\mathbf{F}_1(;2;-\frac{\tilde{Q}_1 \tilde{P}_1}{16})-  \ _0\mathbf{F}_1(;1;-\frac{\tilde{Q}_1 \tilde{P}_1}{16}) \right)  \ _0\mathbf{F}_1(;1;-\frac{\tilde{Q}_2 \tilde{P}_2}{16})  \\
& \left. \qquad \qquad \qquad \qquad \qquad \qquad+ \frac{1}{2} \tilde{Q}_2 \tilde{P}_2  \ _0\mathbf{F}_1(;1;-\frac{\tilde{Q}_1 \tilde{P}_1}{16}) \left( \ _0\mathbf{F}_1(;2;-\frac{\tilde{Q}_2 \tilde{P}_2}{16})-  \ _0\mathbf{F}_1(;1;-\frac{\tilde{Q}_2 \tilde{P}_2}{16}) \right)  \right]\\
=&4 \times 4\pi^4 e^{-\frac{1}{4} \left(\tilde{Q}_1 +\tilde{Q}_2 + \tilde{P}_1+ \tilde{P}_2 \right)}\sum_{a,b=0}^{\infty}(-1)^{a+b} \left[  \frac{\tilde{Q}_1 \tilde{P}_1    + \tilde{Q}_2 \tilde{P}_2}{2}   \times 
L_{a}^{(0)}\left(\frac{\tilde{Q}_1}{2}\right) 
L_{a}^{(0)}\left(\frac{\tilde{P}_1}{2}\right) 
L_{b}^{(0)}\left(\frac{\tilde{Q}_2}{2}\right)
L_{b}^{(0)}\left(\frac{\tilde{P}_2}{2}\right) \right.\\
&\quad \left. +  \frac{\tilde{Q}_1 \tilde{P}_1}{a+1}   
L_{a}^{(1)}\left(\frac{\tilde{Q}_1}{2}\right) 
L_{a}^{(1)}\left(\frac{\tilde{P}_1}{2}\right) 
L_{b}^{(0)}\left(\frac{\tilde{Q}_2}{2}\right) 
L_{b}^{(0)}\left(\frac{\tilde{P}_2}{2}\right)  
+ \frac{\tilde{Q}_2 \tilde{P}_2}{b+1}  
L_{a}^{(0)}\left(\frac{\tilde{Q}_1}{2}\right) 
L_{a}^{(0)}\left(\frac{\tilde{P}_1}{2}\right) 
L_{b}^{(1)}\left(\frac{\tilde{Q}_2}{2}\right) 
L_{b}^{(1)}\left(\frac{\tilde{P}_2}{2}\right) \right] 
\end{aligned}
\end{equation}
\end{small}

\begin{lemma}\label{Orthogonality_la}
Orthogonality of the generalized Laguerre polynomials \cite{abramowitz1948handbook}
\begin{equation}
\int_{0}^{\infty} x^{\alpha} e^{-x} L_{n}^{(\alpha)}(x) L_{m}^{(\alpha)}(x) d x=\frac{(n+\alpha) !}{n !} \delta_{n, m}
\end{equation}
\end{lemma}

\begin{lemma}\label{Recurrence_la}
Recurrence relations of the generalized Laguerre polynomials \cite{abramowitz1948handbook}
\begin{equation}
L_{n}^{(\alpha)}(x)=L_{n}^{(\alpha+1)}(x)-L_{n-1}^{(\alpha+1)}(x)=\sum_{j=0}^{k}\left(\begin{array}{l}
{k} \\
{j}
\end{array}\right) L_{n-j}^{(\alpha+k)}(x)
\end{equation}
\end{lemma}

\begin{definition}\label{kronecker}
Linear combination of Kronecker-delta symbols
\begin{equation}
\begin{aligned}
\Delta_{ij}\equiv&\int_{0}^{\infty} x e^{-x} L_{i}^{(0)}(x) L_{j}^{(1)}(x) d x=\int_{0}^{\infty} x e^{-x} \left[ L_{i}^{(1)}(x) - L_{i-1}^{(1)}(x) \right] L_{j}^{(1)}(x) d x
= (i+1) \delta_{i, j} - i \delta_{i-1, j}
\end{aligned}
\end{equation}
and: 
\begin{equation}
\begin{aligned}
\Delta^{*}_{ij}\equiv&\int_{0}^{\infty} x e^{-x} L_{i}^{(0)}(x) L_{j}^{(0)}(x) d x =\int_{0}^{\infty} x e^{-x} \left[L_{i}^{(1)}(x) - L_{i-1}^{(1)}(x)\right] \left[L_{j}^{(1)}(x) - L_{j-1}^{(1)}(x)\right] d x\\
=&(i+1) \delta_{i,j} - (i+1) \delta_{i,j-1} -i \delta_{i-1,j} + i \delta_{i-1,j-1} 
\end{aligned}
\end{equation}
\end{definition}

\begin{definition}
The constant factor in $\bar{S}_{\boldsymbol{\tilde{p}}_1,\boldsymbol{\tilde{p}}_2}$
\begin{equation}
\omega^{n_1 n_2}_{i}  = \sqrt{\frac{n_1 ! \cdot n_2 !}{(n_1+i) ! \cdot (n_2-i) !}}
\end{equation}
\end{definition}

Then by using the orthogonality of the Laguerre polynomials, we can see that the contribution to the energy gap turns out to be the combination of delta functions. So when we use this formula, the only input will be $c^{m_1 m_2}$ and $d_0^{n_1 n_2}$: 
\begin{small}
\begin{equation}
\begin{aligned}
\delta \tilde{E}^{(0)}_{\boldsymbol{q}}=&-\frac{4 \times 2^2}{6 \eta \times 2^{8} \pi^2}\int \frac{d \tilde{Q}_1  d \tilde{Q}_2 d \tilde{P}_1  d \tilde{P}_2}{2^4} \sum_{m_1 m_2}\sum_{n_1 n_2} \sum_{a,b}(-1)^{m_1+m_2+a+b} c^{m_1 m_2} d_{0}^{n_1 n_2} \omega^{n_1 n_2}_{0}\times e^{-\frac{1}{2}\left(\tilde{Q}_1+\tilde{Q}_2+\tilde{P}_1+\tilde{P}_2\right)}\\
L^{(0)}_{n_1}\left(\frac{\tilde{P}_1}{2}\right)& L^{(0)}_{m_1}\left(\frac{\tilde{Q}_1}{2}\right) L^{(0)}_{n_2}\left(\frac{\tilde{P}_2}{2} \right)     L^{(0)}_{m_2}\left(\frac{\tilde{Q}_2}{2} \right) \times \left[ \frac{\tilde{Q}_1 \tilde{P}_1  + \tilde{Q}_2 \tilde{P}_2}{2 \times 4}    \times L^{(0)}_{a}\left(\frac{\tilde{P}_1}{2}\right) L^{(0)}_{a}\left(\frac{\tilde{Q}_1}{2}\right) L^{(0)}_{b}\left(\frac{\tilde{P}_2}{2} \right)     L^{(0)}_{b}\left(\frac{\tilde{Q}_2}{2} \right) \right.\\
+  \frac{\tilde{Q}_1 \tilde{P}_1}{4(a+1)} & \left. L^{(1)}_{a}\left(\frac{\tilde{P}_1}{2}\right) L^{(1)}_{a}\left(\frac{\tilde{Q}_1}{2}\right) L^{(0)}_{b}\left(\frac{\tilde{P}_2}{2} \right)  L^{(0)}_{b}\left(\frac{\tilde{Q}_2}{2} \right) + \frac{\tilde{Q}_2 \tilde{P}_2}{4(b+1)}   L^{(0)}_{a}\left(\frac{\tilde{P}_1}{2}\right) L^{(0)}_{a}\left(\frac{\tilde{Q}_1}{2}\right) L^{(1)}_{b}\left(\frac{\tilde{P}_2}{2} \right)   L^{(1)}_{b}\left(\frac{\tilde{Q}_2}{2} \right)\right] \\
=& -\frac{1}{2^{6}\times3 \eta \pi^2} \sum_{m_1 m_2}\sum_{n_1 n_2} c^{m_1 m_2} d_{0}^{n_1 n_2} \sum_{a,b}(-1)^{m_1+m_2+a+b}  \left(\frac{2}{a+1} \Delta_{n_1 a} \Delta_{m_1 a} \delta_{n_2,b} \delta_{m_2,b}   + \Delta^{*}_{n_1 a} \Delta^{*}_{m_1 a} \delta_{n_2,b} \delta_{m_2,b}\right.\\
&  \left. + \frac{2}{b+1}  \Delta_{n_2 b} \Delta_{m_2 b} \delta_{n_1 a} \delta_{m_1 a} + \Delta^{*}_{n_2,b} \Delta^{*}_{m_2,b} \delta_{n_1 a} \delta_{m_1 a} \right) \equiv -\frac{1}{2^{6}\times3 \eta \pi^2} \Gamma^{0}_{m_1 m_2 n_1 n_2} c^{m_1 m_2} d_{0}^{n_1 n_2}
\end{aligned}
\end{equation}
\end{small}

\begin{proposition}
The tensor describing the contribution to $\delta E_{\boldsymbol{q}}$ from the diagonal terms is given by
\begin{small}
\begin{equation}
\begin{aligned}
\Gamma^{0}_{m_1 m_2 n_1 n_2}=& 2(n_1^2 + n_2^2 + n_1 +n_2 +2) \delta_{n_1, m_1} \delta_{n_2, m_2}-n_1 (n_1 -1) \delta_{n_1, m_1+2} \delta_{n_2, m_2} 
(n_1 +1) (n_1 + 2) \delta_{n_1, m_1-2} \delta_{n_2, m_2}  \\
&-n_2 (n_2 -1)  \delta_{n_1, m_1} \delta_{n_2, m_2+2}  - (n_2 +1) (n_2 + 2)  \delta_{n_1, m_1} \delta_{n_2, m_2-2}
\end{aligned}
\end{equation}
\end{small}
\end{proposition}

\begin{proof}
By considering Definition.\ref{kronecker}, we have:
\begin{small}
\begin{equation}
\begin{aligned}
&  \sum_{a,b}(-1)^{m_1+m_2+a+b}  \frac{1}{a+1} \Delta_{n_1 a} \Delta_{m_1 a}  \delta_{n_2,b} \delta_{m_2,b}\\
=&   \sum_{a,b}(-1)^{m_1+m_2+a+b}  \frac{1}{a+1}[ (n_1+1) \delta_{n_1, a} - n_1 \delta_{n_1-1, a}][(m_1+1) \delta_{m_1, a} - m_1 \delta_{m_1-1, a}] \delta_{n_2,b} \delta_{m_2,b} \\
=&   \sum_{a,b}(-1)^{m_1+m_2}  \frac{(-1)^a}{a+1}[ (n_1+1) \delta_{n_1, a} - n_1 \delta_{n_1-1, a}][(m_1+1) \delta_{m_1, a} - m_1 \delta_{m_1-1, a}] (-1)^b\delta_{n_2,b} \delta_{m_2,b} \\
=&   \sum_{a,b}(-1)^{m_1+m_2}  (-1)^{n_1}(\delta_{n_1, a} +  \delta_{n_1-1, a})[(m_1+1) \delta_{m_1, a}- m_1 \delta_{m_1-1, a}] (-1)^b\delta_{n_2,b} \delta_{m_2,b} \\
=&   \sum_{a}(-1)^{m_1+m_2}  (-1)^{n_1}[(m_1+1)\delta_{n_1, a} \delta_{m_1, a} + (m_1+1) \delta_{n_1-1, a} \delta_{m_1, a} - m_1 \delta_{m_1-1, a}\delta_{n_1, a} - m_1 \delta_{m_1-1, a} \delta_{n_1-1, a} ] (-1)^{m_2}\delta_{n_2,m_2} \\
=&    [(-1)^{2 n_1}(n_1+1) \delta_{n_1,m_1} + (-1)^{2 n_1-1} n_1 \delta_{n_1-1, m_1} - (-1)^{2 n_1+1} (n_1+1) \delta_{m_1-1, n_1} - (-1)^{2 n_1} n_1 \delta_{m_1, n_1} ] \delta_{n_2,m_2} \\
=&   [\delta_{n_1, m_1} + (n_1+1) \delta_{n_1,m_1-1} - n_1 \delta_{n_1, m_1+1} ]\delta_{n_2,m_2} 
\end{aligned}
\end{equation}
\end{small}
Similarly by considering the symmetric operation of $m_1, m_2 \leftrightarrow n_1, n_2$ and  $a \leftrightarrow b$, one can easily get:
\begin{equation}
\begin{aligned}
&  \sum_{a,b}(-1)^{m_1+m_2+a+b}  \frac{1}{b+1} \Delta_{n_2 b} \Delta_{m_2 b}  \delta_{n_1,a} \delta_{m_1,a}=  [ \delta_{n_2, m_2} + (n_2+1) \delta_{n_2,m_2-1} - n_2 \delta_{n_2, m_2+1} ]\delta_{n_1,m_1}
\end{aligned}
\end{equation}
and
\begin{small}
\begin{equation}
\begin{aligned}
& \sum_{a,b}(-1)^{m_1+m_2+a+b}   \Delta^{*}_{n_1 a} \Delta^{*}_{m_1 a} \delta_{n_2,b} \delta_{m_2,b}\\
=&\sum_{a,b}(-1)^{m_1+m_2+a+b}   [ (n_1+1) \delta_{n_1,a} - (n_1+1) \delta_{n_1,a-1} -n_1 \delta_{n_1-1,a} + n_1 \delta_{n_1-1,a-1} ]\\
& \quad \times [ (m_1+1) \delta_{m_1,a} - (m_1+1) \delta_{m_1,a-1} -m_1 \delta_{m_1-1,a} + m_1 \delta_{m_1-1,a-1} ] \delta_{n_2,b} \delta_{m_2,b}\\
=&[(2 n_1 +1)^2 \delta_{n_1 ,m_1} +  n_1 (2 n_1+1) \delta_{n_1, m_1+1} + (2n_1+1)(n_1+1) \delta_{n_1,m_1-1}-(n_1+1)(2 n_1+3) \delta_{n_1,m_1-1}-(n_1 +1)^2 \delta_{n_1,m_1}\\
&  \quad - (n_1 +1)(n_1+2)\delta_{n_1,m_1-2} -n_1 (2 n_1-1) \delta_{n_1,m_1+1}-n_1 (n_1 -1) \delta_{n_1,m_1+2} - n_1^2 \delta_{n_1,m_1}]\delta_{n_2,m_2}\\
=&(2 n_1^2 + 2 n_1) \delta_{n_1 ,m_1} +2 n_1 \delta_{n_1,m_1+1}  - 2 (n_1+1)\delta_{n_1,m_1-1} -n_1 (n_1 -1) \delta_{n_1,m_1+2}- (n_1 +1)(n_1+2)\delta_{n_1,m_1-2}]\delta_{n_2,m_2}
\end{aligned}
\end{equation}
\end{small}

Similarly we have:
\begin{small}
\begin{equation}
\begin{aligned}
& \sum_{a,b}(-1)^{m_1+m_2+a+b}   \Delta^{*}_{n_2,b} \Delta^{*}_{m_2,b} \delta_{n_1 a} \delta_{m_1 a}\\
=&[(2 n_2^2 + 2 n_2 )\delta_{n_2 ,m_2} + 2 n_2 \delta_{n_2,m_2+1}  - 2 (n_2+1)\delta_{n_2,m_2-1} -n_2 (n_2 -1) \delta_{n_2,m_2+2}- (n_2 +1)(n_2+2)\delta_{n_2,m_2-2}]\delta_{n_1,m_1}
\end{aligned}
\end{equation}
\end{small}
Then by combining these terms one can easily get $\Gamma^{0}_{m_1 m_2 n_1 n_2}$. 
\end{proof}

\subsection{When $i \ne j$ and $\Delta n_1 = 0$}
\begin{lemma}
Index Parity of Bessel functions \cite{abramowitz1948handbook}
\begin{equation}
J_{-n}(x) = (-1)^n J_n(x)
\end{equation}
\end{lemma}

\begin{proposition}\label{prop_inej}
There is no contribution to $\delta E_{\boldsymbol{q}}$ from the terms with $i \ne j$ and $\Delta n_1 = 0$.
\end{proposition}

\begin{proof}
The angular integral can be written as:
\begin{equation}
\begin{aligned}
&  \iiiint_{-\pi}^{\pi}   d \theta_{\tilde{q}_1} d \theta_{\tilde{q}_2} d \theta_{\tilde{p}_1} d \theta_{\tilde{p}_2} e^{-\frac{i}{2}[|\tilde{q}_1| |\tilde{p}_1| \sin(\theta_{\tilde{q}_1} - \theta_{\tilde{p}_1})+ |\tilde{q}_2| | \tilde{p}_2|  \sin(\theta_{\tilde{q}_2} - \theta_{\tilde{p}_2})]} \\
& \quad \times \frac{1}{2^4}  |\tilde{\boldsymbol{q}}_1| |\tilde{\boldsymbol{p}}_1 | |\tilde{\boldsymbol{q}}_2| |\tilde{\boldsymbol{p}}_2 | \left(e^{i \theta_{\tilde{q}_1}}+e^{-i  \theta_{\tilde{q}_1}} \right) \left( e^{i  \theta_{\tilde{p}_1}}+e^{-i \theta_{\tilde{p}_1}} \right) \left( e^{i  \theta_{\tilde{q}_2}}+e^{-i  \theta_{\tilde{q}_2}} \right) \left( e^{i \theta_{\tilde{p}_2}}+e^{-i  \theta_{\tilde{p}_2}} \right)  \\
=&  \iiiint_{-\pi}^{\pi}   d \theta_{\tilde{q}_1} d \theta_{\tilde{q}_2} d \theta_{\tilde{p}_1} d \theta_{\tilde{p}_2} e^{-\frac{i}{2} \left[|\tilde{q}_1| |\tilde{p}_1| \sin(\theta_{\tilde{q}_1} - \theta_{\tilde{p}_1})+ |\tilde{q}_2| | \tilde{p}_2|  \sin(\theta_{\tilde{q}_2} - \theta_{\tilde{p}_2}) \right]} \\
& \quad  \times \frac{1}{2^4} |\tilde{\boldsymbol{q}}_1| |\tilde{\boldsymbol{p}}_1 | |\tilde{\boldsymbol{q}}_2| |\tilde{\boldsymbol{p}}_2 | \left[e^{i  (\theta_{\tilde{q}_1}-\theta_{\tilde{p}_1})}+e^{-i ( \theta_{\tilde{q}_1}-\theta_{\tilde{p}_1})} \right] \left[ e^{i  (\theta_{\tilde{q}_2}-\theta_{\tilde{p}_2})}+e^{-i ( \theta_{\tilde{q}_2}-\theta_{\tilde{p}_2})}  \right]\\
=&  \pi^4 |\tilde{\boldsymbol{q}}_1| |\tilde{\boldsymbol{p}}_1 | |\tilde{\boldsymbol{q}}_2| |\tilde{\boldsymbol{p}}_2 | \times \left[ J_1 \left( \frac{|\tilde{\boldsymbol{q}}_1| |\tilde{\boldsymbol{p}}_1 |}{2} \right) J_1 \left( \frac{|\tilde{\boldsymbol{q}}_2| |\tilde{\boldsymbol{p}}_2 |}{2} \right)+J_1 \left( \frac{|\tilde{\boldsymbol{q}}_1| |\tilde{\boldsymbol{p}}_1 |}{2} \right) J_{-1} \left( \frac{|\tilde{\boldsymbol{q}}_2| |\tilde{\boldsymbol{p}}_2 |}{2} \right) \right.\\
& \left. \quad +J_{-1} \left(\frac{|\tilde{\boldsymbol{q}}_1| |\tilde{\boldsymbol{p}}_1 |}{2} \right) J_1 \left( \frac{|\tilde{\boldsymbol{q}}_2| |\tilde{\boldsymbol{p}}_2 |}{2} \right)+J_{-1} \left( \frac{|\tilde{\boldsymbol{q}}_1| |\tilde{\boldsymbol{p}}_1 |}{2} \right) J_{-1} \left( \frac{|\tilde{\boldsymbol{q}}_2| |\tilde{\boldsymbol{p}}_2|}{2}  \right) \right] =0  \qedhere
\end{aligned}
\end{equation}
\end{proof}

\subsection{When $i \ne j$ and $\Delta n_1 = \pm 2$}

By observing the integral one can easily find that the case with $\Delta n_1 = - 2$ and the one with $\Delta n_1 = 2$ are completely symmetric. So after one of them is solved, the other one can be derived from exchanging the indices. Without losing generality, we can calculate the angular integral $\Theta^{(-2)}(\tilde{Q}_1,\tilde{Q}_2,\tilde{P}_1,\tilde{P}_2)$ with $i \ne j$ and $\Delta n_1 = -2$ first, given by:
\begin{equation}
\begin{aligned}
&  \iiiint_{-\pi}^{\pi}   d \theta_{\tilde{q}_1} d \theta_{\tilde{q}_2} d \theta_{\tilde{p}_1} d \theta_{\tilde{p}_2} e^{-\frac{i}{2}[|\tilde{q}_1| |\tilde{p}_1| \sin(\theta_{\tilde{q}_1} - \theta_{\tilde{p}_1})+ |\tilde{q}_2| | \tilde{p}_2|  \sin(\theta_{\tilde{q}_2} - \theta_{\tilde{p}_2})]} \times \frac{1}{2^4} \times e^{-2 i (\theta_{\tilde{p}_1} - \theta_{\tilde{p}_2})}\\
& \quad  \times 2 |\tilde{\boldsymbol{q}}_1| |\tilde{\boldsymbol{p}}_1 | |\tilde{\boldsymbol{q}}_2| |\tilde{\boldsymbol{p}}_2 | \left(e^{i \theta_{\tilde{q}_1}}+e^{-i  \theta_{\tilde{q}_1}} \right) \left( e^{i  \theta_{\tilde{p}_1}}+e^{-i \theta_{\tilde{p}_1}} \right) \left( e^{i  \theta_{\tilde{q}_2}}+e^{-i  \theta_{\tilde{q}_2}} \right) \left( e^{i \theta_{\tilde{p}_2}}+e^{-i  \theta_{\tilde{p}_2}} \right)  \\
=& \frac{1}{2^3} \iiiint_{-\pi}^{\pi}   d \theta_{\tilde{q}_1} d \theta_{\tilde{q}_2} d \theta_{\tilde{p}_1} d \theta_{\tilde{p}_2} e^{-\frac{i}{2} \left[|\tilde{q}_1| |\tilde{p}_1| \sin(\theta_{\tilde{q}_1} - \theta_{\tilde{p}_1})+ |\tilde{q}_2| | \tilde{p}_2|  \sin(\theta_{\tilde{q}_2} - \theta_{\tilde{p}_2}) \right]}  |\tilde{\boldsymbol{q}}_1| |\tilde{\boldsymbol{p}}_1 | |\tilde{\boldsymbol{q}}_2| |\tilde{\boldsymbol{p}}_2 | \left[e^{i  (\theta_{\tilde{q}_1}-\theta_{\tilde{p}_1})} e^{-i ( \theta_{\tilde{q}_2}-\theta_{\tilde{p}_2})}  \right]\\
=& 2 \pi^4 |\tilde{\boldsymbol{q}}_1| |\tilde{\boldsymbol{p}}_1 | |\tilde{\boldsymbol{q}}_2| |\tilde{\boldsymbol{p}}_2 | \times  
J_1 \left( \frac{|\tilde{\boldsymbol{q}}_1| |\tilde{\boldsymbol{p}}_1 |}{2} \right)
J_1 \left( \frac{|\tilde{\boldsymbol{q}}_2| |\tilde{\boldsymbol{p}}_2 |}{2} \right) = 2 \pi^4 \frac{\tilde{Q}_1 \tilde{P}_1 \tilde{Q}_2 \tilde{P}_2 }{2^4}  \times \ _0\mathbf{F}_1 (;2;-\frac{\tilde{Q}_1 \tilde{P}_1}{16} ) \ _0\mathbf{F}_1 (;2;-\frac{\tilde{Q}_2 \tilde{P}_2}{16} ) 
\end{aligned}
\end{equation}

\begin{lemma}\label{lag_int}
Generalized Laguerre polynomial integral
\begin{equation}
\begin{aligned}
\int_0^{\infty} \frac{d \tilde{P}_1}{2} \left[ L_{a}^{(1)}\left(\frac{\tilde{P}_1}{2}\right)
L_{n_1}^{(-2)}\left(\frac{\tilde{P}_1}{2}\right)
e^{-\frac{1}{2} \tilde{P}_1 }\right] = \delta_{n_1,0} -\delta_{n_1,1} + \delta_{n_1 -2, a} - \delta_{n_1-1, a} 
\end{aligned}
\end{equation}
\end{lemma}

\begin{proposition}
The tensors describing the contribution to $\delta E_{\boldsymbol{q}}$ from the non-diagonal terms is given by
\begin{small}
\begin{equation}
\begin{aligned}
\Gamma^{-}_{m_1 m_2 n_1 n_2}=&   \sqrt{\frac{n_1 ! n_2! }{(n_1-2)! (n_2+2)!}} \times(n_2 +1)(n_2 + 2) \left( \delta_{m_2, n_2}  - \delta_{m_2, n_2+2} \right)\times (\delta_{m_1, n_1} - \delta_{m_1, n_1-2} )\\
\Gamma^{+}_{m_1 m_2 n_1 n_2}=& \sqrt{\frac{n_1 ! n_2! }{(n_1+2)! (n_2-2)!}} \times (n_1 +1)(n_1 + 2)\left( \delta_{m_1, n_1}  -\delta_{m_1, n_1+2} \right)\times (\delta_{m_2, n_2} - \delta_{m_2, n_2-2}  )
\end{aligned}
\end{equation}
\end{small}
\end{proposition}

\begin{proof}
The contribution to the graviton mode gap from these terms can be calculated by:
\begin{equation}
\begin{aligned}
\delta \tilde{E}^{(-2)}_{\boldsymbol{q}}=&-\frac{1}{6 \eta}\int \frac{d \tilde{Q}_1  d \tilde{Q}_2}{4 \times (2 \pi)^4}  \int\frac{d \tilde{P}_1  d \tilde{P}_2}{4 \times 2^4 \times 4 \pi^2}  \underbrace{ \Theta^{(-2)}(\tilde{Q}_1,\tilde{Q}_2,\tilde{P}_1,\tilde{P}_2)\cdot V_{\tilde{Q}_1,\tilde{Q}_2} \cdot \left|\bar{S}_{\boldsymbol{\tilde{p}}_1,\boldsymbol{\tilde{p}}_2} \right|}_{\text{Expanded by generalized L-G polynomials}}
\end{aligned}
\end{equation}
where these three functions can all be expanded by generalized Laguerre polynomials, which gives:
\begin{small}
\begin{equation}
\begin{aligned}
&\frac{2 \pi^4}{2^4} \left[\sum_{a=0}^{\infty}(-1)^{a}   \frac{4 \tilde{Q}_1 \tilde{P}_1}{a+1} 
L_{a}^{(1)}\left(\frac{\tilde{Q}_1}{2}\right) 
L_{a}^{(1)}\left(\frac{\tilde{P}_1}{2}\right) e^{-\frac{1}{4} \left(\tilde{Q}_1 + \tilde{P}_1\right)} \right] \left[\sum_{b=0}^{\infty}(-1)^{b}   \frac{4 \tilde{Q}_2 \tilde{P}_2}{b+1}   
L_{b}^{(1)}\left(\frac{\tilde{Q}_2}{2}\right) 
L_{b}^{(1)}\left(\frac{\tilde{P}_2}{2}\right) e^{-\frac{1}{4} \left(\tilde{Q}_2 + \tilde{P}_2 \right)}\right]\\
& \left[ \sum_{m_1, m_2=0}^{\infty} c^{m_1 m_2}
L_{m_1}^{(0)}\left(\frac{\tilde{Q}_1}{2}\right) 
L_{m_2}^{(0)}\left(\frac{\tilde{Q}_2}{2}\right) 
 e^{-\frac{1}{4} \left(\tilde{Q}_1 +\tilde{Q}_2 \right)}\right]
\left[ \sum_{n_1, n_2=0}^{\infty}(-1)^{n_1 +n_2}  d_{(-2)}^{n_1 n_2} \omega_{-2}^{n_1 n_2} \frac{ \tilde{P}_2}{\tilde{P}_1}  
L_{n_1}^{(-2)}\left(\frac{\tilde{P}_1}{2}\right) 
L_{n_2}^{(2)}\left(\frac{\tilde{P}_2}{2}\right) e^{-\frac{1}{4} \left( \tilde{P}_1+ \tilde{P}_2 \right)}\right]
\end{aligned}
\end{equation}
\end{small}
Then by considering Lemma.\ref{ortho_condition} and \ref{lag_int}, we have
\begin{small}
\begin{equation}
\begin{aligned}
&-\frac{1}{6 \eta}
\frac{2\pi^4 \times 4^2 }{4 \times (2 \pi)^4 \times4 \times 4 \pi^2} 
\sum_{a=0}^{\infty} 
\sum_{m_1, m_2=0}^{\infty} 
\sum_{b=0}^{\infty} 
\sum_{n_1, n_2=0}^{\infty} 
c^{m_1 m_2}     
d_{-2}^{n_1 n_2} (-1)^{a+b+n_1+n_2} \omega_{-2}^{n_1 n_2} 
\frac{1}{a+1} \frac{1}{b+1}\\
& \times\int_0^{\infty} \frac{d \tilde{Q}_1}{2} \left[  \frac{\tilde{Q}_1}{2} 
L_{a}^{(1)}\left(\frac{\tilde{Q}_1}{2}\right) 
L_{m_1}^{(0)}\left(\frac{\tilde{Q}_1}{2}\right)
e^{-\frac{1}{2} \tilde{Q}_1} \right] \times \int_0^{\infty} \frac{d \tilde{Q}_2}{2} \left[ \frac{\tilde{Q}_2 }{2} 
L_{b}^{(1)}\left(\frac{\tilde{Q}_2}{2}\right)
L_{m_2}^{(0)} \left(\frac{\tilde{Q}_2}{2}\right) 
e^{-\frac{1}{2} \tilde{Q}_2 }\right]\\
& \times \int_0^{\infty} \frac{d \tilde{P}_1}{2} \left[  \frac{2}{\tilde{P}_1} \frac{\tilde{P}_1}{2} L_{a}^{(1)}\left(\frac{\tilde{P}_1}{2}\right)
L_{n_1}^{(-2)}\left(\frac{\tilde{P}_1}{2}\right)
e^{-\frac{1}{2} \tilde{P}_1 }\right] \times \int_0^{\infty} \frac{d \tilde{P}_2}{2} \left[  \left(\frac{\tilde{P}_2}{2} \right)^2 
L_{b}^{(1)}\left(\frac{\tilde{P}_2}{2}\right)
L_{n_2}^{(2)}\left(\frac{\tilde{P}_2}{2}\right) 
e^{-\frac{1}{2} \tilde{P}_2 }\right]\\
=&-\frac{1}{2^{6}\times3 \eta \pi^2}
\sum_{a=0}^{\infty} 
\sum_{m_1, m_2=0}^{\infty} 
\sum_{b=0}^{\infty} 
\sum_{n_1, n_2=0}^{\infty} 
c^{m_1 m_2}     
d_{-2}^{n_1 n_2} (-1)^{a+b+n_1+n_2} \omega_{-2}^{n_1 n_2} 
\frac{1}{a+1} \frac{1}{b+1}\\
& \quad \times \left[ (m_1 + 1) \delta_{m_1,a} - m_1 \delta_{m-1, a} \right]
\times \left[ (m_2 + 1) \delta_{m_2, b} - m_2 \delta_{m_2-1, b} \right]\\
& \quad \times \left[\delta_{n_1,0} -\delta_{n_1,1} + \delta_{n_1 -2, a} - \delta_{n_1-1, a} \right]
\times \left[ (n_2 + 1)(n_2 + 2) (\delta_{n_2,b} - \delta_{n_2+1, b}) \right]\\
=&-\frac{1}{2^{6}\times3 \eta \pi^2}
\sum_{m_1, m_2=0}^{\infty} 
\sum_{n_1, n_2=0}^{\infty} 
c^{m_1 m_2}     
d_{-2}^{n_1 n_2} \left [ \omega_{-2}^{n_1 n_2} (n_2+1)(n_2+1) \right.\\
&\times \left. (\delta_{n_1,0} +\delta_{n_1,1} + \delta_{m_1,n_1-2}  - \delta_{m_1,n_1-1} + \delta_{m_1,n_1-1} - \delta_{m_1,n_1})(\delta_{m_2,n_2-2}  - \delta_{m_2,n_2-1} + \delta_{m_2,n_2-1} - \delta_{m_2,n_2}) \right]\\
=& -\frac{1}{2^{6}\times3 \eta \pi^2} \Gamma^{-}_{m_1 m_2 n_1 n_2} c^{m_1 m_2} d_{-2}^{n_1 n_2}
\end{aligned}
\end{equation}
\end{small}
Note that $n^\prime_1 = n_1 - 2 \ge 0$ in this case so $n_1  \ge 2$. Thus both $\delta_{n_1,0} $ and $\delta_{n_1,1}$ vanish. And the final result will be:
\begin{equation}
\begin{aligned}
\Gamma^{-}_{m_1 m_2 n_1 n_2} =  \omega_{-2}^{n_1 n_2} \times(n_2 +1)(n_2 + 2) \left( \delta_{m_2, n_2}  - \delta_{m_2, n_2+2} \right)\times (\delta_{m_1, n_1} - \delta_{m_1, n_1-2} )
\end{aligned}
\end{equation}

Then for $\Delta n_1 = +2$ everything is totally symmetric by substituting $n_1 \leftrightarrow n_2$ and $m_1 \leftrightarrow m_2$. So we can directly write down the result:
\begin{equation}
\begin{aligned}
\Gamma^{+}_{m_1 m_2 n_1 n_2}=& \omega_{+2}^{n_1 n_2} \times (n_1 +1)(n_1 + 2)\left( \delta_{m_1, n_1}  -\delta_{m_1, n_1+2} \right)\times (\delta_{m_2, n_2} - \delta_{m_2, n_2-2}  ) \qedhere
\end{aligned}
\end{equation}
\end{proof}

\subsection{Conclusion}
\begin{definition}
Characteristic tensor
\begin{equation}
\begin{aligned}
 \tilde{\Gamma}_{m_1 m_2 n_1 n_2}^{\text{3bdy}}= \mathcal{C} \times \left(\Gamma_{m_1 m_2 n_1 n_2}^{0} + \frac{\alpha_{n_1+2, n_2-2}}{\alpha_{n_1, n_2}} \times \Gamma_{m_1 m_2 n_1 n_2}^{+}+\frac{\alpha_{n_1-2, n_2+2}}{\alpha_{n_1, n_2}} \times \Gamma_{m_1 m_2 n_1 n_2}^{-} \right)  
\end{aligned}
\end{equation}
where the constant coefficient $\mathcal{C} = -\frac{1}{2^{6}\times3 \eta \pi^2}$.
\end{definition}

\begin{proposition}
The graviton mode gap with respect to the three-body interaction is given by
\begin{equation}
\begin{aligned}
\delta \tilde{E}_{\boldsymbol{q} \rightarrow 0} = \tilde{\Gamma}_{m_1 m_2 n_1 n_2}^{\text{3bdy}} c^{m_1 m_2} \bar{d}^{n_1 n_2}
\end{aligned}
\end{equation}
\end{proposition}
\begin{proof}
From the last three sections we know that the graviton mode gap can be written as:
\begin{equation}
\begin{aligned}
\delta \tilde{E}_{\boldsymbol{q} \rightarrow 0} = \mathcal{C} \times \left(\Gamma_{m_1 m_2 n_1 n_2}^{0} c^{m_1 m_2} d_{0}^{n_1 n_2} + \Gamma_{m_1 m_2 n_1 n_2}^{-} c^{m_1 m_2} d_{-2}^{n_1 n_2} + \Gamma_{m_1 m_2 n_1 n_2}^{+} c^{m_1 m_2} d_{2}^{n_1 n_2}\right)
\end{aligned}
\end{equation}
Then according to Corollary.\ref{regularized_d}, we can use $\bar{d}^{n_1 n_2}$ to express all the expansion coefficients $d_{\Delta n_1}^{n_1 n_2}$:
\begin{equation}
\begin{aligned}
\delta \tilde{E}_{\boldsymbol{q} \rightarrow 0} =& \left[\mathcal{C} \times \left(\Gamma_{m_1 m_2 n_1 n_2}^{0} + \frac{\alpha_{n_1+2, n_2-2}}{\alpha_{n_1, n_2}} \times \Gamma_{m_1 m_2 n_1 n_2}^{+}+\frac{\alpha_{n_1-2, n_2+2}}{\alpha_{n_1, n_2}} \times \Gamma_{m_1 m_2 n_1 n_2}^{-} \right) \right] c^{m_1 m_2} \bar{d}^{n_1 n_2}\\
=&  \left(\tilde{\Gamma}_{m_1 m_2 n_1 n_2}^{0} + \tilde{\Gamma}_{m_1 m_2 n_1 n_2}^{+} + \tilde{\Gamma}_{m_1 m_2 n_1 n_2}^{-} \right)  c^{m_1 m_2} \bar{d}^{n_1 n_2}\\
\equiv & \tilde{\Gamma}_{m_1 m_2 n_1 n_2}^{\text{3bdy}} c^{m_1 m_2} \bar{d}^{n_1 n_2} \qedhere
\end{aligned}
\end{equation}
\end{proof}

The slight differences in the definition of the tensors between the main text and here (up to a constant coefficient) won't influence the conclusion. If the Hamiltonian is defined with different coefficient in the first place, then the constant coefficient $\mathcal{C}$ here might also change. But this will not change the gaplessness of a state and the characteristic tensor $\tilde{\Gamma}_{m_1 m_2 n_1 n_2}^{\text{3bdy}}$ is always well defined so all that one needs to calculate $\delta \tilde{E}_{\boldsymbol{q} \rightarrow 0}$ is to input the corresponding $c^{m_1 m_2}$ and $\bar{d}^{n_1 n_2}$. And the most powerful advantage of this formalism is the structure of the characteristic tensor $\tilde{\Gamma}_{m_1 m_2 n_1 n_2}^{\text{3bdy}}$ allows us to predict the graviton mode energy of some states with respect to a Hamiltonian without knowing the exact values of $c^{m_1 m_2}$ and $\bar{d}^{n_1 n_2}$. Actually we only need to know whether these coefficients vanish or not, as explained in the main text.

\section{Graviton-mode energy with two-body interactions}
In this section, we will show the characteristic matrix formalism of the graviton mode gap with respect to two-body interactions, which was firstly proposed in Ref.\cite{yang2020microscopic}. Based on the same idea, this can be regarded as a special case of the three-body result and most of the lemmas will also be used here.
\begin{definition}
The generic two-body Hamiltonian in a single Landau level (LL)
\begin{equation}
\hat H_{\text{2bdy}}=\int \frac{d^{2} \boldsymbol{q}}{(2 \pi)^2} V_{\bm q}\sum_{i\neq j}e^{i q_a\left(\hat R_i^a-\hat R_j^a\right)}\nonumber
=\int \frac{d^{2} \boldsymbol{q}}{(2 \pi)^2} V_{\bm q}\hat\rho_{\bm q}\hat\rho_{-\bm q} - N_e\int \frac{d^{2} \boldsymbol{q}}{(2 \pi)^2} V_{\bm q}
\end{equation}
where $N_e$ is the number of electrons.
\end{definition}

\begin{definition}
The regularised two-body ground state structure factor:
\begin{equation}
\begin{aligned}
S_{\bm q}&=\frac{1}{N_e}\langle\delta\hat\rho_{\bm q}\delta\hat\rho_{\bm q}\rangle_0=\tilde S_{\bm q}-\frac{1}{N_e}\langle\hat\rho_{\bm q}\rangle_0\langle\hat\rho_{-\bm q}\rangle_0 =1+\sum_{k=0}^\infty d^{\left(0\right)}_{2k+1} L_{2k+1}\left(q^2\right)e^{-\frac{1}{2}q^2}-\frac{1}{N_e}\langle\hat\rho_{\bm q}\rangle_0\langle\hat\rho_{-\bm q}\rangle_0
\end{aligned}
\end{equation}
\end{definition}

\begin{proposition}
The energy of the graviton mode with respect to $\hat H_{\text{2bdy}}$ can be written as
\begin{equation}
\begin{aligned}
\delta \tilde{E}_{\boldsymbol{q} \rightarrow 0 }=\lim_{\bm q\rightarrow 0} \sum_{i \ne j}\int \frac{d^{2} \boldsymbol{q}^\prime}{(2 \pi)^2} V_{\boldsymbol{q}^\prime}\times \frac{\left\langle\psi_{0}\left| \left[\delta \hat{\rho}_{-\boldsymbol{q}},\left[\hat{\rho}^i_{\boldsymbol{q}^\prime} \hat{\rho}^j_{\boldsymbol{-q}^\prime}, \delta \hat{\rho}_{\boldsymbol{q}}\right]\right] \right| \psi_{0}\right\rangle}{2 S_{\boldsymbol{q}}}
\end{aligned}
\end{equation}
\end{proposition}

\begin{proof}
The proof is quite similar to Corollary.\ref{3bcommutator}, so here we adopt a simplified notation and skip some of the same steps:
\begin{small}
\begin{equation}
\begin{aligned}
\delta \tilde{E}_{\boldsymbol{q} \rightarrow 0 }&=\lim_{\bm q\rightarrow 0}\frac{\left\langle\psi_{\boldsymbol{q}}|\hat{H}_{\text{2bdy}}| \psi_{\boldsymbol{q}}\right\rangle}{\left\langle\psi_{\boldsymbol{q}} | \psi_{\boldsymbol{q}}\right\rangle}-E_{0}=\lim_{\bm q\rightarrow 0}\frac{2\left\langle\psi_{\boldsymbol{q}}|\hat{H}_{\text{2bdy}}| \psi_{\boldsymbol{q}}\right\rangle-2 E_0 \left\langle\psi_{\boldsymbol{q}} | \psi_{\boldsymbol{q}}\right\rangle}{2 \left\langle\psi_{\boldsymbol{q}} | \psi_{\boldsymbol{q}}\right\rangle}\\
&=\lim_{\bm q\rightarrow 0}\frac{\left\langle\delta \hat{\rho}_{-\boldsymbol{q}} \hat{H}_{\text{2bdy}} \delta \hat{\rho}_{\boldsymbol{q}}\right\rangle_0+\left\langle\delta \hat{\rho}_{\boldsymbol{q}}\hat{H}_{\text{2bdy}}\delta \hat{\rho}_{-\boldsymbol{q}}\right\rangle_0- \left\langle \delta \hat{\rho}_{-\boldsymbol{q}}\delta \hat{\rho}_{\boldsymbol{q}} \hat{H}_{\text{2bdy}}\right\rangle_0-  \left\langle\hat{H}_{\text{2bdy}}\delta \hat{\rho}_{\boldsymbol{q}}\delta \hat{\rho}_{-\boldsymbol{q}} =\right\rangle_0}{2 \left\langle\psi_{\boldsymbol{q}} | \psi_{\boldsymbol{q}}\right\rangle}\\
&=\lim_{\bm q\rightarrow 0}\frac{\left\langle\psi_{0}\left|\left[\delta \hat{\rho}_{-\boldsymbol{q}},\left[\hat{H}_{\text{2bdy}}, \delta \hat{\rho}_{\boldsymbol{q}}\right]\right]\right| \psi_{0}\right\rangle}{2 S_{\boldsymbol{q}}}=\lim_{\bm q\rightarrow 0} \sum_{i \ne j}\int \frac{d^{2} \boldsymbol{q}^\prime}{(2 \pi)^2} V_{\boldsymbol{q}^\prime}\times \frac{\left\langle\psi_{0}\left| \left[\delta \hat{\rho}_{-\boldsymbol{q}},\left[\hat{\rho}^i_{\boldsymbol{q}^\prime} \hat{\rho}^j_{\boldsymbol{-q}^\prime}, \delta \hat{\rho}_{\boldsymbol{q}}\right]\right] \right| \psi_{0}\right\rangle}{2 S_{\boldsymbol{q}}}
\label{2bdyenergy1}
\end{aligned}
\end{equation}
\end{small}
where we have considered
\begin{equation}
\begin{aligned}
\left[\hat{A},\left[\hat{B},\hat{C}\right]\right]&=\left[\hat{A},\hat{B} \hat{C}- \hat{C} \hat{B}\right]=\hat{A}\hat{B}\hat{C}-\hat{A}\hat{C}\hat{B}-\hat{B}\hat{C}\hat{A}+\hat{C}\hat{B}\hat{A}
\label{3commutator}
\end{aligned}
\end{equation}
and substituted $\hat A=\delta \hat{\rho}_{-\boldsymbol{q}}, \hat B=\hat{H}_{\text{2bdy}}, \hat C=\delta \hat{\rho}_{-\boldsymbol{q}}$ in Eq.\ref{3commutator}
\end{proof}

Note that the constant terms and the particle indices have no contributions to the commutator, thus:
\begin{equation}
\begin{aligned}
\left[\delta \hat{\rho}_{-\boldsymbol{q}},\left[\hat{\rho}^i_{\boldsymbol{q}^\prime} \hat{\rho}^j_{\boldsymbol{-q}^\prime}, \delta \hat{\rho}_{\boldsymbol{q}}\right]\right]
\cong \left[\delta \hat{\rho}_{-\boldsymbol{q}},\left[\delta \hat{\rho}_{\boldsymbol{q}^\prime} \delta \hat{\rho}_{\boldsymbol{-q}^\prime}, \delta \hat{\rho}_{\boldsymbol{q}}\right]\right] (*)
\end{aligned}
\end{equation}
Then we have
\begin{small}
\begin{equation}
\begin{aligned}
(*)=&[\delta \hat{\rho}_{-\boldsymbol{q}},\delta \hat{\rho}_{\boldsymbol{q}’}[\delta \hat{\rho}_{-\boldsymbol{q}’},\delta \hat{\rho}_{\boldsymbol{q}}]+[\delta \hat{\rho}_{\boldsymbol{q}’},\delta \hat{\rho}_{\boldsymbol{q}}]\delta \hat{\rho}_{-\boldsymbol{q}’}]=[\delta \hat{\rho}_{-\boldsymbol{q}},\delta \hat{\rho}_{\boldsymbol{q}’}[\delta \hat{\rho}_{-\boldsymbol{q}’},\delta \hat{\rho}_{\boldsymbol{q}}]]+[\delta \hat{\rho}_{-\boldsymbol{q}},[\delta \hat{\rho}_{\boldsymbol{q}’},\delta \hat{\rho}_{\boldsymbol{q}}]\delta \hat{\rho}_{-\boldsymbol{q}’}]\\
=&[\delta \hat{\rho}_{-\boldsymbol{q}},\delta \hat{\rho}_{\boldsymbol{q}’}][\delta \hat{\rho}_{-\boldsymbol{q}’},\delta \hat{\rho}_{\boldsymbol{q}}]+\delta \hat{\rho}_{\boldsymbol{q}’}[\delta \hat{\rho}_{-\boldsymbol{q}},[\delta \hat{\rho}_{-\boldsymbol{q}’},\delta \hat{\rho}_{\boldsymbol{q}}]]+[\delta \hat{\rho}_{-\boldsymbol{q}},[\delta \hat{\rho}_{\boldsymbol{q}’},\delta \hat{\rho}_{\boldsymbol{q}}]]\delta \hat{\rho}_{-\boldsymbol{q}’}+[\delta \hat{\rho}_{\boldsymbol{q}’},\delta \hat{\rho}_{\boldsymbol{q}}][\delta \hat{\rho}_{-\boldsymbol{q}},\delta \hat{\rho}_{-\boldsymbol{q}’}]\\
=& \left( 2 i \sin \frac{-\boldsymbol{q} \times \boldsymbol{q}'}{2} \delta \hat{\rho}_{-\boldsymbol{q}+\boldsymbol{q}'} \right) 
\left(2 i \sin \frac{-\boldsymbol{q}' \times \boldsymbol{q}}{2} \delta \hat{\rho}_{-\boldsymbol{q}'+\boldsymbol{q}} \right)
+\delta \hat{\rho}_{\boldsymbol{q}’}
\left[\delta \hat{\rho}_{-\boldsymbol{q}}, \left(2 i \sin \frac{-\boldsymbol{q}' \times \boldsymbol{q}}{2} \delta \hat{\rho}_{-\boldsymbol{q}'+\boldsymbol{q}} \right) \right]\\
&+ \left[\delta \hat{\rho}_{-\boldsymbol{q}},
\left(2 i \sin \frac{\boldsymbol{q}' \times \boldsymbol{q}}{2} \delta \hat{\rho}_{\boldsymbol{q}'+\boldsymbol{q}}\right) \right]
\delta \hat{\rho}_{-\boldsymbol{q}’}+
\left( 2 i \sin \frac{\boldsymbol{q}' \times \boldsymbol{q}}{2} \delta \hat{\rho}_{\boldsymbol{q}'+\boldsymbol{q}} \right)
\left(2 i \sin \frac{ -\boldsymbol{q} \times -\boldsymbol{q}'}{2} \delta \hat{\rho}_{-\boldsymbol{q}-\boldsymbol{q}'} \right)\\
=& \left(-2 \sin \frac{\boldsymbol{q}' \times \boldsymbol{q}}{2} \right)^2 \delta \hat{\rho}_{-\boldsymbol{q}+\boldsymbol{q}'} \delta \hat{\rho}_{-\boldsymbol{q}+\boldsymbol{q}'} 
+ \left(2 \sin \frac{\boldsymbol{q}' \times \boldsymbol{q}}{2} \right)^2 \delta \hat{\rho}_{\boldsymbol{q}+\boldsymbol{q}'} \delta \hat{\rho}_{-\boldsymbol{q}-\boldsymbol{q}'}\\
&-\delta \hat{\rho}_{\boldsymbol{q}’}
\left[\delta \hat{\rho}_{-\boldsymbol{q}}, \left(2 i \sin \frac{-\boldsymbol{q}' \times \boldsymbol{q}}{2} \delta \hat{\rho}_{-\boldsymbol{q}'+\boldsymbol{q}} \right)\right] 
- \left[\delta \hat{\rho}_{-\boldsymbol{q}}, \left(2 i \sin \frac{\boldsymbol{q}' \times \boldsymbol{q}}{2} \delta \hat{\rho}_{\boldsymbol{q}'+\boldsymbol{q}} \right) \right]
\delta \hat{\rho}_{-\boldsymbol{q}’}\\
=& 4 \sin^2 \frac{\boldsymbol{q}' \times \boldsymbol{q}}{2} 
\delta \hat{\rho}_{-\boldsymbol{q}+\boldsymbol{q}'} \delta \hat{\rho}_{-\boldsymbol{q}+\boldsymbol{q}'} 
+ 4 \sin^2 \frac{\boldsymbol{q}' \times \boldsymbol{q}}{2} \delta \hat{\rho}_{\boldsymbol{q}+\boldsymbol{q}'} \delta \hat{\rho}_{-\boldsymbol{q}-\boldsymbol{q}'}\\
&- 2 i \sin \frac{-\boldsymbol{q}' \times \boldsymbol{q}}{2} \delta \hat{\rho}_{\boldsymbol{q}’}
\left[\delta \hat{\rho}_{-\boldsymbol{q}}, \delta \hat{\rho}_{-\boldsymbol{q}'+\boldsymbol{q}} \right] 
- 2 i \sin \frac{\boldsymbol{q}' \times \boldsymbol{q}}{2} 
\left[\delta \hat{\rho}_{-\boldsymbol{q}},\delta \hat{\rho}_{\boldsymbol{q}'+\boldsymbol{q}} \right]
\delta \hat{\rho}_{-\boldsymbol{q}’}\\
=&  4 \sin^2 \frac{\boldsymbol{q}' \times \boldsymbol{q}}{2}
\delta \hat{\rho}_{-\boldsymbol{q}+\boldsymbol{q}'} \delta \hat{\rho}_{-\boldsymbol{q}+\boldsymbol{q}'}+ 4 \sin^2 \frac{\boldsymbol{q}' \times \boldsymbol{q}}{2} \delta \hat{\rho}_{\boldsymbol{q}+\boldsymbol{q}'} \delta \hat{\rho}_{-\boldsymbol{q}-\boldsymbol{q}'}\\
&-\left(2 i \sin \frac{-\boldsymbol{q}' \times \boldsymbol{q}}{2} \right)
(2 i \sin \frac{-\boldsymbol{q} \times (-\boldsymbol{q}'+\boldsymbol{q})}{2}) \delta \hat{\rho}_{\boldsymbol{q}’} \delta \hat{\rho}_{-\boldsymbol{q}’}-(2 i \sin \frac{\boldsymbol{q}' \times \boldsymbol{q}}{2})(2 i \sin \frac{-\boldsymbol{q} \times (\boldsymbol{q}'+\boldsymbol{q})}{2})\delta \hat{\rho}_{\boldsymbol{q}’} \delta \hat{\rho}_{-\boldsymbol{q}’}
\end{aligned}
\end{equation}
\end{small}
where we have considered:
\begin{equation}
\begin{aligned}
\left[\hat D,\left[\hat A \hat B, \hat C \right ]\right]&=\left[\hat D,\hat  A\left[\hat B, \hat C \right]+\left[\hat A,\hat C\right]\hat B\right]=\left[\hat D, \hat A \left[\hat B,\hat C \right] \right]+\left[\hat D,\left[\hat A, \hat C\right]\hat B \right]\\
&=\left[\hat D, \hat A \right] \left[\hat B, \hat C \right] + \hat A \left[ \hat D, \left[\hat B,\hat C \right]\right] + \left[\hat D, \left[\hat A, \hat C \right]\right]\hat B + \left[\hat A,\hat C \right] \left[\hat D, \hat B\right]
\end{aligned}
\end{equation}
with $
\hat A=\delta \hat{\rho}_{\boldsymbol{q'}}, 
\hat B=\delta \hat{\rho}_{-\boldsymbol{q'}}, 
\hat C=\delta \hat{\rho}_{\boldsymbol{q}}, 
\hat D=\delta \hat{\rho}_{-\boldsymbol{q}}
$ and the GMP algebra in Eq.\ref{GMP} which shows us the commutator of the regularised guiding center density operators with momentum $\boldsymbol{q}_1$ and $\boldsymbol{q}_2$ gives a new density operator with the sum of momentum $\boldsymbol{q}_1$ and $\boldsymbol{q}_2$.

By omitting the second order terms of $q$ we have:
\begin{small}
\begin{equation}
\begin{aligned}
\left[\delta \rho_{-\boldsymbol{q}},\left[\delta \rho_{\boldsymbol{q'}}\delta \rho_{-\boldsymbol{q'}},\delta \rho_{\boldsymbol{q}}\right]\right]&=  
\left(2 \sin \frac{\boldsymbol{q}' \times \boldsymbol{q}}{2} \right)^2 \delta \bar{\rho}_{-\boldsymbol{q}+\boldsymbol{q}'} \delta \bar{\rho}_{-\boldsymbol{q}+\boldsymbol{q}'}
+ \left(2 \sin \frac{\boldsymbol{q}' \times \boldsymbol{q}}{2} \right)^2 \delta \bar{\rho}_{\boldsymbol{q}+\boldsymbol{q}'} \delta \bar{\rho}_{-\boldsymbol{q}-\boldsymbol{q}'}
-2 \left(2 \sin \frac{\boldsymbol{q}' \times \boldsymbol{q}}{2} \right)^2 \delta \rho_{\boldsymbol{q'}} \delta \rho_{\boldsymbol{-q'}}\\
&= \left(2 \sin \frac{\boldsymbol{q}' \times \boldsymbol{q}}{2} \right)^2(\delta \bar{\rho}_{-\boldsymbol{q}+\boldsymbol{q}'} \delta \bar{\rho}_{-\boldsymbol{q}+\boldsymbol{q}'}+ \delta \bar{\rho}_{\boldsymbol{q}+\boldsymbol{q}'} \delta \bar{\rho}_{-\boldsymbol{q}-\boldsymbol{q}'}-2 \delta \rho_{\boldsymbol{q'}} \delta \rho_{\boldsymbol{-q'}})
\end{aligned}
\end{equation}
\end{small}
Thus the expectation value of the commutator with respect to the ground state is given by:
\begin{equation}
\begin{aligned}
\left\langle\psi_{0}|\left[\delta \rho_{-\boldsymbol{q}},\left[\delta \rho_{\boldsymbol{q'}}\delta \rho_{-\boldsymbol{q'}},\delta \rho_{\boldsymbol{q}}\right]\right]|\psi_{0}\right\rangle
&= \left(2 \sin \frac{\boldsymbol{q}' \times \boldsymbol{q}}{2} \right)^2
\left\langle\psi_{\boldsymbol{q}}|
\left(\delta \bar{\rho}_{-\boldsymbol{q}+\boldsymbol{q}'} \delta \bar{\rho}_{-\boldsymbol{q}+\boldsymbol{q}'}+ \delta \bar{\rho}_{\boldsymbol{q}+\boldsymbol{q}'} \delta \bar{\rho}_{-\boldsymbol{q}-\boldsymbol{q}'}-2 \delta \rho_{\boldsymbol{q'}} \delta \rho_{-\boldsymbol{q'}} \right)
|\psi_{\boldsymbol{q}}\right\rangle\\
&= \left(2 \sin \frac{\boldsymbol{q}' \times \boldsymbol{q}}{2} \right)^2 (S_{\boldsymbol{q}'+\boldsymbol{q}}+S_{\boldsymbol{q}'-\boldsymbol{q}}-2S_{\boldsymbol{q}'})
\end{aligned}
\end{equation}
So the commutator in Eq.\ref{2bdyenergy1} turns out to be:

\begin{equation}
\begin{aligned}
&\left\langle\psi_{0}|\left[\delta \rho_{-\boldsymbol{q}},\left[\hat H_{\text{2bdy}}, \delta \rho_{\boldsymbol{q}}\right]\right]|\psi_{0}\right\rangle\\
=&\int \frac{d^{2} \boldsymbol{q}'}{4 \pi} V_{\boldsymbol{q'}}\left\langle\psi_{\boldsymbol{q}}|\left[\delta \rho_{-\boldsymbol{q}},\left[\delta \rho_{\boldsymbol{q'}}\delta \rho_{-\boldsymbol{q'}},\delta \rho_{\boldsymbol{q}}\right]\right]|\psi_{\boldsymbol{q}}\right\rangle=\int \frac{d^{2} \boldsymbol{q}'}{4 \pi} V_{\boldsymbol{q'}}(2 \sin \frac{\boldsymbol{q}' \times \boldsymbol{q}}{2})^2(S_{\boldsymbol{q}'+\boldsymbol{q}}+S_{\boldsymbol{q}'-\boldsymbol{q}}-2S_{\boldsymbol{q}'})
\\
=&\int \frac{d^{2} \boldsymbol{q}'}{4 \pi} V_{\boldsymbol{q'}}
\left(2 \sin \frac{\boldsymbol{q}' \times \boldsymbol{q}}{2} \right)^2 [(S_{\boldsymbol{q}'+\boldsymbol{q}}-S_{\infty})+(S_{\boldsymbol{q}'-\boldsymbol{q}}-S_{\infty})-2(S_{\boldsymbol{q}'}-S_{\infty})]\\
=&\int \frac{d^{2} \boldsymbol{q}'}{4 \pi} V_{\boldsymbol{q'}}(2 \sin \frac{\boldsymbol{q}' \times \boldsymbol{q}}{2})^2(s_{\boldsymbol{q}'+\boldsymbol{q}}+s_{\boldsymbol{q}'-\boldsymbol{q}}-2s_{\boldsymbol{q}'})
\end{aligned}
\end{equation}
which gives the graviton mode gap as:
\begin{small}
\begin{equation}
\begin{aligned}
\delta \tilde{E}_{\boldsymbol{q} \rightarrow 0 }
&=\lim_{\bm q\rightarrow 0} \frac{\left\langle\psi_{0}\left|\left[\delta \rho_{-\boldsymbol{q}},\left[\hat H_{\text{2bdy}}, \delta \rho_{\boldsymbol{q}}\right]\right]\right| \psi_{0}\right\rangle}{2 S_{\boldsymbol{q}}} =\lim_{\bm q\rightarrow 0} \frac{1}{2S_q}\int \frac{d^{2} q^{\prime}}{(2 \pi)^2} V_{\boldsymbol{q}^{\prime}}\left[2 \sin \left(\frac{1}{2} \boldsymbol{q}^{\prime} \times \boldsymbol{q}\right)\right]^{2} \times \left(\boldsymbol{s}_{\boldsymbol{q}^{\prime}+\boldsymbol{q}}+s_{\boldsymbol{q}^{\prime}-\boldsymbol{q}}-2 s_{\boldsymbol{q}^{\prime}}\right)
\label{energygap}
\end{aligned}
\end{equation}
\end{small}

Meanwhile the linear combination of the structure factor terms can be expanded around $|\boldsymbol{q}|=0$:
\begin{small}
\begin{equation}
\begin{aligned}
\boldsymbol{s}_{\boldsymbol{q}^{\prime}+\boldsymbol{q}}+s_{\boldsymbol{q}^{\prime}-\boldsymbol{q}}-2 s_{\boldsymbol{q}^{\prime}} =& \boldsymbol{s}_{\boldsymbol{q}^{\prime}}+q_x \frac{\partial s_{\boldsymbol{q}^{\prime}}}{\partial q^{\prime}_{x}}+q_y \frac{\partial s_{\boldsymbol{q}^{\prime}}}{\partial q^{\prime}_{y}}+\frac{1}{2!}
\left(q^2_x \frac{\partial^2 s_{\boldsymbol{q}^{\prime}}}{\partial q^{\prime 2}_{x}}+2q_x q_y \frac{\partial^2 s_{\boldsymbol{q}^{\prime}}}{\partial q^{\prime}_{x} \partial q^{\prime}_{y}}+q^2_y \frac{\partial^2 s_{\boldsymbol{q}^{\prime}}}{\partial q^{\prime 2}_{y}} \right)
+O(\boldsymbol{q}^3)\\
&+\boldsymbol{s}_{\boldsymbol{q}^{\prime}}-q_x \frac{\partial s_{\boldsymbol{q}^{\prime}}}{\partial q^{\prime}_{x}}-q_y \frac{\partial s_{\boldsymbol{q}^{\prime}}}{\partial q^{\prime}_{y}}+\frac{1}{2!}
\left(q^2_x \frac{\partial^2 s_{\boldsymbol{q}^{\prime}}}{\partial q^{\prime 2}_{x}}+2q_x q_y \frac{\partial^2 s_{\boldsymbol{q}^{\prime}}}{\partial q^{\prime}_{x} \partial q^{\prime}_{y}}+q^2_y \frac{\partial^2 s_{\boldsymbol{q}^{\prime}}}{\partial q^{\prime 2}_{y}} \right)
+O(\boldsymbol{q}^3)-2\boldsymbol{s}_{\boldsymbol{q}^{\prime}}\\
=&q^2_x \frac{\partial^2 s_{\boldsymbol{q}^{\prime}}}{\partial q^{\prime 2}_{x}}+2q_x q_y \frac{\partial^2 s_{\boldsymbol{q}^{\prime}}}{\partial q^{\prime}_{x} \partial q^{\prime}_{y}}+q^2_y \frac{\partial^2 s_{\boldsymbol{q}^{\prime}}}{\partial q^{\prime 2}_{y}}+O(\boldsymbol{q}^3)= \left(q_x \frac{\partial}{\partial q^{\prime}_{x}}+q_y \frac{\partial}{\partial q^{\prime}_{y}} \right)^2 s_{\boldsymbol{q}^{\prime}}+O(\boldsymbol{q}^3) \sim (q \cdot \boldsymbol{\nabla}^\prime)^2  s_{\boldsymbol{q}^{\prime}}
\label{sqqp}
\end{aligned}
\end{equation}
\end{small}
where the third-order and higher-order terms in Eq.\ref{sqqp} have been dropped and by taking the Fourier transformation of $s_{\boldsymbol{q}^{\prime}}$:
\begin{equation}
\begin{aligned}
&-(\boldsymbol{q} \cdot \boldsymbol{\nabla}^\prime)^2 \left[ \int\frac{d^2 q^{\prime \prime}}{2\pi} e^{i \boldsymbol{q}^{\prime} \times \boldsymbol{q}^{\prime \prime}}s_{\boldsymbol{q}^{\prime \prime}}\right]=-\int\frac{d^2 q^{\prime \prime}}{2\pi}(\boldsymbol{q} \cdot \boldsymbol{\nabla}^\prime) \cdot i(q_x q^{\prime \prime}_y -q_y q^{\prime \prime}_x)\cdot e^{i \boldsymbol{q}^{\prime} \times \boldsymbol{q}^{\prime \prime}}s_{\boldsymbol{q}^{\prime \prime}}\\
=&-i\int\frac{d^2 q^{\prime \prime}}{2\pi} (\boldsymbol{q} \times \boldsymbol{q}^{\prime \prime})\cdot (\boldsymbol{q} \cdot \boldsymbol{\nabla}^\prime) e^{i \boldsymbol{q}^{\prime} \times \boldsymbol{q}^{\prime \prime}}s_{\boldsymbol{q}^{\prime \prime}}=\int\frac{d^2 \boldsymbol{q}^{\prime \prime}}{2\pi} (\boldsymbol{q} \times \boldsymbol{q}^{\prime \prime})^2\cdot e^{i \boldsymbol{q}^{\prime} \times \boldsymbol{q}^{\prime \prime}}s_{\boldsymbol{q}^{\prime \prime}}\\
\label{sqqp1}
\end{aligned}
\end{equation}

Similarly we can expand the $\left[2 \sin \left(\frac{1}{2} \boldsymbol{q}^{\prime} \times \boldsymbol{q}\right)\right]^{2}$ term around $|\boldsymbol{q}|=0$ as:
\begin{equation}
\begin{aligned}
&\left[2 \sin \left(\frac{1}{2} \boldsymbol{q}^{\prime} \times \boldsymbol{q}\right)\right]^{2}=4 \sin^2 \left(\frac{1}{2} \boldsymbol{q}^{\prime} \times \boldsymbol{q}\right) =4 \sin^2 \left(\frac{1}{2} (q^{\prime}_x q_y -q^{\prime}_y q_x)\right)\\
=& 4 \left[\sin^2(0) - q_x \cdot 2\sin \left(\frac{1}{2} (q^{\prime}_x q_y -q^{\prime}_y q_x)\right)\cdot \cos \left(\frac{1}{2} (q^{\prime}_x q_y -q^{\prime}_y q_x)\right)\cdot \frac{q^{\prime}_y}{2} \right.\\
& \left. + q_y \cdot 2\sin \left(\frac{1}{2} (q^{\prime}_x q_y - q^{\prime}_y q_x)\right)\cdot \cos \left(\frac{1}{2} (q^{\prime}_x q_y - q^{\prime}_y q_x)\right) \cdot \frac{q^{\prime}_x}{2} + O(q^2) \right]\\
=& 4 \left[ \left(\frac{ q^{\prime}_x q_y}{2} - \frac{q^{\prime}_y q_x}{2} \right) \sin \left( 2 \left(\frac{q^{\prime}_x q_y}{2} - \frac{q^{\prime}_y q_x}{2} \right) \right)+O(\boldsymbol{q}^2 )\right]= 4\left[ \left(\frac{\boldsymbol{q}^{\prime} \times \boldsymbol{q}}{2} \right)\sin(\boldsymbol{q}^{\prime} \times \boldsymbol{q})|_{\boldsymbol{q}=0}+O(\boldsymbol{q}^2)\right]\\
=& 4 \left[ \left(\frac{\boldsymbol{q}^{\prime} \times \boldsymbol{q}}{2} \right)\sin(\boldsymbol{q}^{\prime} \times \boldsymbol{q})|_{\boldsymbol{q}=0} + q^2_x \cdot \cos(\boldsymbol{q}^{\prime} \times \boldsymbol{q}) \left(-\frac{q^{\prime}_y}{2} \right)^2 \right.\\
& + \left. q^2_y \cdot \cos(\boldsymbol{q}^{\prime} \times \boldsymbol{q}) \left(\frac{q^{\prime}_x}{2} \right)^2 + 2 q_x q_y \cdot \cos(\boldsymbol{q}^{\prime} \times \boldsymbol{q}) \left(\frac{-q^{\prime}_y}{2} \right) \left(\frac{q^{\prime}_x}{2} \right)+O(\boldsymbol{q}^3) \right]\\
=& 4\left\{ \left(\frac{\boldsymbol{q}^{\prime} \times \boldsymbol{q}}{2} \right)\sin(\boldsymbol{q}^{\prime} \times \boldsymbol{q})|_{\boldsymbol{q}=0}+\cos(\boldsymbol{q}^{\prime} \times \boldsymbol{q}) \cdot \left[q^2_x \cdot \left(-\frac{q^{\prime}_y}{2} \right)^2 + q^2_y \cdot \left(\frac{q^{\prime}_x}{2} \right)^2+ 2 q_x q_y \cdot \left(-\frac{q^{\prime}_y}{2}\right) \left(\frac{q^{\prime}_x}{2}\right) \right]+O(\boldsymbol{q}^3)\right\}\\
=& 4\left[\cos(\boldsymbol{q}^{\prime} \times \boldsymbol{q})|_{\boldsymbol{q}=0} \cdot \left(\frac{q^{\prime}_x q_y -q^{\prime}_y q_x}{2} \right)^2+O(\boldsymbol{q}^3)\right] =(\boldsymbol{q}^{\prime} \times \boldsymbol{q})^2+O(\boldsymbol{q}^3)
\label{sinqqp}
\end{aligned}
\end{equation}

By substituting Eq.\ref{sqqp1} and Eq.\ref{sinqqp} in the integral Eq.\ref{energygap} abd considering Theorem.\ref{long-wavelength_limit}, we get:
\begin{equation}
\begin{aligned}
\delta \tilde{E}_{\boldsymbol{q} \rightarrow 0}&=\frac{1}{2S_q}\int \frac{d^{2} q^{\prime}}{(2 \pi)^2} V_{\boldsymbol{q}^{\prime}}\left[2 \sin \left(\frac{1}{2} \boldsymbol{q}^{\prime} \times \boldsymbol{q}\right)\right]^{2} \times \left(\boldsymbol{s}_{\boldsymbol{q}^{\prime}+\boldsymbol{q}}+s_{\boldsymbol{q}^{\prime}-\boldsymbol{q}}-2 s_{\boldsymbol{q}^{\prime}}\right)\\
&=\frac{1}{2 \eta q^{4}} \iint \frac{d^{2} q^{\prime} d^{2} q^{\prime \prime}}{8 \pi^{3}} V_{\boldsymbol{q}^{\prime}}(\boldsymbol{q} \times \boldsymbol{q}^{\prime})^2 (\boldsymbol{q} \times \boldsymbol{q}^{\prime \prime})^2 e^{i \boldsymbol{q}^{\prime} \times \boldsymbol{q}^{\prime \prime}} s_{\boldsymbol{q}^{\prime \prime}} \\
&=\frac{q^{4}}{2 \eta q^{4}} \iint \frac{d^{2} q^{\prime} d^{2} q^{\prime \prime}}{8 \pi^{3}} V_{\boldsymbol{q}^{\prime}}|\boldsymbol{q}^{\prime}|^2 |\boldsymbol{q}^{\prime \prime}|^2 \sin^2(\theta_{\boldsymbol{q}}-\theta_{\boldsymbol{q}^{\prime}}) \sin^2(\theta_{\boldsymbol{q}}-\theta_{\boldsymbol{q}^{\prime \prime}}) e^{i \boldsymbol{q}^{\prime} \times \boldsymbol{q}^{\prime \prime}} s_{\boldsymbol{q}^{\prime \prime}}
\label{eg1}
\end{aligned}
\end{equation}

It is worth noticing that Eq.\ref{eg1} only depends on the relative angles between $\theta_{\boldsymbol{q}^{\prime}}$ $\theta_{\boldsymbol{q}^{\prime \prime}}$ and $\theta_{\boldsymbol{q}}$. Thus for simplicity we can set $\theta_{\boldsymbol{q}}=\frac{\pi}{2}$ without losing generality:

\begin{equation}
\begin{aligned}
\delta \tilde{E}_{\boldsymbol{q} \rightarrow 0}&=\frac{q^{4}}{2 \eta q^{4}} \iint \frac{d^{2} q^{\prime} d^{2} q^{\prime \prime}}{8 \pi^{3}} V_{\boldsymbol{q}^{\prime}}|\boldsymbol{q}^{\prime}|^2 |\boldsymbol{q}^{\prime \prime}|^2 \sin^2(\frac{\pi}{2}-\theta_{\boldsymbol{q}^{\prime}}) \sin^2(\frac{\pi}{2}-\theta_{\boldsymbol{q}^{\prime \prime}}) e^{i \boldsymbol{q}^{\prime} \times \boldsymbol{q}^{\prime \prime}} s_{\boldsymbol{q}^{\prime \prime}}\\
&=\frac{1}{2 \eta} \iint \frac{d^{2} q^{\prime} d^{2} q^{\prime \prime}}{8 \pi^{3}} V_{\boldsymbol{q}^{\prime}}\left(q_{x}^{\prime} \cdot q_{x}^{\prime \prime}\right)^{2} e^{i \boldsymbol{q}^{\prime} \times \boldsymbol{q}^{\prime \prime}} s_{\boldsymbol{q}^{\prime \prime}}
\label{eg2}
\end{aligned}
\end{equation}
where $\left(q_{x}^{\prime} \cdot q_{x}^{\prime \prime}\right)^{2}$ can also be replaced by $\left(q_{y}^{\prime} \cdot q_{y}^{\prime \prime}\right)^{2}$. Then we can integrate the angular part of this integral:
\begin{equation}
\begin{aligned}
\delta \tilde{E}_{\boldsymbol{q} \rightarrow 0}&=\frac{1}{2 \eta} \iint \frac{d^{2} q^{\prime} d^{2} q^{\prime \prime}}{8 \pi^{3}} V_{\boldsymbol{q}^{\prime}}\left(q_{x}^{\prime} q_{x}^{\prime \prime}\right)^{2} e^{i \boldsymbol{q}^{\prime} \times \boldsymbol{q}^{\prime \prime}} s_{\boldsymbol{q}^{\prime \prime}}\\
&=\frac{1}{2 \eta} \iint_{0}^{\infty} \frac{2|\boldsymbol{q}^{\prime}|d |\boldsymbol{q}^{\prime}| \cdot 2|\boldsymbol{q}^{\prime \prime}|d |\boldsymbol{q}^{\prime \prime}|}{4\times 8 \pi^{3}} V_{\boldsymbol{q}^\prime} s_{\boldsymbol{q}^{\prime \prime}} |\boldsymbol{q}^{\prime}|^2 |\boldsymbol{q}^{\prime \prime}|^2 \iint_{-\pi}^{\pi} d \theta_1 d \theta_2 \cos^2(\theta_1) \cos^2(\theta_2) e^{i |\boldsymbol{q}^{\prime}| |\boldsymbol{q}^{\prime \prime}|\sin(\theta_1 - \theta_2)}\\
&=\frac{1}{2 \eta} \left(\iint_{0}^{\infty} \frac{d q_1 d q_2}{32 \pi^{3}} q_1 q_2 V_{q_1} s_{q_2} \right) \cdot \left[ \iint_{-\pi}^{\pi} d \theta_1 d \theta_2 \cos^2(\theta_1) \cos^2(\theta_2) e^{i |\boldsymbol{q}^{\prime}| |\boldsymbol{q}^{\prime \prime}|\sin(\theta_1 - \theta_2)} \right]
\label{eg3}
\end{aligned}
\end{equation}
where we have defined $q_1 = |\boldsymbol{q}^{\prime}|^2$ and $q_2 = |\boldsymbol{q}^{\prime \prime}|^2$. The direct transformation of $ V_{\boldsymbol{q}^\prime} $ and $s_{\boldsymbol{q}^{\prime \prime}}$ is because both of them are the function of $q_1$ and $q_2$ as Eq.\ref{2bexpansion} shows. After the variable separation we can only focus on the angular integral:
\begin{small}
\begin{equation}
\begin{aligned}
&\iint_{-\pi}^{\pi} d \theta_1 d \theta_2 \cos^2(\theta_1) \cos^2(\theta_2) e^{i |\boldsymbol{q}^{\prime}| |\boldsymbol{q}^{\prime \prime}|\sin(\theta_1 - \theta_2)}=\frac{1}{4}\iint_{-\pi}^{\pi} d \theta_1 d \theta_2 \frac{(2+ e^{i 2 \theta_1} + e^{-i 2 \theta_1}) (2+ e^{i 2 \theta_2} + e^{-i 2 \theta_2})}{4}   e^{i |\boldsymbol{q}^{\prime}| |\boldsymbol{q}^{\prime \prime}|\sin(\theta_1 - \theta_2)}\\
=&\frac{1}{4}\left[\frac{4\pi^2 J_0(|\boldsymbol{q}^{\prime}| |\boldsymbol{q}^{\prime \prime}|)+ 2 \pi^2 J_2 (|\boldsymbol{q}^{\prime}| |\boldsymbol{q}^{\prime \prime}|)}{4} \right]=\frac{\pi^2}{4} \left[ J_0(|\boldsymbol{q}^{\prime}| |\boldsymbol{q}^{\prime \prime}|)+ \frac{ J_2 (|\boldsymbol{q}^{\prime}| |\boldsymbol{q}^{\prime \prime}|)}{2}\right]=\frac{\pi^2}{8} \left[ _0\mathbf{F}_1(;1;-\frac{q_1 q_2}{4})+_0\mathbf{F}_1(;2;-\frac{q_1 q_2}{4}) \right]\\
\label{angleint}
\end{aligned}
\end{equation}
\end{small}
where $J_\alpha(x)$ are the Bessel functions of the first kind and we have considered Lemma.\ref{Bessel_Hypergeometric}. 

\begin{definition}
Characteristic matrix
\begin{equation}
\begin{aligned}
\Gamma^{\text{2bdy}}_{m n} = (-1)^{m} \cdot \left[2(m^2+m+1)\delta_{m,n}-(m+1)(m+2)\delta_{m,n-2}-m(m-1) \delta_{m,n+2} \right]
\end{aligned}
\end{equation}
\end{definition}

\begin{proposition}
The graviton mode gap with respect to the two-body interaction is given by
\begin{equation}
\delta \tilde{E}_{\boldsymbol{q} \rightarrow 0} =\frac{1}{2^8 \eta \pi} \Gamma^{\text{2bdy}}_{m n} c^{m} d^{n}
\end{equation}
\end{proposition}
\begin{proof}
By substituting Eq.\ref{angleint} to Eq.\ref{eg3} and considering Lemma.\ref{Fourier}, we have
\begin{equation}
\delta \tilde{E}_{\boldsymbol{q} \rightarrow 0}
= \frac{1}{512 \eta \pi} \iint_{0}^{\infty} d q_{1} d q_{2} V_{q_{1}}\left(S_{q_{2}}-S_{\infty}\right) q_{1} q_{2}  \times\left[ _{0} \mathbf{F}_{1}\left(1,-\frac{q_{1} q_{2}}{4}\right)+_{0} \mathbf{F}_{1}\left(2,-\frac{q_{1} q_{2}}{4}\right)\right]
\label{eg4}
\end{equation}

Here the structure factor and the interaction can also be expanded as in the three-body cases:
\begin{equation}
V_{\boldsymbol{q}_1}= \sum_{m} c^{m}  L_{m}(q_1) e^{-\frac{q_1}{2}}, \quad s_{\boldsymbol{q}_2}=\sum_{n} d^{n} L_{n}(q_2) e^{-\frac{q_2}{2}}
\label{2bexpansion}
\end{equation}
to Eq.\ref{eg4} (where $m$ and $n$ are both odd and positive integers), we have
\begin{equation}
\delta \tilde{E}_{\boldsymbol{q} \rightarrow 0}= \frac{1}{512 \eta \pi}\sum_{m} \sum_{n} c^{m} d^{n} \iint_{0}^{\infty} d q_{1} d q_{2} e^{-\frac{q_{1}+q_{2}}{2}}  q_{1} q_{2} L_{m}(q_1) L_{n}(q_2)  \times\left[ _{0} \mathbf{F}_{1}\left(1,-\frac{q_{1} q_{2}}{4}\right)+_{0} \mathbf{F}_{1}\left(2,-\frac{q_{1} q_{2}}{4}\right)\right]
\label{eg5}
\end{equation}
then using the Einstein summation and denoting the matrix element $\Gamma_{m n}$ as:
\begin{equation}
\Gamma_{m n} \equiv \frac{1}{2} \iint_{0}^{\infty} d q_{1} d q_{2} e^{-\frac{q_{1}+q_{2}}{2}} q_{1} q_{2} L_{m}\left(q_{1}\right) L_{n}\left(q_{2}\right)\times\left[ _{0} \mathbf{F}_{1}\left(1,-\frac{q_{1} q_{2}}{4}\right)+_{0} \mathbf{F}_{1}\left(2,-\frac{q_{1} q_{2}}{4}\right)\right]
\label{Gamma}
\end{equation}

To transform the Hypergeometric functions to Laguerre Polynomials in $\Gamma^{\text{2bdy}}_{m n}$  we can use the Hardy-Hille Formula (Lemma.\ref{Hardy-Hille}). Then we get back to Eq.\ref{Gamma}:
\begin{equation}
\begin{aligned}
\Gamma_{m n} =& \frac{1}{2} \iint_{0}^{\infty} d q_{1} d q_{2} e^{-\frac{q_{1}+q_{2}}{2}} q_{1} q_{2} L_{m}\left(q_{1}\right) L_{n}\left(q_{2}\right)\times\left(_{0} \mathbf{F}_{1}\left(1,-\frac{q_{1} q_{2}}{4}\right)+_{0} \mathbf{F}_{1}\left(2,-\frac{q_{1} q_{2}}{4}\right)\right)\\
=& \iint_{0}^{\infty} d q_{1} d q_{2} e^{-(q_{1}+q_{2})} q_{1} q_{2} L^{(0)}_{m}\left(q_{1}\right) L^{(0)}_{n}\left(q_{2}\right)\times\frac{1}{2} e^{\frac{q_{1}+q_{2}}{2}} \left(_{0} \mathbf{F}_{1}\left(1,-\frac{q_{1} q_{2}}{4}\right)+_{0} \mathbf{F}_{1}\left(2,-\frac{q_{1} q_{2}}{4}\right)\right)\\
=& \iint_{0}^{\infty} d q_{1} d q_{2} e^{-(q_{1}+q_{2})} q_{1} q_{2} L^{(0)}_{m}\left(q_{1}\right) L^{(0)}_{n}\left(q_{2}\right) \times\sum_{k=0}^{\infty} (-1)^{k} \left(L_{k}^{(0)}(q_{1}) L_{k}^{(0)}(q_{2})+\frac{2}{k+1} L_{k}^{(1)}(q_{1}) L_{k}^{(1)}(q_{2}) \right)\\
=& \sum_{k=0}^{\infty} (-1)^{k} \left[\int_{0}^{\infty} d q_{1} e^{-q_{1}} q_{1} L^{(0)}_{m}\left(q_{1}\right) L_{k}^{(0)}(q_{1})  \times \int_{0}^{\infty} d q_{2} e^{-q_{2}} q_{2} L^{(0)}_{n}\left(q_{2}\right) L_{k}^{(0)}(q_{2} \right.\\
& \quad \left. +\frac{2}{k+1}\int_{0}^{\infty} d q_{1} e^{-q_{1}} q_{1} L^{(0)}_{m}\left(q_{1}\right) L_{k}^{(1)}(q_{1}) \times \int_{0}^{\infty} d q_{2} e^{-q_{2}} q_{2} L^{(0)}_{n}\left(q_{2}\right) L_{k}^{(1)}(q_{2}) \right]
\label{Gamma2}
\end{aligned}
\end{equation}
As we can see, $q_1$ and $q_2$ are symmetric in Eq.\ref{Gamma2} and the integrals for $q_1$ and $q_2$ can be separated from each other. So we can focus on one of them for now:
\begin{equation}
\begin{aligned}
& \int_{0}^{\infty} d q_{1} e^{-q_{1}} q_{1} L^{(0)}_{m}\left(q_{1}\right) L_{k}^{(0)}(q_{1})\\
=&\int_{0}^{\infty} d q_{1} e^{-q_{1}} q_{1} \left[L^{(1)}_{m}(q_{1})-L^{(1)}_{m-1}(q_{1}) \right]  \left[ L_{k}^{(1)}(q_{1})-L_{k-1}^{(1)}(q_{1})  \right]\\
=&\int_{0}^{\infty} d q_{1} e^{-q_{1}} q_{1} \left[ L^{(1)}_{m}(q_{1}) L_{k}^{(1)}(q_{1})-L^{(1)}_{m}(q_{1}) L_{k-1}^{(1)}(q_{1}) - L^{(1)}_{m-1}(q_{1})L_{k}^{(1)}(q_{1})+ L^{(1)}_{m-1}(q_{1})L_{k-1}^{(1)}(q_{1}) \right]\\
=&\frac{(m+1)!}{m!} \delta_{m,k} - \frac{(m+1)!}{m!} \delta_{m,k-1} -\frac{m!}{(m-1)!} \delta_{m-1,k} + \frac{m!}{(m-1)!} \delta_{m-1,k-1} \\
=&(m+1) \delta_{m,k} - (m+1) \delta_{m,k-1} -m \delta_{m-1,k} + m \delta_{m-1,k-1} 
\label{p1int}
\end{aligned}
\end{equation}
where we have considered the orthogonality of the generalized Laguerre polynomials, the recurrence relations of the generalized Laguerre polynomials. Also the second integrals could be written as:
\begin{equation}
\begin{aligned}
&\int_{0}^{\infty} d q_{1} e^{-q_{1}} q_{1} L^{(0)}_{m}\left(q_{1}\right) L_{k}^{(1)}(q_{1}) =\int_{0}^{\infty} d q_{1} e^{-q_{1}} q_{1} [L^{(1)}_{m}(q_{1})-L^{(1)}_{m-1}(q_{1})]L_{k}^{(1)}(q_{1})\\
=&\int_{0}^{\infty} d q_{1} e^{-q_{1}} q_{1}[L^{(1)}_{m}(q_{1})L_{k}^{(1)}(q_{1})-L^{(1)}_{m-1}(q_{1})L_{k}^{(1)}(q_{1})]=\frac{(m+1)!}{m!} \delta_{m,k} -\frac{m!}{(m-1)!} \delta_{m-1,k} \\
=&(m+1) \delta_{m,k} -m \delta_{m-1,k}
\label{p2int}
\end{aligned}
\end{equation}
The reason why we have not simplified the Kronecker Deltas above is that there is a coefficient  $(-1)^{k}$ that could influence the result significantly. Thus,

\begin{small}
\begin{equation}
\begin{aligned}
\Gamma_{m n} &= \sum_{k=0}^{\infty} (-1)^{k} \{[(m+1) \delta_{m,k} - (m+1) \delta_{m,k-1} -m \delta_{m-1,k} + m \delta_{m-1,k-1} ]\\
&\ \ \ \times[(n+1) \delta_{n,k} - (n+1) \delta_{n,k-1} -n \delta_{n-1,k} + n \delta_{n-1,k-1} ]+\frac{2}{k+1}[(m+1) \delta_{m,k} -m \delta_{m-1,k}][(n+1) \delta_{n,k} -n \delta_{n-1,k}]\}\\
&= \sum_{k=0}^{\infty} (-1)^{k}\{[(m+1)(n+1) \delta_{m,k} \delta_{n,k}+(m+1)n \delta_{m,k}\delta_{n-1,k-1}+(m+1) (n+1)\delta_{m,k-1} \delta_{n,k-1}\\
&\ \ \ +(m+1)n\delta_{m,k-1} \delta_{n-1,k}+m(n+1)\delta_{m-1,k}\delta_{n,k-1}+mn\delta_{m-1,k}\delta_{n-1,k} +m(n+1)\delta_{m-1,k-1} \delta_{n,k}\\
&\ \ \ +mn\delta_{m-1,k-1}\delta_{n-1,k-1}]+\frac{2}{k+1}[(m+1)(n+1) \delta_{m,k} \delta_{n,k}+mn\delta_{m-1,k} \delta_{n-1,k}]\}
\label{Gamma3}
\end{aligned}
\end{equation}
\end{small}
Then we can combine the Kronecker Deltas and get:
\begin{equation}
\begin{aligned}
\Gamma_{m n} &=(-1)^{m}\{[(m+1)^2 \delta_{m,n}+m(m+1) \delta_{m,n}-(m+1)^2 \delta_{m,n}-(m+1)(m+2)\delta_{m,n-2}\\
&\ \ \ -m(m-1) \delta_{m,n+2}-m^2 \delta_{m,n}+m(m+1)\delta_{m,n}+m^2\delta_{m,n}] +2[(m+1)\delta_{m,n}-m\delta_{m,n}]\}\\
&=(-1)^{m}\{[(m+1)^2+m(m+1)-(m+1)^2-m^2+m(m+1)+m^2+2(m+1)-2m]\delta_{m,n}\\
&\ \ \ -(m+1)(m+2)\delta_{m,n-2}-m(m-1) \delta_{m,n+2}\}\\
&=(-1)^{m}[2(m^2+m+1)\delta_{m,n}-(m+1)(m+2)\delta_{m,n-2}-m(m-1) \delta_{m,n+2}]=\Gamma^{\text{2bdy}}_{m n} \qedhere
\label{Gamma4}
\end{aligned}
\end{equation}
\end{proof}

\section{Anti-symmetric FQH three-body wavefunctions}

\begin{table}[]
\centering
\renewcommand{\arraystretch}{2.5}
\begin{tabular}{||c|c|c||}
\hline
$(k,l)$ & $\alpha = 2k + 3 l$ &  $|\Psi_{kl} \rangle= \sum_{n_1, n_2} \alpha^{n_1, n_2} | n_1 ,n_2\rangle $\\ \hline 
(0,1) & 3 &  $\frac{1}{2}|3,0\rangle-\frac{\sqrt{3}}{2}|1,2\rangle$\\ \hline
(1,1) & 4 &  $-\frac{\sqrt{5}}{4}|5,0\rangle+\frac{1}{2 \sqrt{2}}|3,2\rangle+\frac{3}{4}|1,4\rangle$\\ \hline
(0,2) & 6 &  $\frac{\sqrt{3}}{4}|5,1\rangle-\frac{1}{2} \sqrt{\frac{5}{2}}|3,3\rangle+\frac{\sqrt{3}}{4}|1,5\rangle$\\ \hline
(2,1)& 7 &  $-\frac{\sqrt{21}}{8}|7,0\rangle+\frac{1}{8}|5,2\rangle+\frac{\sqrt{15}}{8}|3,4\rangle+\frac{3 \sqrt{3}}{8}|1,6\rangle$\\ \hline
(1,2)& 8 &  $\frac{3}{4 \sqrt{2}}|7,1\rangle-\frac{\sqrt{7}}{4 \sqrt{2}}|5,3\rangle-\frac{\sqrt{7}}{4 \sqrt{2}}|3,5\rangle+\frac{3}{4 \sqrt{2}}|1,7\rangle$\\ \hline
(0,3)& 9 &  $\frac{1}{16}|9,0\rangle-\frac{3}{8}|7,2\rangle+\frac{3 \sqrt{7}}{8 \sqrt{2}}|5,4\rangle-\frac{\sqrt{21}}{8}|3,6\rangle+\frac{3}{16}|1,8\rangle$\\ \hline
(3,1)& 9 &  $-\frac{\sqrt{21}}{8}|9,0\rangle+\frac{\sqrt{3}}{4 \sqrt{2}}|5,4\rangle+\frac{1}{2}|3,6\rangle+\frac{\sqrt{21}}{8}|1,8\rangle$\\ \hline
\end{tabular}%
\caption{\textbf{The anti-symmetric three-body wavefunctions expanded in the basis of $| n_1, n_2 \rangle$.} The values of $\alpha^{n_1, n_2}$ can be found easily by looking at the coefficient of the corresponding basis \cite{yang2018three}. For example, $\alpha^{5, 0}$ is the coefficient of $|5,1\rangle$, i.e. $-\frac{\sqrt{5}}{4}$.}
\label{tabexpansion}
\end{table}

The complete orthonormal basis states of t anti-symmetric FQH three-body wavefunctions provided in Ref.\cite{laughlin1983quantized} are written as:
\begin{equation}
|\Psi_{kl} \rangle =\frac{1}{\left[2^{6 l+4 k+1}(3 l+k) ! k ! \pi^{2}\right]^{1 / 2}}\left[\frac{\left(z_{a}+i z_{b}\right)^{3 l}-\left(z_{a}-i z_{b}\right)^{3 l}}{2 i}\right]\left(z_{a}^{2}+z_{b}^{2}\right)^{k} e^{-(1 / 4)\left(\left|z_{a}\right|^{2}+\left|z_{b}\right|^{2}\right)}
\end{equation}
the expansion of which in the $n_1, n_2$ basis can be found in Table.\ref{tabexpansion}.


\end{document}